\documentclass[11pt]{article} % Spécification de la classe du document (article) et de la taille de police (11pt)
\usepackage{amsmath} % Charger le package amsmath pour les mathématiques avancées
\usepackage{amssymb} % Charger le package amssymb pour les symboles mathématiques
\usepackage{epsf}    % Charger le package epsf pour l'inclusion d'images EPS
\usepackage{epic}    % Charger le package epic pour des fonctionnalités graphiques avancées
\usepackage{cite}    % Charger le package cite pour des citations bibliographiques
\usepackage{xcolor}  % Charger le package xcolor pour la gestion des couleurs
\usepackage{tikz}  %Package for graph pictures

\usepackage{arydshln} % For dashed lines

\usetikzlibrary{positioning, arrows.meta}

\usepackage[utf8]{inputenc} % Charger le package inputenc pour les caractères accentués

\numberwithin{equation}{section}
\newcommand{\un}{{\mathbb I}}
\newcommand{\ra}{\rightarrow}
\newcommand{\tr}{{\rm tr}}
\newcommand{\spa}{{\rm span}}
\newcommand{\ran}{\rm Ran\, }
\renewcommand{\ker}{{\rm Ker\, }}
\newcommand{\bra}{\langle}
\newcommand{\ket}{\rangle}
\newcommand{\be}{\begin{equation}}
\newcommand{\ee}{\end{equation}}
\newcommand{\bea}{\begin{eqnarray}}
\newcommand{\eea}{\end{eqnarray}}
\newcommand{\eps}{\varepsilon}

\newcommand{\ch}{{\cal H}}

\newcommand{\e}{{\rm e}}

\newcommand{\ffi}{\varphi}
\newcommand{\sign}{\mbox{sign}}

\newcommand{\ode}{{\cal O}}
\newcommand{\grintl}{[\kern-.18em [}
\newcommand{\grintr}{]\kern-.18em ]}

\newcounter{resultcounter}[section]

\newtheorem{thm}[resultcounter]{Theorem}
\newtheorem{lem}[resultcounter]{Lemma}
\newtheorem{prop}[resultcounter]{Proposition}
\newtheorem{cor}[resultcounter]{Corollary}
\newtheorem{definition}[resultcounter]{Definition}
\newtheorem{rem}[resultcounter]{Remark}
\def\bed{\begin{definition}}
\def\eed{\end{definition}}

\def\proof{\noindent{\bf Proof}\  }
 \def\cB{{\cal B}} 
\def\cD{{\cal D}}  
\def\cG{{\cal G}} \def\cH{{\cal H}} 
 \def\cK{{\cal K}} 
\def\cM{{\cal M}}  
  \def\cR{{\cal R}}
\def\cS{{\cal S}} \def\cT{{\cal T}} \def\cU{{\cal U}}
\def\cV{{\cal V}}  
 
\newcommand{\R}{{\mathbb R}}
\newcommand{\N}{{\mathbb N}}
\newcommand{\Q}{{\mathbb Q}}
\newcommand{\C}{{\mathbb C}}
\newcommand{\Z}{{\mathbb Z}}

\renewcommand{\P}{{\mathbb P}}
\def\qed{\hfill $\Box$\medskip}
\newcommand{\vertiii}[1]{{\left\vert\kern-0.25ex\left\vert\kern-0.25ex\left\vert #1 
    \right\vert\kern-0.25ex\right\vert\kern-0.25ex\right\vert}}

\newcommand{\diag}{{\rm Diag}}

\newcommand{\rank}{{\rm Rank\,}}

\begin{document}

\title{Unitary and Open Scattering Quantum Walks on Graphs}

\author{Alain Joye\footnote{ Univ. Grenoble Alpes, CNRS, Institut Fourier, F-38000 Grenoble, France}}

\date{ }

\maketitle

\abstract{

We consider a class of Unitary Quantum Walks on arbitrary graphs, parameterized by a family of scattering matrices. These Scattering Quantum Walks model the discrete dynamics of a system on the graph's edges, where a scattering process at the vertices is governed by the scattering matrices assigned to each vertex. 
We show the Scattering Quantum Walks encompass several known Quantum Walks.
We further introduce two classes of Scattering Open Quantum Walks on arbitrary graphs, also parameterized by scattering matrices: one defined on the edges, the other on the vertices of the graph. We show these walks give rise to proper quantum channels and describe their main spectral and dynamical properties, relating them to naturally associated classical Markov chains.

\section{Introduction}

Quantum Walks (QWs) defined on graphs or lattices are popular discrete time linear dynamical systems which are used in  different scientific fields under different guises and for various applications, see {\it e.g.} the books and reviews \cite{Ke, V-A, P, J4, ABJ3, GZ, QMS}, and references therein. 
For example, in quantum optics or condensed matter physics, QWs provide a versatile tool to approximate the complex dynamics of certain systems in some physically relevant regimes, \cite{KFC+ , ZKG+, CC, WM, TMT}. QWs are  used in the field of quantum computing and information processing, notably as quantum search algorithms, see  \cite{C, Sa, P , KMOR}. From a probabilistic perspective, QWs are considered as non commutative extensions of classical random walks and Markov chains, see {\it e.g.} \cite{Sz, Kon, G, APSS1}.

The success and diversity of the different types of QWs have sparked the interest of the mathematical community which has rigorously addressed several of their properties, such as the spectral and dynamical properties of certain QWs, \cite{Ko1, BHJ, HKSS2, MS-B, JMa, HS, AG-PS, HSS, RST1, RST2, T, RT, CJWW}, Anderson (de-)localization in various models of random and quasiperiodic QWs \cite{Ko, J2, HJS1, HJS2, J3, JM, ASW, ABJ1, ABJ2, HJ, BM, CFGW, CFO, CFL+}, many-body systems of QWs \cite{AAM+, HJ2, R, AJR},  or the topological properties of QWs \cite{ABJ4, DFT, CGG+, SS-B, ST, D, ABJ5}, to mention a few.

In this paper, we introduce and study the notion of Scattering Quantum Walks (SQWs) on arbitrary graphs, a class of QWs broad enough to encompass many QWs present in the literature, and which moreover admits both unitary and open versions. Unitary SQWs are used to address closed quantum systems, while  open SQWs are relevant in the analysis of open quantum systems and are known as quantum channels.

Let us describe informally the definition of  a unitary SQW on a graph $G$. It is a unitary operator on a Hilbert space associated with $G$, which models the linear dynamics of a quantum system, or quantum walker, over a single time step. For a simple graph $G  = (V, E)$ with vertex set $V$ and undirected edge set $E$, { we consider that each edge gives rise to two directed edges with opposite directions, and denote by $D$ the set of directed edges.} We attach to each directed edge a canonical basis vector of a Hilbert space $l^2(D)$. The SQW on $l^2(D)$ is parameterized by a family of scattering matrices, each assigned to a vertex in $V$. Informally, we say that the quantum walker lives on the directed edges of the graph, as  its quantum state is represented by a normalized vector of  $l^2(D)$.

At each time step, the state of the quantum walker on a given edge with a specific direction undergoes a scattering process at the vertex towards which the edge points, say $x\in V$. This process is governed by the scattering matrix $S(x)$ assigned to $x$. The walker's state is then transformed into the corresponding linear combination of states on adjacent edges, with directions now pointing away from the vertex $x$ where scattering took place, see Figure \ref{fig:scatpro}. 

\begin{figure}[h]

\hspace{2.5cm}\begin{tikzpicture}
    % Define vertices with smaller circles
    \node[draw, circle, minimum size=0.8cm] (v1) at (-1, 1) {};
    \node[draw, circle, minimum size=0.8cm] (v2) at (0, 0) {$x$}; % Distinguished vertex
    \node[draw, circle, minimum size=0.8cm] (v3) at (1, 1) {};
    \node[draw, circle, minimum size=0.8cm] (v4) at (1, -1) {};
    \node[draw, circle, minimum size=0.8cm] (v5) at (-1, -1) {};
    
    % Define edges (smaller)
    \draw[->, very thick] (v1) -- (v2); % Arrow pointing to the distinguished vertex v2
    %\draw[-] (v2) -- (v3);
    \draw[-] (v2) -- node[below] {\phantom{xxxx}$S(x)$} (v3); % Edge between v2 and v3 with label S(x)
    \draw[-] (v2) -- (v4);
    \draw[-] (v2) -- (v5);
    
    % Hanging peripheral edges
    \draw[-] (v1) -- ++(-0.8, 0.8); % Existing hanging edge for v1
    \draw[-] (v1) -- ++(-1.131, 0);   % New horizontal hanging edge for v1
    \draw[-] (v3) -- ++(0.8, 0.8);  % Existing hanging edge for v3
     \draw[-] (v3) -- ++(1.131, 0);    % New horizontal hanging edge 
    \draw[-] (v4) -- ++(0.8, -0.8); % Existing hanging edge for v4
    \draw[-] (v4) -- ++(1.131, 0);    % New horizontal hanging edge for v4
     \draw[-] (v4) -- ++(0, -1.131);    % New vertical hanging edge for v4
    \draw[-] (v5) -- ++(-0.8, -0.8);% Existing hanging edge for v5

%\end{tikzpicture}

\hspace{2.5cm}{\huge$\leadsto$}

\hspace{2.5cm}

%\begin{tikzpicture}
    % Define vertices with smaller circles
    \node[draw, circle, minimum size=0.8cm] (v1) at (-1, 1) {};
    \node[draw, circle, minimum size=0.8cm] (v2) at (0, 0) {$x$}; % Distinguished vertex
    \node[draw, circle, minimum size=0.8cm] (v3) at (1, 1) {};
    \node[draw, circle, minimum size=0.8cm] (v4) at (1, -1) {};
    \node[draw, circle, minimum size=0.8cm] (v5) at (-1, -1) {};
    
    % Define edges (smaller) with colored arrows
    \draw[->, very thick] (v2) -- (v1); % Blue arrow pointing away from the distinguished vertex v2
    %\draw[->, thick] (v2) -- (v3); % Blue arrow
     \draw[->, very thick] (v2) -- node[below] {\phantom{xxxx}$S(x)$} (v3);
    \draw[->, very thick] (v2) -- (v4); % Blue arrow
    \draw[->, very thick] (v2) -- (v5); % Blue arrow
    
    % Hanging peripheral edges
     \draw[-] (v1) -- ++(-0.8, 0.8); % Existing hanging edge for v1
    \draw[-] (v1) -- ++(-1.131, 0);   % New horizontal hanging edge for v1
    \draw[-] (v3) -- ++(0.8, 0.8);  % Existing hanging edge for v3
     \draw[-] (v3) -- ++(1.131, 0);    % New horizontal hanging edge 
    \draw[-] (v4) -- ++(0.8, -0.8); % Existing hanging edge for v4
    \draw[-] (v4) -- ++(1.131, 0);    % New horizontal hanging edge for v4
     \draw[-] (v4) -- ++(0, -1.131);    % New vertical hanging edge for v4
    \draw[-] (v5) -- ++(-0.8, -0.8);% Existing hanging edge for v5

\end{tikzpicture}
\caption{\small The scattering process at work in SQWs}
\label{fig:scatpro}

\end{figure}
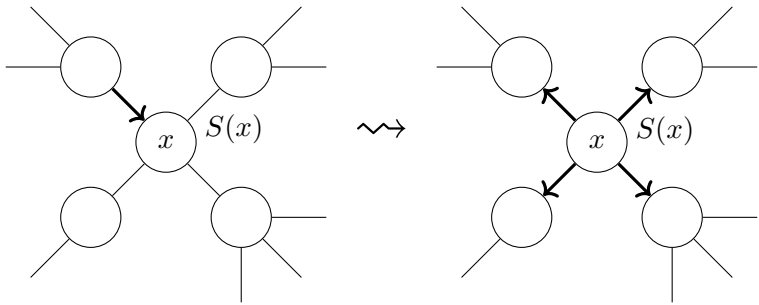

The freedom in choosing the directed graph and the unitary scattering matrices at its vertices explains why SQWs encompass a significant range of unitary QWs considered in the literature. We explicitly show that Coined QWs, the Chalker-Coddington model, and a slight generalization of the Grover Walk belong to this class. Following \cite{HKSS2, HSS} we revisit the spectral properties of the latter by relating them to those of a self-adjoint operator defined on $l^2(V)$, the Hilbert space associated  the graph's vertices, through natural boundary operators. We also study SQWs on the star-graph. QWs constructed in a similar way for specific cases have appeared in the physics litterature, see {\it e.g.} \cite{FH1, FH2}.

Let us describe informally the construction of an open SQW on the edges of $G$ now. An open SQW acts on density matrices on $l^2(D)$ representing the states of a quantum walker. We start with an initial state, and measure the position of the quantum walker. This quantum measurement yields a new state depending on the outcome, with a certain quantum mechanical probability. We then evolve this new state using the unitary SQW, resulting in another state. The final state is obtained by taking the expectation value of this state with respect to the quantum mechanical probability distribution of the position observable, which induces decoherence. The open SQW is  the map from the initial state to the final state obtained this way. 
Finally, passing from $l^2(D)$ associated with the edges of $G$ to $l^2(V)$ associated with its vertices by means of  boundary operators, which are instrumental in the analysis of the Grover Walk, {and after appropriate normalization,} we define a second open SQW on the vertices of $G$, we call induced by the open SQW on its edges.

Our analysis of the dynamical properties of both these open SQWs shows that they are related to those of natural Markov processes defined on the graph, as in \cite{G, APSS1}, which depend on the choice of the scattering matrices. A few examples are worked out to illustrate the richness of these constructions.

\medskip

The paper is organized as follows. Section \ref{purestate} is devoted to the definition of unitary SQWs on the edges of an arbitrary directed graph, their quantum mechanical interpretation, and their general properties. Particular cases such as the Chalker-Coddington model, Coined Quantum Walks, the star-graph case and the Generalized Grover Walk are discussed in Section \ref{sec:particular}. The notions of open SQWs on the edges of a graph and their quantum mechanical interpretation are introduced in Section \ref{SOQWG}, while their spectral and dynamical properties are addressed in Section \ref{sec:sdp}. Finally, Section \ref{sec:IOQW} presents the definition and analysis of induced open SQWs on the vertices of the graph. Both notions of open SQWs are illustrated by examples. An appendix containing technical proofs closes the paper.

\section{Unitary Scattering Quantum Walks on a Graph}\label{purestate}

Let us set the notation. 
A graph $G=(V,E)$ is specified by the set of its vertices  $V$ and the set of its {undirected} edges $E$. The order of the graph $G$ is $|V|$, the cardinal of $V$, which may be infinite. In the latter case,  $V$ is assumed to be countable and $G$ is called infinite. For two vertices $x\in V, y\in V$ forming an edge, we will write $x\sim y $ the fact that {$[xy]\in E$}. The graphs we consider have no loop, {\it i.e.} {$[xx]\not \in E$} and to avoid trivialities, we will always assume that $|V|\geq 2$ and that $G$ is connected. Each edge gives rise to two directed edges of opposite direction. A directed edge from vertex $x$ to vertex $y$ is denoted by $(yx)$, and the set of directed edges is $D$. Finally, the degree of a vertex $x\in V$, { is $d_x=|\{y\in V \ {\rm s.t.} \ y\sim x\}|\in \N^*$},  and we have $\sum_{x\in V}d_x=2|E|=|D|$. If $G$ is infinite, we further assume that $\sup_{x\in V}d_x=\bar d<\infty$. Furthermore, we pick a vertex $r\in V$ we call the root and for any $x\in V$, we denote by $|x|$ the graph distance between $x$ and $r$. For any $n\in \N$, the set $B_n=\{x\in V \ \text{s.t.} \ |x|\leq n\}$ has finite cardinal bounded above by $c(\bar d-1)^{n-1}$ for $n$ large, $\bar d>2$, and $c$ a $\bar{d}-$dependent constant.
\\

We associate to $G$ the Hilbert space $l^2(D)$ as follows. To each directed edge $(yx)$ from vertex $x$ to vertex $y$, we associate a canonical basis vector of $l^2(D)$, denoted by $|yx\rangle$, so that, for a finite graph, $l^2(D)=\mbox{\rm span}\{|yx\rangle\}_{x, y \in V\atop x\sim y}$, with $\dim l^2(D)=\sum_{x\in V}d_x<\infty$. In case $G$ is infinite with countably many vertices, we set
\be
 l^2(D)=\bigg\{\psi =\sum_{x \in V}\sum_{y \sim x}\psi_{xy}|xy\rangle, \, \psi_{xy}\in \C  \ \Big| \  \sum_{x \in V}\sum_{ y\sim x}|\psi_{xy}|^2<\infty\bigg\}.
 \ee
For later purposes, we note the direct sum decompositions 
\begin{align}\label{dirsumspa}
&l^2(D)=\oplus_{x\in V}\cH^{\rm I}_x=\oplus_{x\in V}\cH^{\rm O}_x, \ \ \mbox{where} \nonumber \\ 
&\cH^{\rm I}_x=\mbox{\rm span}\{|xy\rangle\}_{y\sim x}, \ \ \cH^{\rm O}_x=\mbox{\rm span}\{|yx\rangle\}_{y\sim x}.
\end{align}
The $d_x$-dimensional subspaces  $\cH^{\rm I}_x$, resp. $\cH^{\rm O}_x$, are the incoming, resp. outgoing, subspaces attached to the vertex $x\in V$, spanned by incoming, resp. outgoing, edges to $x$.  We will need below the corresponding orthogonal projectors $P^{\rm I}_x$, resp. $P^{\rm O}_x$ onto these subspaces. With the notation $|\ffi\ket\bra\psi |\in \cB(l^2(D))$ to denote, for any pair of vectors $\ffi, \psi\in l^2(D)$, the rank one linear operator such that $|\ffi\ket\bra\psi | \chi=\ffi \bra\psi|\chi\ket$, $\forall \chi\in l^2(D)$, these projectors read and satisfy
\begin{align}\label{projio}
&P^{\rm I}_x=\sum_{y\sim x}|xy\ket\bra xy|, \ \ P^{\rm O}_x=\sum_{y\sim x}|yx\ket\bra yx| \ \ \mbox{s.t.} \ \ P_x^{\#}P_y^{\#}=\delta_{x y}P_x^{\#}, \ \ \#\in\{{\rm I, O}\}, \\  \nonumber
& P_x^{\rm I}P_y^{\rm O}=P_y^{\rm O}P_x^{\rm I}=\begin{cases}
    |xy\rangle\langle xy| & \text{if } x \sim y \\
    0 & \text{otherwise}.
\end{cases}
\end{align}

We proceed with the definition of QWs defined on $l^2(D)$ as unitary operators parameterized by a family of scattering matrices. To each vertex
 $x\in V$ of degree $d_x$, we associate a unitary matrix called a scattering matrix
\begin{align}
&  S(x)\in U(d_x), \ \mbox{with matrix elements} \ S(x)=(S_{zy}(x))_{y\sim x \atop z\sim x} .
\end{align}
The labelling of the matrix elements in the canonical basis of $\C^{d_x}$ is such that we can associate vectors of $\cH^{\rm I}_x$ and $\cH^{\rm O}_x$ to $S_{zy}(x)$ in the following definition.
\begin{definition}
Given s set of scattering matrices $\cS=\{S(x)\}_{x\in V}$, the QW operator on the graph $G$, $U_\cS$, is defined on $l^2(D)$ by its action on the basis vectors 
\be\label{defscatu}
U_\cS |xy\rangle = \sum_{z\sim x} S_{zy}(x) |zx\rangle, \ \ \text{for all} \ \ x, y\in V, x\sim y,
\ee
or, equivalently,
\be\label{defug}
U_\cS = \sum_{x\in V}\sum_{y\sim x\atop z\sim x} S_{zy}(x) |zx\ket \bra xy|.
\ee
\end{definition}
The set of scattering matrices $\cS$ parameterizing the QW operator $U_\cS$ is  emphasized in the notation. We call the QW defined by (\ref{defug}) a Scattering Quantum Walk, SQW for short. 
We note that for all $x\in V$, $U_\cS$ intertwines between the projectors $P_x^{\rm I}$ and $P_x^{\rm O}$
\begin{align}\label{oui}
&U_\cS P_x^{\rm I}=P_x^{\rm O}U_\cS=P_x^{\rm O}U_\cS P_x^{\rm I}= \sum_{y\sim x\atop z\sim x} S_{zy}(x) |zx\ket \bra xy|.
\end{align}
\begin{rem}\label{remconvu} If $V$ is an infinite graph, the sum  in (\ref{defug}) is to be understood as \\
$\lim_{n\ra \infty}\sum_{x\in V \atop |x|\leq n}(\dots)$  in the strong sense. The convergence is ensured by the unitarity of $S(x)$:  for any $x\in V$, any $\psi \in l^2(D)$,
\begin{align}
\big\|\sum_{y\sim x\atop z\sim x} S_{zy}(x) |zx\ket \bra xy|\psi\big\|^2&=\sum_{y\sim x\atop z\sim x}\sum_{y'\sim x\atop z'\sim x} \overline{S_{z'y'}}(x)S_{zy}(x)\bra \psi  |xy'\ket \bra z'x| |zx\ket \bra xy|\psi\ket \nonumber\\
&= \sum_{y\sim x\atop y'\sim x} \Big(\sum_{ z\sim x}\overline{S_{z y'}}(x)S_{zy}(x)\Big)\bra \psi  |xy'\ket\bra xy|\psi\ket \nonumber\\
&= \sum_{y\sim x\atop y'\sim x} \delta_{y y'}|\bra \psi  |xy\ket|^2= \|P_x^{\rm I}\psi \|^2,
\end{align}
where the orthogonal projectors $\{P_x^{\#}\}_{x\in V}$, $\#\in \{\rm I, O\}$ form  resolutions of the identity.
\end{rem}

\begin{lem} \label{unitUS}
For a graph $G=(V,E)$ and a set $\cS=\{S(x)\}_{x\in V}$, the operator $U_\cS$ acting on the Hilbert space $l^2(D)$ defined by (\ref{defug}) is a unitary operator. 
\end{lem}
\proof:
The property follows from the unitarity of $S(x)$:
\begin{align}
U_\cS^*U_\cS&=  \sum_{x'\in V}\sum_{y'\sim x'\atop z'\sim x'} \overline{S_{z'y'}(x')} |x'y'\ket \bra z'x'|\sum_{x\in V}\sum_{y\sim x\atop z\sim x} S_{zy}(x) |zx\ket \bra xy|\nonumber\\
=&\sum_{x\in V}\sum_{y\sim x\atop y'\sim x' } \bigg( \sum_{z\sim x } \overline{S_{zy'}(x)}S_{zy}(x)\bigg) |xy'\ket \bra xy| \nonumber\\
&= \sum_{x\in V}\sum_{y\sim x \atop y'\sim x} \delta_{y y'} |xy'\ket \bra xy| = \sum_{x\in V} P_x^{\rm I}= \un,
\end{align}
and the reverse identity $U_\cS U_\cS^*=\un$ is proven similarly. 
\qed

The  relations (\ref{oui}) show that $U_\cS$ couples subspaces $\cH_x^{\rm I}$ attached at neighbouring vertices of the graph.

\medskip
{\bf Quantum mechanical interpretation:} 
The quantum system at hand, or quantum walker, has configuration space given by the directed edges $(xy)$ of the graph $G$, giving rise to the canonical basis of the Hilbert space $l^2(D)$. In a  state described by a normalized vector $\psi\in l^2(D)$, the probability of the quantum walker to be on the directed edge $(xy)$ of the graph $G$ is $|\bra xy |\psi\ket |^2=|\psi_{xy}|^2$.  The operator $U_\cS$ defines the one time step evolution of the quantum system. By construction, the state $|xy\ket$ of a quantum walker incoming at the vertex $x\in V$ along the edge $(xy)$ undergoes a local scattering process monitored by the unitary matrix $S(x)$ which sends it to a linear combination of outgoing states along the edges $(zx)$, according to (\ref{defscatu}). This local scattering point of view on the quantum dynamics is at work in several physically motivated specific models that, as we will show, are special cases of our definition (\ref{defug}). 

\medskip 
The questions of interest concern the determination of the spectrum of $U_\cS$ and of the related properties of the discrete time dynamical system  $(U_\cS^n)_{n\in \Z}$ on $l^2(D)$, as a function of the characteristics of $G$, and of the set $\cS=\{S(x)\}_{x\in V}$ that parametrize the quantum evolution. The behaviour in time of the probability distribution $\P_n^{\psi_0}(\cdot )$ on the set of directed edges $\{(xy)\}_{x\sim y}$ induced by the quantum dynamics is also of  interest: let $\psi_0\in l^2(D)$ be an initial state and $\psi_n=U_\cS^n\psi_0\in l^2(D)$ be the corresponding state at time $n\in \Z$. For $(xy)$ a directed edge of $G$,  $\P_n^{\psi_0}(\cdot )$ defined by
\be
\P_n^{\psi_0}(xy)=\big|\bra xy|U_\cS^n\psi_0\ket\big|^2=\big|\bra xy|\psi_n\ket\big|^2
\ee
 yields the probability to find the quantum walker at time $n$ on the directed edge $(xy)$ by a measurement of its position. 
Another important time dependent distribution on the vertices $x \in V$ of the graph, $\Q_n^{\psi_0}(\cdot )$, is the probability to find the quantum walker at time $n\in \Z$ in the subspace $\cH_x^{\rm I}$ or, improperly, on the vertex $x\in V$, by a measurement of its position. It is defined by
\be\label{probavertex}
\Q_n^{\psi_0}(x)=\big\|P_x^{\rm I}U_\cS^n\psi_0\big\|^2=\big\|P_x^{\rm I}\psi_n\big\|^2=\sum_{y\sim x}\P_n^{\psi_0}(xy).
\ee

\subsection{General Properties}
We proceed with a few general considerations before discussing some special cases.

\medskip

{\bf Perturbation theory:}
To compare two SQWs defined on the same graph with different sets of scattering matrices, we can resort to the following result which holds in infinite dimension. 
Such estimates have proven to be instrumental in the analyses of currents  in the Chalker-Coddington model, \cite{ABJ4, ABJ5} , and other network models.
\begin{lem} 
For a graph $G=(V,E)$ and two sets of scattering matrices $\cS=\{S(x)\}_{x\in V}$, resp. $\cS'=\{S'(x)\}_{x\in V}$, let $U_\cS$, resp. $U_{\cS'}$, be defined on $l^2(D)$ according to (\ref{defug}). Then,
\be
\| U_\cS -U_{\cS'} \| \leq  \sup_{x\in V}\|S(x)-S'(x)\|_{\rm HS},
\ee
where $\|\cdot \|_{\rm HS}$ denotes the Hilbert-Schmidt norm on finite matrices.
\end{lem}
\proof:
For any $\psi=\sum_{x \in V}\sum_{y \sim x}\psi_{xy}|xy\ket \in l^2(D)$, 
\begin{align}
 \|(U_\cS -U_{\cS'})\psi\|^2&=\sum_{x\in V}\Big\|P_x^{\rm O}\sum_{y\sim x \atop z\sim x} (S(x)-S'(x))_{zy}|zx\ket \bra xy| \psi\Big\|^2\nonumber \\
 &= \sum_{x\in V} \sum_{z\sim x}\Big|\sum_{y\sim x } (S(x)-S'(x))_{zy}\psi_{xy}\Big|^2\nonumber \\
 &\leq \sum_{x\in V} \sum_{ z\sim x} \Big(\sum_{y\sim x }  \big|(S(x)-S'(x))_{zy}\big|^2\Big) \Big( \sum_{y\sim x } | \psi_{xy}|^2\Big)\nonumber\\ 
 &\leq \sup_{x\in V}\|S(x)-S'(x)\|_{\rm HS}^2 \|\psi\|^2,
\end{align}
{where we used Cauchy-Schwarz to get the first inequality and $\|S(x)-S'(x)\|_{\rm HS}^2\leq \sup_{x\in V}\|S(x)-S'(x)\|_{\rm HS}^2 $ in the last step.}
\qed

{\bf The graph $G$ as subgraph of $K$:}
Any finite graph $G=(V,E)$ of order $|V|$ can be viewed as a subgraph of the complete graph of order $|V|$, denoted by $K$. Defining $\widehat G=(V,\widehat E)$ as the graph with same set of vertices $V$ and set of edges $\widehat E$ distinct from $E$ such that  $K=(V, E\cup \widehat E)$, {{\it i.e.} $\widehat{E}$ is the complement of $E$ in the set of edges of $K$.} We will denote that situation as $K=G\cup \widehat G$. 
Accordingly, denoting by $l^2(\widehat D)$ and $l^2(D_K)$ the Hilbert spaces associated with the directed edges of $\widehat G$ and $K$, we 
show that $l^2(D_K)=l^2(D)\oplus l^2(\widehat D)$ and that $U_\cS$ can be viewed as the restriction to $l^2(D)$ of a unitary SQW, $U_{\cS_K}$, defined on $l^2(D_K)$. 

\medskip

Indeed, we note that  for any $x, y\in V$, $(xy)$ is a directed edge of $K$ and is associated with the basis vector $|xy\ket$ of $l^2(D_K)$.  Since $(xy)$  either belongs  to $E$ or to $\widehat E$, then either $|xy\ket \in l^2(D)$ or $| xy\ket \in l^2(\widehat D)$. Sorting the basis vectors $\{| xy\ket\}_{x,y\in V}$ accordingly, we have $l^2(D)=\text{span}\{ |xy\ket \}_ {(xy)\in E }\subset l^2(D_K)$, $l^2(\widehat D) =\text{span}\{ |xy\ket\}_ {(xy)\in \widehat E }\subset l^2(D_K)$ with $l^2(D_K)=l^2(D)\oplus l^2(\widehat D)$.

Then, for each $x\in V$, let $\widehat d_x\in \N$ be such that $d_x+\widehat d_x=|V|-1$ and consider $\C^{|V|-1}=\C^{d_x}\oplus \C^{\widehat d_x}$. 
In addition to  $\cS=\{S(x)\}_{x\in V}$, with $S(x)\in U(d_x)$ acting on $\C^{d_x}$,  we consider $\widehat \cS=\{\widehat S(x)\}_{x\in V}$, where $\widehat S(x)\in U(\widehat d_x)$ acts on $\C^{\widehat d_x}$. 
This allows us to define $U_{\widehat \cS}$ on $l^2(\widehat D)$. Now, for $K=G\cup \widehat G$, we use the direct sums of scattering matrices $\cS_{K}=\{S(x)\oplus \widehat S(x)\}_{x\in V}$ acting on $\C^{|V|-1}$ for each $x\in V$, to construct 
$U_{\cS_K}$ on $l^2(D_K)$ so that for any $x,y\in V$, 
\be
U_{\cS_K}|xy\ket=\begin{cases}
\sum_{z\sim y \atop (zx)\in E}S_{zy}(x)|zx\ket & \text{ if } \ (xy)\in G\\
\sum_{z\sim y \atop (zx)\in \widehat E}\widehat S_{zy}(x)|zx\ket & \text{ if } \ (xy)\in \widehat G.
\end{cases}
\ee
We have thus shown:
\begin{lem}
For $G$ finite and with the notation above, 
\be
U_{\cS_K}=U_\cS\oplus U_{\widehat \cS} \ \ \text{acting on } \ l^2(D_K)=l^2(D)\oplus l^2(\widehat D).
\ee
\end{lem}
\begin{rem}
The graph $\widehat G=(V, \widehat E)$ is not necessarily connected, in which case it provides $U_{\widehat \cS}$ with a direct sum structure, regardless of  the choice of $\widehat \cS=\{\widehat S(x)\}_{x\in V}$.
\end{rem}

{\bf Asymptotics of $\Q_n^{\psi_0}(x)$:}
The large $n$  behaviour of  $\Q_n^{\psi_0}(x)$, the probability to find the quantum walker on the vertex $x\in V$ at time $n$, is oscillatory in case $G$ is finite. We get a finite limit considering the Ces\`aro mean, which  suppresses the oscillations. In the infinite dimensional case, the Ces\`aro mean limit involves  the nature of $\sigma(U_\cS)$, according to the RAGE theorem, see {\it e.g.} \cite{S}, Theorem 5.5.6, and \cite{RT}, Theorems  B1 and B2. 

\medskip

More precisely, let $E(\cdot)$ defined on $]-\pi, \pi]$ denote the self-adjoint spectral projectors of  $U_\cS$ so that the spectral theorem reads in the resolution of the identity form $U_\cS=\int_{]-\pi, \pi]}e^{i\theta} dE(\theta)$. In particular, $E(\{\theta\})$ is the spectral projector on $\ker (U_\cS -e^{i\theta}\un)$. Denoting by $\sigma_p(U_\cS)$ the set of eigenvalues of $U_\cS$, we have

\begin{lem}\label{rage}
Let  $U_\cS$ have spectral projectors $E(\cdot)$ and $\Q_n^{\psi_0}(x)$ be defined by (\ref{probavertex}). Then
\be
\lim_{N\ra \infty}\frac{1}{N}\sum_{j=0}^{N-1}\Q_n^{\psi_0}(x)
=\sum_{e^{i\theta}\in \sigma_p(U_\cS)} \|P_x^{\rm I}E(\{\theta\})\psi_0 \|^2.
\ee
\end{lem}
\begin{rem} 
i) In case $G$ is infinite and $\psi_0$ belongs to the continuous spectral subspace of $U_\cS$, if any, the limit vanishes.\\
ii) For $G$ finite,  the finite $N$ Ces\`aro mean and  the RHS differ by an $O(N^{-1})$  term.  
\end{rem}
\proof : We write $U$ for $U_\cS$ below, to simplifiy the notation. The proof is a direct consequence of Theorem B2 in \cite{RT},
\be
\text{s-}\lim_{N\ra \infty}\frac{1}{N}\sum_{n=0}^{N-1}U^n |\psi_0\ket\bra \psi_0| U^{-n}=\sum_{e^{i\theta}\in \sigma_p(U_\cS)}E(\{\theta\}) |\psi_0\ket\bra \psi_0| E(\{\theta\}),
\ee 
and of the cyclicity of the trace which yields
\be
\Q_n^{\psi_0}(x)=\tr(P_x^{\rm I}U^n  |\psi_0\ket\bra \psi_0| U^{-n})
\ee
and, {recall that $P_x^{\rm I}$ is finite rank,}
\be
\tr(P_x^{\rm I}E(\{\theta\}) |\psi_0\ket\bra \psi_0| E(\{\theta\}))=\|P_x^{\rm I}E(\{\theta\})\psi_0\|^2.
\ee
\qed
 
 An application of this result to the star-graph is presented in Section \ref{sec:star}.

\section{Particular Cases}\label{sec:particular}

\subsection{The Chalker-Coddington Model}
The Chalker-Coddington model \cite{CC}, providing a simplified description of the Quantum Hall effect, can be cast in the framework of SQW defined on $G=\Z^2$, as we briefly show. See \cite{ABJ1, ABJ2} for a mathematical approach of this model, and references therein for background. 

 Let 
\be
\{{\textsf S}_{j,2k}\}_{j,k\in\Z}, \ \text{with } \ {\textsf S}_{j,2k}\in U(2), \ \forall j,k,
\ee
be a collection of scattering matrices, called even or odd according to the parity of $j$. The unitary operator which defines the
 Chalker-Coddington model 
\be
U_{CC}:l^2(\Z^2)\to l^2(\Z^2)
\ee
reads as follows: Let $\{|j,k\ket\}_{j,k\in\Z}$ be the canonical basis vectors of $l^2(\Z^2)$;  $U_{CC}$  is defined according to figure \ref{fig:scatteringnetwork} by:
\begin{figure}[hbt]
\centerline {
\includegraphics[width=8.cm]{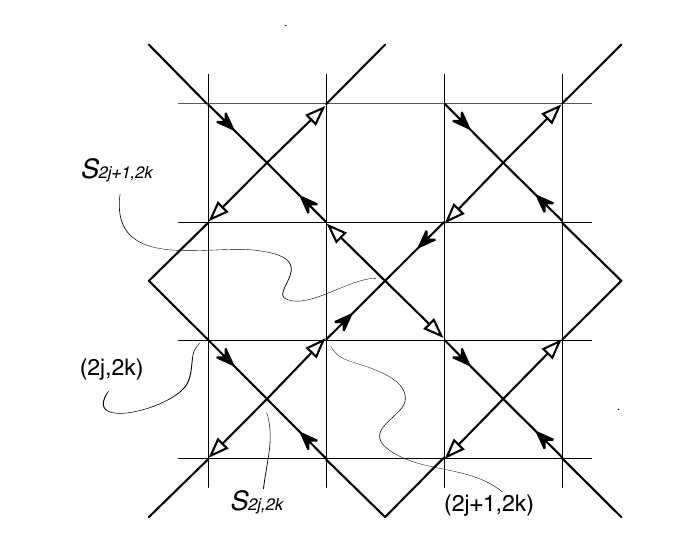}
}
\caption{\small  A Chalker--Coddington model with its incoming (black arrows) and outgoing (white arrows) links.}
\label{fig:scatteringnetwork}
\end{figure}

\begin{align}
&\begin{pmatrix} U_{CC} |2j,2k\rangle \cr U_{CC} |2j+1,2k-1\rangle \end{pmatrix} = {\textsf S}_{2j,2k} \begin{pmatrix}|2j+1,2k\rangle\cr  |2j,2k-1\rangle  \end{pmatrix},\notag \\
&\begin{pmatrix} U_{CC}  |2j+2,2k+1\rangle  \cr U_{CC} |2j+1,2k\rangle \end{pmatrix} =  {\textsf S}_{2j+1,2k} \begin{pmatrix} |2j+1,2k+1\rangle \cr  |2j+2,2k\rangle  \end{pmatrix}. \label{def:UCC}
\end{align}
The convention above is as follows (without indices): 
\be
\begin{pmatrix} U |a\ket \cr U|b\ket \end{pmatrix}={\textsf S}\begin{pmatrix}  |c\ket \cr |d\ket \end{pmatrix} \Leftrightarrow \begin{cases} U |a\ket = {\textsf S}_{11} |c\ket +{\textsf S}_{21} |d\ket \cr 
U |b\ket = {\textsf S}_{12} |c\ket +{\textsf S}_{22} |d\ket  \end{cases},
\ee
where ${\textsf S}_{ij}$ are the matrix elements of ${\textsf S}$, in line with (\ref{defscatu}).

The underlying graph is $G=\Z^2$, with vertices at the intersections of the diagonal lines, where the scattering matrices sit. On the grid of Figure \ref{fig:scatteringnetwork},  the vertices carrying even matrices $\textsf{S}_{2j,2k}$, respectively odd matrices $\textsf{S}_{2j+1,2k}$,  have coordinates 
\begin{align}
x^e(j,k)&=(2j, 2k)+(1/2, -1/2), \ \ \text{respectively}\nonumber\\ 
x^o(j,k)&=(2j+1, 2k)+(1/2, 1/2), \ \ \forall (j,k)\in \Z^2.
\end{align}
The edges are the diagonal segments labelled by their center $(j,k)$ on that grid and they have one direction only given by their arrow. 

Note that we are missing half the directed edges with respect to the construction of $l^2(D)$, where each edge comes with two orientations. So we associate with each $(j,k)\in \Z^2$ the canonical basis vectors $|j,k\ket_-$ of $l^2(\Z^2)$, corresponding to edges labeled 
by $(j,k)$ with orientation  {\it opposite } to the arrow. These vectors span another copy of $l^2(\Z^2)$ we denote by $l^2_-(\Z^2)$. For example, $|2j,2k\ket$  and $|2j+1,2k-1\ket$ are incoming to $x^e(j,k)$ while $|2j,2k\ket_-$  is incoming to $x^o(j-1,k)$ and $|2j+1,2k-1\ket_-$ is incoming to $x^o(j,k-1)$, and similarly for the other directed edges. 
Moreover,  we have $l^2(D)=l^2(\Z^2)\oplus l^2_-(\Z^2)$.

Since the degree of all vertices is 4, we need $4\times 4$ scattering matrices at each vertex to realize the Chalker Coddington model  as a SQW on $l^2(D)$.
We label and order the canonical basis of $\C^4$ at each vertex as $\{e_{NW}, e_{SE},  e_{NE}, e_{SW}\}$, where the cardinal symbols  indicate the incoming and outgoing directions at the vertices, see (\ref{defug}).   Then for every vertex $x^\#(j,k)$, $(j,k)\in \Z^2$, $\#\in \{e,o\}$, we define {with ${\mathbb O}$ the zero $2\times 2$ matrix}
\be
S(x^\#(j,k))=\begin{pmatrix}  {\mathbb O} & {\textsf S}_ {2j+1, 2k}\cr
{\textsf S}_{2j, 2k} &  {\mathbb O} 
\end{pmatrix} \in U(4)
\ee
with respect to the chosen ordered basis,  and set $\cS_{CC}=\{S(x^\#(j,k))\}_{(j,k)\in \Z^2 \atop \#\in \{e,o\}}$. 
\begin{prop} With the notation and choices made above, the SQW  
defined on $l^2(D)$ according to (\ref{defug}) and parameterized by $\cS_{CC}$, denoted by $U_{\cS_{CC}}$,  satisfies
\be
U_{\cS_{CC}}= U_{CC}\oplus \tilde U_{CC}  \ \ \text{\em on}\ \  l^2(D)=l^2(\Z^2)\oplus l^2_-(\Z^2)
\ee
where $U_{CC}=U_{\cS_{CC}}|_{l^2(\Z^2)}$ is the Chalker-Coddington operator (\ref{def:UCC}). \\
The restriction  $\tilde U_{CC}=U_{\cS_{CC}}|_{l^2_-(\Z^2)}$ is a copy of this model, where all orientations are reversed and for all 
$(j,k)\in \Z^2$, the scattering matrices $S_{2j, 2k}$ and $S_{2j+1,2k}$ are exchanged.
\end{prop}
\begin{rem}
Casting the Chalker-Coddington model into a SQW amounts to add another independent copy of the model in the process of completing the set of directed edges.
\end{rem}
\proof:
It is a matter of computation to check that the action of $U_{\cS_{CC}}$ on basis vectors of the form $|j,k\ket$ coincides with (\ref{def:UCC}) while its action on vectors  $|j,k\ket_-$ yields
\begin{align}\label{deftildeUCC}
&\begin{pmatrix} U_{\Z^2}^{\cS_{CC}} |2j+1,2k\rangle_- \cr U_{\Z^2}^{\cS_{CC}}|2j,2k-1\rangle_- \end{pmatrix} = {\textsf S}_{2j+1,2k} \begin{pmatrix}|2j,2k\rangle_-\cr  |2j+1,2k-1\rangle_-  \end{pmatrix},\notag \\
&\begin{pmatrix} U_{\Z^2}^{\cS_{CC}} |2j+1,2k+1\rangle_-  \cr U_{\Z^2}^{\cS_{CC}} |2j+2,2k\rangle_- \end{pmatrix} =  {\textsf S}_{2j,2k} \begin{pmatrix} |2j+2,2k+1\rangle_- \cr  |2j+1,2k\rangle_-  \end{pmatrix}.
\end{align}
Therefore both $l^2(\Z^2)$ and $l^2_-(\Z^2)$ are invariant under $U_{\cS_{CC}} $ and, by inspection, $\tilde U_{CC}$ coincides with a Chalker-Coddington model as described.  
\qed

The Chalker-Coddington model is one of many quantum network models. While this designation is not universal, it often refers to QWs defined on graphs in a similar fashion as (\ref{defscatu}) with the understanding that the edges can be travelled by the quantum walker in one direction only. This imposes the degrees $d_x$ to be even for all $x\in V$, so that  the  scattering matrix $S(x)\in U(d_x)$ assigned to  $x\in V$ maps the $d_x/2$ incoming edges at $x$ to the $d_x/2$ outgoing edges from $x$, see \cite{D} for example. The Scattering Zippers, introduced and studied in \cite{MS-B, BM} provide another example related to the Chalker-Coddington model defined on a strip. 

Without going into the details, such quantum network models can be viewed as restrictions of SQW defined on a doubled Hilbert space, as we saw for the Chalker-Model model.

\subsection{Coined Quantum Walks}

Coined QWs defined on graphs as introduced by \cite{AAKV}, provide a useful concept in quantum computing and a versatile modeling tool in quantum dynamics. We show that they also belong to the class of SQWs on regular graphs.

Let us consider the case where all scattering matrices are identity $d_x\times d_x$ matrices, $S(x)=\un_{d_x}$, which reduces (\ref{defug}) to 
\be\label{fdeff}
F=\sum_{x\in V}\sum_{y\sim x}  |yx\ket \bra xy| = F^*=F^{-1}.
\ee
The notation $F$ is justified by the flip property $F |xy\ket=|yx\ket$, for all $x\sim y$. Therefore $F$ is reduced by all orthogonal two-dimension subspaces $\cH_{xy}=\mbox{span }\{|xy\ket, |yx\ket\}$, where $x\sim y$, with matrix representation in the ordered basis $\{|xy\ket, |yx\ket\}$,
\be
F|_{\cH_{xy}}=\begin{pmatrix}0 & 1 \cr 1 & 0\end{pmatrix} \ \Rightarrow \ F=\bigoplus_{(xy)\in E}F|_{\cH_{xy}} \ \mbox{and} \ \sigma (F)=\{-1,1\}.
\ee
This trivial case allows for the decomposition of the general case according to 
\be\label{decomp}
U_\cS=F(F U_\cS) \ \ \mbox{where} \ \ F|_{\cH_x^{\rm O}}:\cH_x^{\rm O}\ra \cH_x^{\rm I}, \ \ \mbox{and} \ \ F U_\cS|_{\cH_x^{\rm I}} : \cH_x^{\rm I} \ra \cH_x^{\rm I}, \ \ \forall x\in V,
\ee 
such that
\be \label{C}
F U_\cS=\bigoplus_{x\in V} S(x)  \ \ \mbox{where} \ \ S(x)\simeq \sum_{y\sim x\atop z\sim x} S_{zy}(x) |xz\ket \bra xy|:\cH_x^{\rm I} \ra \cH_x^{\rm I}.
\ee
In the decomposition (\ref{decomp}), $F$ is responsible for the motion of the quantum walker between different subspaces $\cH_x^{\rm I}$, while $FU_\cS$ changes its state locally within each $\cH_x^{\rm I}$. 

\medskip

The general structure \eqref{decomp} of $U_\cS$  as the composition of a local unitary operator and an operator that implements the motion is that of a Coined QW. However, in many models of quantum dynamics by Coined QWs on infinite regular graphs derived from the Schr\"odinger equation, the operator responsible for the motion of the walker is unitarily equivalent to a direct sum of shifts, analogous to the effect of the Laplacian, and therefore has absolutely continuous spectrum. The effect of the potential is local and analogous to the action of $F U_\cS$. These specificities are relevant when analyzing the spectral and dynamical (de-)localization properties of {Random QWs},  see {\it e.g.} \cite{Ko, J3, HJ, JM, ABJ2, ABJ3, ABJ4, ABJ5}. Coined QWs of this sort defined on regular infinite graphs can be viewed as SQW:
\medskip

Assume $G$ is an infinite $d$ regular graph, and let $\theta:\{1,2,\dots, d\}$ be a permutation. Each vertex $x\in V$ forms exactly $d$ edges with vertices in $V$  that are labeled  $x_1, x_2, \dots, x_d$. Set $F_\theta$, the  operator defined by its action on the basis vectors of $l^2(D)$
\be\label{defft}
F_\theta|_{\cH_x^{\rm I}}: \cH_x^{\rm I}\ra \cH_x^{\rm O} \ \ \text{s.t.} \ \ F_\theta |xx_j\ket = |x_{\theta(j)} x\ket, \ \ \forall \ j\in \{1,2,\dots, d\}, x\in V.
\ee
By definition, $F_\theta$ is unitary, as a change of basis, and  its adjoint $F_\theta^*$ acts as
\be\label{defftstar}
F_\theta^*|_{\cH_x^{\rm O}}: \cH_x^{\rm O}\ra \cH_x^{\rm I} \ \ \text{s.t.} \ \ F_\theta^* |x_j x\ket = |x x_{\theta^{-1}(j)}\ket, \ \ \forall \ j\in \{1,2,\dots, d\}, x\in V.
\ee
We have a  decomposition similar to (\ref{decomp}) thanks to $F_\theta$
\be\label{decompcqw}
U_\cS=F_\theta(F_\theta^* U_\cS) \ \ \mbox{where}  \ \ F_\theta^* U_\cS|_{\cH_x^{\rm I}} : \cH_x^{\rm I} \ra \cH_x^{\rm I}, \ \ \forall x\in V,
\ee
\medskip
and, with the shorthand $S_{jk}(x)=S_{x_j x_k}(x)$, 
\be\label{coinsqw}
F_\theta^* U_\cS |x x_j\ket=\sum_{k}S_{\theta(k) j}(x) |x x_{k}\ket.
\ee
Then we observe that  for each $x\in V$, the matrix $S_\theta(x)=(S_{\theta(k) j}(x))_{k,j}\in U(d)$ is a scattering matrix as well, obtained by permuting the rows of $S(x)$, so that
\be
F_\theta^* U_\cS=\bigoplus_{x\in V}S_\theta(x), \ \ \text{where} \ S_\theta(x): \cH_x^{\rm I}\ra  \cH_x^{\rm I}
\ee 
is a local (coin) operator  with arbitrary scattering matrices $S_\theta(x)$. The spectral properties of $F_\theta$ depend on the specificities of the graph and of the permutation $\theta$. 
\medskip

We spell out the case  $G=\Z^d$, where the degree of each vertex $x\in \Z^d$ is $2d$. Set  $I=\{\pm 1, \pm 2, \dots, \pm d\}$ and let $\{e_1, e_2, \dots, e_d\}$ be the canonical basis of $\R^d$. We denote a generic basis vector of $\cH_x^{\rm I}\subset l^2(D)$ as 
\be\label{canzd}
|x x_\tau\ket, \ \ \text{where} \ \ \tau\in I, \ \ \text{and} \ \ x_\tau=x-\sign (\tau) e_{|\tau|}\in \Z^d.
\ee 
The next lemma shows that the SQW $U_\cS$ on $l^2(D)$ is unitarily equivalent to a Coined QW:

\begin{lem}
Let $G=\Z^d$ and $l^2(D)$ with canonical basis given by $\{|x x_\tau\ket\}_{x\in \Z^d, \tau \in I}$, as in  (\ref{canzd}).
We have the unitary equivalence  
\be\label{eqspaces}
l^2(D)\simeq l^2(V)\otimes \C^{2d}=\text{\em span } \{|x\ket \otimes |\tau\ket , \ x\in V, \tau \in I\},
\ee
with the identification
\be
 |x x_\tau\ket \simeq |x\ket \otimes |\tau\ket
\ee
where $\{|\tau\ket\}_{\tau\in I}$ labels an orthonormal basis of  $\C^{2d}$.\\
Let  $\theta$ defined by
\be\label{chothet}
\theta: I\ra I \ \ \text{s.t.} \ \ \theta (\tau)=-\tau, \ \ \forall \tau \in I,
\ee
then we have the unitary equivalences
\be
F_\theta \simeq T , \ \ F_\theta^* U_\cS \simeq  C,
\ee
where for all $x\in V$, $\tau\in I$,
\begin{align}\label{sqwcoin}
&C (|x\ket \otimes |\tau\ket) =|x\ket \otimes \sum_{\tau'\in I}S_{\theta}(x)_{\tau' \tau}  |\tau'\ket,\nonumber\\
&T  (|x\ket \otimes |\tau\ket )= |x+{\rm sign} (\tau) e_{|\tau|} \ket \otimes |\tau\ket,
\end{align}
so that $U_\cS$ is unitarily equivalent to the Coined QW defined by the composition $TC$.
\end{lem}
\proof:
With our convention that $(x x_\tau)$ points in the direction of $x$, this edge is independent of $x\in \Z^d$. The corresponding basis vectors of $\cH_x^{\rm I}$ read 
\be
|x x_\tau\ket = | \sign (\tau) e_{|\tau|} \ket, \ \ \tau \in I,
\ee
and point in the $2d$ directions of the lattice. Denoting by $\{|\tau\ket\}_{\tau\in I}$ an orthonormal basis of  $\C^{2d}$, we have 
\be
\cH_x^{\rm I}=\text{span } \{|x x_\tau\ket , \ \tau \in I\}\simeq \C^{2d}=\text{span } \{  |\tau\ket \in I\}.
\ee
Making use of $l^2(D)=\oplus_{x\in V}\cH_x^{\rm I}$, see (\ref{dirsumspa}), we deduce the announced unitary equivalence (\ref{eqspaces}).
Now, considering  (\ref{defft}), we get with (\ref{chothet}) and (\ref{canzd})
\be\label{imshift}
F_\theta |x x_{\tau}\ket=|x_{-\tau}x\ket = |y y_\tau\ket,
\ee
with $y=x+\sign (\tau) e_{|\tau|}$. In the space $ l^2(V)\otimes \C^{2d}$, this amounts to the action of $T$, the translation by $\sign (\tau) e_{|\tau|}$ in $V$, leaving the component in $\C^{2d}$ invariant. 
Finally, (\ref{coinsqw}) yields  for each $x\in V$, $\tau \in I$
\be
(F_\theta^* U_\cS)|x x_\tau\ket= \sum_{\tau'\in I}S_{\theta}(x)_{\tau' \tau}  |x x_{\tau'}\ket \in \cH_x^{\rm I},
\ee
which gives the equivalence with $C$.
\qed

\begin{rem} i) The composition $T C$ is the prototypical Coined QW on $\Z^d$.\\
ii) Other choices {than \eqref{chothet} for} the permutation $\theta$, besides the identity permutation, may lead $F_\theta$ to be reduced by infinitely many finite dimensional subspaces and to be pure point spectrum. 
\end{rem}

\medskip

Similar considerations, which we don't make explicit here, show that the Coined QWs on  the $d-$regular tree, studied in \cite{HJ, JMa, T} for example, are also special cases of SQWs.

\subsection{Star-graph}\label{sec:star}

Consider $G$ to be the star-graph $SG$ with $N$ branches characterized by vertices $x_0, x_1, \dots, x_N$, with edges between $x_0$ and $x_j$, $1\leq j\leq N$, only, see Figure \ref{stargraph}. 

\begin{figure}[h]
\centering
\begin{tikzpicture}
  \node[circle, draw] (center) at (0,0) {$x_0$};
  
  \foreach \i/\label in {1/x_1, 2/x_2, 3/x_3, 4/x_4, 5/x_N} {
    \pgfmathsetmacro{\angle}{360/5 * (\i - 1)}
    \ifnum\i=4
      \node (dots) at (\angle:1.5cm) {$\dots$};
    \else
      \node[circle, draw] (vertex\i) at (\angle:2cm) {$\label$};
      \draw[->, thick] (center) to[bend left=+15] (vertex\i); 
      \draw[->, thick] (vertex\i) to[bend left=+15] (center);
    \fi
  }
\end{tikzpicture}
\caption{\small Star-graph $SG$ with $N$ branches.}
\label{stargraph}
\end{figure}
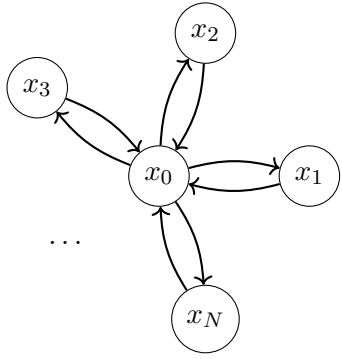

To simplify the notation, we denote the vertices $x_j$
by their label $j$ only. Hence we have $d_0=N$, and $d_j=1$, $1\leq j\leq N$, and the $2N$ canonical basis vectors of $l^2(D)$ are $\{|0j\ket, |j0\ket\}_{1\leq j\leq N}$. To the center vertex $0$ we associate the scattering matrix $S(0)=(S_{kj}(0))_{1\leq k, j \leq N}\in U(N)$, and to the vertices $j$, $1\leq j\leq N$, we associate $S(j)=e^{i\theta_j}\in U(1)$. We denote this set of matrices by $\cS_{SG}$
\medskip

The matrix form of $U_{\cS_{SG}}\in U(2N)$ in the  basis $\{|01\ket, \dots, |0N\ket,  |10\ket, \dots, |N0\ket\}$ (using the same symbol) reads
\begin{align}\label{stargraphN}
&U_{\cS_{SG}}=
\begin{pmatrix}
0 & 0 & \dots & 0 & e^{i\theta_1} & 0 & \dots & 0 \cr
0 & 0 & \dots & 0 & 0 & e^{i\theta_2}  & \dots & 0 \cr
\vdots & &\ddots &\vdots & \vdots& & \ddots&\vdots \cr 
0 & 0 & \dots & 0 & 0 &  0 & \dots & e^{i\theta_N} \cr
S_{11}(0) & S_{12}(0) & \dots & S_{1N}(0) & 0 & 0 & \dots & 0\cr
S_{21}(0) & S_{22}(0) & \dots & S_{2N}(0) & 0 & 0 & \dots & 0\cr
\vdots & & \ddots& \vdots&\vdots & & \ddots&\vdots \cr 
S_{N1}(0) & S_{N2}(0) & \dots & S_{NN}(0) & 0 & 0 & \dots & 0
\end{pmatrix}
=\begin{pmatrix}
0 & D(\theta) \cr S(0) & 0
\end{pmatrix},
\end{align} 
where $D(\theta)\in U(N)$ and $0 \in M_N(\C)$ are defined by the block structure  of $U_{SG_N}$. 
As 
\be
U_{\cS_{SG}}^2=\begin{pmatrix}
 D(\theta)S(0) & 0  \cr  0 & S(0)D(\theta) 
\end{pmatrix},
\ee
where  $S(0)D(\theta)$ and $D(\theta)S(0)$ are unitarily equivalent, 
\be 
\sigma(D(\theta)S(0))=\sigma(S(0)D(\theta))=\{e^{i\alpha_j}\}_{1\leq j\leq N},
\ee 
with eigenvalues repeated according to multiplicity, we deduce that
\be\label{specUSG}
\sigma(U_{\cS_{SG}})=\{\pm e^{i\alpha_j/2}\}_{1\leq j\leq N}. 
\ee

\begin{rem}
Observe that we can allow $d_0=N=\infty$ in this case, with $l^2(D)=l^2(\N^*)\oplus l^2(\N^*)$, and 
$S(0)$ a unitary operator on  $l^2(\N^*)$, and $D(\theta)$ a diagonal unitary operator on $ l^2(\N^*)$. It is also possible to consider $l^2(D)=l^2(\Z)\oplus l^2(\Z)$, labelling the branches of the star graph from $-\infty$ to $\infty$, with $S(0)$ and $D(\theta)$ being unitary operators on $l^2(\Z)$.
\end{rem}

Therefore, the restriction $U_{\cS_{SG}}^2|_{\cH^{\rm I}_0}\simeq D(\theta)S(0)$ yields a general unitary operator acting on the (in-)finite dimensional Hilbert space $\cH^{\rm I}_0$, that can be cast under the form of a QW, by designing $S(0)$ appropriately. Operators of this kind are the starting point for the analysis of certain deterministic unitary operators $S(0)$ perturbed by random phases $D(\theta)={\rm diag} (e^{i\theta_j})$, appearing in models displaying Anderson localization \cite{Ko, J3,  HJS2, JM, ABJ2, ABJ3, ABJ4, ABJ5}.

\medskip

We proceed with the case of the finite star-graph to illustrate Lemma \ref{rage}. We thus need to determine the spectrum and eigenprojectors of (\ref{stargraphN}). Since the main steps hold  in the infinite dimensional case, consider the following slightly more abstract framework:

Let $\cH$ be a separable Hilbert space and let $S\in \cB(\cH)$, $D\in \cB(\cH)$ be unitary operators. We consider $U$ defined on the direct sum of Hilbert spaces $\cK=\cH\oplus \cH$ by the block representation
\be\label{blocU}
U= \begin{pmatrix}
0 & D \\
S & 0
\end{pmatrix}.
\ee 
\begin{lem}\label{resolblocU} For any $|z|\neq 1$ we have
\be
(U-z\un_\cK)^{-1}=\begin{pmatrix}
z(DS-z^2\un_\cH)^{-1} & D(SD-z^2\un_\cH)^{-1}  \\
S(DS-z^2\un_\cH)^{-1}  & z(SD-z^2\un_\cH)^{-1}
\end{pmatrix},
\ee
and $\sigma(U)=\{\lambda \in \mathbb S\, | \, \lambda^2\in \sigma(DS)\}$.\\
If $e^{i\alpha}$ is an eigenvalue of $DS$ with associated spectral projector $P(\{\alpha\})$, then $\pm e^{i\alpha/2}$ is an eigenvalue of $U$ with associated spectral projector
\be\label{ealpha2}
E^\pm(\{\alpha/2\})=\frac12 \begin{pmatrix}
P(\{\alpha\}) & \pm D Q(\{\alpha\})e^{-i\alpha/2}\\
\pm D^{-1} P(\{\alpha\})e^{i\alpha/2} & Q(\{\alpha\})
\end{pmatrix},
\ee
where $Q(\{\alpha\})=D^{-1}P(\{\alpha\}) D$ is the spectral projector of $SD$ associated with $e^{i\alpha}\in \sigma(SD)$. The multiplicity of $\pm e^{i\alpha/2}\in \sigma(U)$ equals that of $e^{i\alpha}\in \sigma(DS)$.
\end{lem}
\begin{rem}
i) We have
\be
E^\pm(\{\alpha/2\})=\begin{pmatrix}
\un_\cH & 0\\
0 & D^{-1}
\end{pmatrix}
\frac12 \begin{pmatrix}
P(\{\alpha\}) & \pm P(\{\alpha\})e^{-i\alpha/2}\\
\pm P(\{\alpha\})e^{i\alpha/2} & P(\{\alpha\})
\end{pmatrix}
\begin{pmatrix}
\un_\cH & 0\\
0 & D
\end{pmatrix},
\ee
where the operator in the middle is the spectral projector of $ \begin{pmatrix}
0 & \un_\cH \\
DS & 0
\end{pmatrix}$, which is unitarily equivalent to $U$.\\
ii) The expression for the resolvent of $U$ holds for bounded operators $S$ and $D$, for $z^2\not\in \sigma(DS)\cup \sigma(SD)$. \\
iii) In case $\cH$ is finite dimensional and $DS$ admits the spectral decomposition
\be
DS=\sum_{j=1}^{n} e^{i\alpha_j}P_j, \ \ \text{with} \ \ P_j=P(\{\alpha_j\}),
\ee
$U$ admits the following spectral decomposition 
\be\label{specdecstar}
U=\sum_{j=1}^{n} e^{i\alpha_j/2} \Pi_j^+ - e^{i\alpha_j/2} \Pi_j^-, \ \ \text{where } \ \ \Pi_j^\tau=E^\tau(\{\alpha_j/2\}), \ \ \tau \in \{+, -\}.
\ee
\end{rem}
\proof:
The expression of the resolvent can be derived using the formulae for $p\in \N$
\be
U^{2p}= \begin{pmatrix}
(DS)^p & 0 \\
0 & (SD)^p
\end{pmatrix}, \ \  U^{2p+1}= \begin{pmatrix}
0 & D(SD)^p \\
S(DS)^p & 0
\end{pmatrix}, 
\ee
and $(U-z\un_{\cK})^{-1}=-\frac1z\sum_{j=0}^\infty U^j/z^j$, for $|z|>1$. {A direct verification shows the formula holds for $|z|<1$ as well.} Since $\sigma(DS)=\sigma(SD)$, the statement about $\sigma(U)$ follows. The spectral projector is obtained by the Riesz formula in case $\pm e^{i\alpha/2}$ is isolated in the spectrum. {In any case, a direct verification shows that $E^\pm(\{\alpha/2\})$ is an orthogonal projector such that $UE^\pm(\{\alpha/2\})=E^\pm(\{\alpha/2\})U=\pm e^{i\alpha/2}E^\pm(\{\alpha/2\})$. Moreover, $\tr (E^\pm(\{\alpha/2\})=\tr(P\{\alpha\}\})$, showing  $\pm e^{i\alpha/2}$ has the same multiplicity as $e^{i\alpha}$. } 
\qed

Let us come back to the star-graph with $N$ branches, and the ordered basis chosen before (\ref{stargraphN}). The structure of that basis allows us to 
express $l^2(D)=\cH^{\rm I}_{x_0}\oplus \cH^{\rm O}_{x_0}$, where both subspaces can be identified to $\C^N$. Accordingly, we write $\psi=\begin{pmatrix}\psi^{\rm I} \\ \psi^{\rm O}\end{pmatrix}$, where the indices $\rm I, O$ refer to $\cH^{\rm I}_{x_0}$, $\cH^{\rm O}_{x_0}$. To keep things simple, we assume that $ \psi_0 = \psi_0^{\rm I}$ and that the spectrum of $D(\theta)S(0)$ is simple. This allows us to compute  the asymptotic (Ces\`aro mean) probabilities to find the quantum walker on the center $x_0$  and on the branch $x_k$ of the star-graph.
\begin{prop}
Consider the SQW $U_{\cS_{SG}}$ defined by (\ref{stargraphN}) and assume the spectrum $\sigma(D(\theta)S(0))$ is simple, with corresponding orthonormal basis of eigenvectors $\{\phi_j\}_{1\leq j\leq N}$. Consider a normalized initial vector $ \psi_0 = \psi_0^{\rm I}\in \cH_{x_0}^{\rm I}$. Then, 
\begin{align}\label{ragetype}
\lim_{N\ra \infty}\frac{1}{N}\sum_{j=0}^{N-1}\Q_n^{\psi_0^{\rm I}}(x_0)
&=\frac12\nonumber\\
\lim_{N\ra \infty}\frac{1}{N}\sum_{j=0}^{N-1}\Q_n^{\psi_0^{\rm I}}(x_k)
&=\sum_{j=1}^N\frac12|\bra k0| \phi_j\ket|^2|\bra\phi_j|\psi_0^{\rm I}\ket|^2.
\end{align}
\end{prop}
\proof: We use the lighter notation (\ref{blocU}) and that of Lemma \ref{resolblocU} in this proof.
 To make use of  Lemma \ref{rage}, we compute 
$\|P^{\rm I}_{x_0}\Pi^\tau_j \psi_0\|^2$, for $ \tau\in\{+,-\}$, $1\leq j \leq N$ {using \eqref{ealpha2} and \eqref{specdecstar}.}
As a matrix in the chosen basis, we have
$
P^{\rm I}_{x_0}= \begin{pmatrix}
\un & 0 \\
0 & 0
\end{pmatrix}
$, so that 
\be\label{pix0}
P^{\rm I}_{x_0}\Pi_j^\tau \psi_0 = \frac12 \begin{pmatrix}P_j\psi_0^{\rm I} \\ 0 \end{pmatrix} \ \ \Rightarrow \ \ \|P^{\rm I}_{x_0}\Pi_j^\tau \psi_0^{\rm I}\|^2=\frac14\|P_j\psi_0^{\rm I}\|^2.
\ee
Consider now $P^{\rm }_{x_k}= |k0\ket\bra k0|$,  $1\leq k\leq N$, such that for any $\phi\in l^2(D)$, $\|P^{\rm }_{x_k}\phi\|^2=|\bra k0| \phi\ket|^2$. We compute
\be
P^{\rm I}_{x_k}\Pi^\tau_j\psi_0^{\rm I}=\begin{pmatrix}0 \\ |k0\ket\bra k0| (\Pi^\tau_j\psi_0^{\rm I})^{\rm O} \end{pmatrix}, \ \ \text{where} \ \ (\Pi^\tau_j\psi_0^{\rm I})^{\rm O} =\frac12 \tau D^{-1}P_je^{i\alpha_j/2}\psi_0^{\rm I}.
\ee
Eventually, making use of the fact that $D$ is diagonal and unitary, and of the assumption  $P_j=|\phi_j\ket\bra \phi_j|$, we get% (\ref{astar})
\be\label{pixk}
\|P^{\rm I}_{x_k}\Pi^\tau_j\psi_0^{\rm O}\|^2=\frac14|\bra k0| \phi_j\ket|^2|\bra\phi_j|\psi_0^{\rm I}\ket|^2.
\ee
{It remains to sum the expressions \eqref{pix0} and \eqref{pixk} over $j\in\{1,\dots, N\}\}$ and $\tau\in\{+,-\}$ to get the limits in \eqref{ragetype} according to Lemma \ref{rage}, using the fact that $\{\Phi_j\}_{1\leq j\leq N}$ form an orthonormal basis in \eqref{pix0}.}
\qed

\subsection{Generalized Grover Walk}

A popular unitary QW used in quantum computing is the so called Grover QW, which corresponds to the following set of scattering matrices, see {\it e.g.}
\cite{P, HKSS2}
\be \label{groscat}
S(x) = \frac{2}{d_x}A-\un \in U(d_x), \ \ \text{for all} \ \ x\in V,
\ee
where $\un$ denotes the identity matrix and, for all $x\in V$, $A$ is the all-ones matrix, $A_{zy}=1$, $y\sim x, z\sim x$.  
In this section, we introduce and analyze a slight generalization of the Grover walk and we prove a spectral mapping theorem between this QW and a self-adjoint operator related to the adjacency matrix of the graph $G$. This kind of result has been shown in various setups in a series of papers, see \cite{HKSS2, HS, HSS}, for the Grover walk, by means of boundary operators. Our approach of the generalized Grover walk based on  the Feshbach-Schur method proposes an alternative route to that used in these references, that we believe is interesting in its own right. We provide a first result for finite graphs, that we then generalize to the infinite dimensional case in an abstract setting. 

Note that the  boundary operator method developed in \cite{HKSS2, HS, HSS} to reduce the problem originally stated on $l^2(D)$ to an operator on $l^2(V)$, will be used again in Section \ref{sec:IOQW} to define induced open SQWs.

\medskip

Let $G$ be finite and let $\alpha\in (-\pi, \pi]$. For each $x\in V$, let $\omega(x)\in \C^{d_x}$ such that $\|\omega(x)\|=1$. The generalized Grover walk of parameter $\alpha$, or $\alpha-$Grover walk, is defined by the set $\cS=\{S(x), x\in V\}$ where:
\be\label{Salpha}
S(x)=|\omega(x)\ket\bra\omega(x)|+e^{i\alpha}(\un_{d_x} - |\omega(x)\ket\bra\omega(x)|), \ \ \forall x\in V.
\ee 
For $\alpha= 0$, $S(x)= \un_{d_x}$, and $\alpha=\pi$ corresponds to the Grover walk. Note that in case $d_x=1$, $S(x)=1$.

We thus assume $\alpha\neq 0$ in the following and note that (\ref{Salpha}) yields the spectral decomposition of $S(x)$, with $\sigma(S(x))=\{1, e^{i\alpha}\}$, where 1 has multiplicity 1 and $e^{i\alpha}$ has multiplicity $1-d_x$. The corresponding SQW is denoted by $U_{\alpha}$ in this section.
As in \eqref{C}, we identify $\C^{d_x}$ and $\cH_x^{\rm I}$, and abuse slightly notations to write
\be
\omega(x)=\sum_{y\sim x}\omega_y(x)|xy\ket.
\ee
Accordingly, with $F$ given in \eqref{fdeff},
\be\label{Calpha}
F U_{\alpha}= (\Pi + \e^{i\alpha}(\un - \Pi)),
\ee
see \eqref{C}, where 
\be\label{piproj}
\Pi=\bigoplus_{x\in V} |\omega(x)\ket\bra\omega(x)|=\Pi^2=\Pi^*
\ee
is the spectral projector corresponding to the eigenvalue 1 of $F U_\alpha$, and $\un-\Pi$ is that corresponding to the eigenvalue $e^{i\alpha}$ of that unitary operator, with
\be\label{dimprojpi}
\dim \Pi = |V|, \ \text{and} \ \dim (\un -\Pi)=|D|-|V|.
\ee

We introduce the orthogonal subspaces $\cH_1=\Pi l^2(D)$, $\cH_2=(\un-\Pi)l^2(D)$ and express the operator $F$ in matrix form according to $l^2(D)=\cH_1\oplus \cH_2$ as
\be
F=\begin{pmatrix} F_{11} & F_{12} \\ F_{21} & F_{22}\end{pmatrix}, \ \text{where} \ F_{ij}: \cH_j \ra \cH_i.
\ee
In particular, $F_{jj}=F_{jj}^*$ on $\cH_j$ and since $F$ is unitary, $F_{jj}$ is a contraction on $\cH_j$:  $\|F_{jj}\|\leq 1$, $j\in \{1,2\}$.

\medskip

The following map will be important for the result to come: \\
 Let $\alpha\in (-\pi, \pi]\setminus\{0\}$ and  $\ffi_\alpha : {\mathbb S}^1\ra \R$, with ${\mathbb S}^1$ the unit circle, defined by 
\be\label{fial}
\ffi_\alpha(\lambda)=\frac{\lambda^2-\e^{i\alpha}}{\lambda(1-\e^{i\alpha})}.
\ee 
It is readily seen, writing $\lambda=\e^{i\theta}$, $\theta \in (-\pi, \pi]$,  that  we have
\be\label{lamthe}
\ffi_\alpha(\e^{i\theta})=\mu=\frac{\sin(\alpha/2-\theta)}{\sin(\alpha/2)}\in [-|\sin(\alpha/2)^{-1}|, |\sin(\alpha/2)^{-1}|].
\ee
Moreover,  for $\alpha\neq \pi$, $\ffi_\alpha^{-1}(\{\mu\})$ always consists of two distinct values $\lambda_+ \neq \lambda_-$ such that $\lambda_+ \lambda_-=-\e^{i\alpha}$. For $\alpha=\pi$, the same is true, except for $\ffi_\pi^{-1}(\{\pm1\})=\pm1$.

\medskip

Here is our first spectral mapping result between $U_{\alpha}$ and $\Pi F\Pi|_{\Pi l^2(D)}=F_{11}$ in the finite dimensional case:

\begin{thm} \label{findim} Let $\alpha\in (-\pi, \pi]\setminus\{0\}$, $\ffi_\alpha$ be defined by \eqref{fial}, and $\dim l^2(D)<\infty.$
Then,
\begin{align}
&\lambda \in \sigma(U_{\alpha})\setminus \{\pm\e^{i\alpha}\} \Rightarrow \mu \in  \sigma(F_{11}), \ \text{where} \ \mu=\ffi_\alpha(\lambda), \\
& \mu \in  \sigma(F_{11})\setminus\{\pm 1\} \Rightarrow \lambda_-, \lambda_+ \in \sigma(U_{\alpha}) , \ \text{where} \ \ffi_\alpha^{-1}(\{\mu\})=\{\lambda_-, \lambda_+\}.
\end{align}
Moreover,  for $\mu\in (-1,1)$, $\ffi_\alpha^{-1}(\{\mu\})=\{\lambda_-, \lambda_+\}$,
\begin{align}\label{30}
\dim \ker (U_{\alpha}-\lambda_\pm\un)=\dim \ker (F_{11}-\mu \un_1)=\dim \ker (F_{22}+\mu \un_2),
\end{align}
and, for $\alpha\neq \pi$, 
\be\label{eqker}
\ker (U_{\alpha}\pm \un) = \ker (F_{11}\pm \un_1),  \
\ker (U_{\alpha}\pm e^{i\alpha}\un)=\ker (F_{22}\pm\un_2).
\ee
For $\alpha=\pi$,  
\be\label{eqkerpi}
\ker (U_\pi\pm \un) = \ker (F_{11}\pm \un_1)+\ker (F_{22}\mp\un_2).
\ee
\end{thm}
\begin{rem}
i) 
The result only depends on the  factorization \eqref{Calpha}.\\
ii) {From the first statement, if $\lambda=\e^{i\theta}\in \sigma(U_{\alpha})\setminus \{\pm\e^{i\alpha}\} $,} with $\theta \in (-\pi, \pi]$ as in (\ref{lamthe}), we have $|\sin(\alpha/2-\theta)|\leq |\sin(\alpha/2)|$, 
since $F_{11}$ is a contraction and $|\mu|\leq 1$.\\
iii) We have $\ker(F_{11}\mp\un_1)=\Pi\ker(F \mp\un)$, $\ker(F_{22}\mp\un_1)=(\un-\Pi)\ker(F\mp\un)$.\\
iv) It follows from (\ref{dimprojpi}) that 
\begin{align}
\dim \ker (F_{22}- \un)+\dim& \ker (F_{22}+ \un)\nonumber\\
&=\dim \ker (F_{11}- \un)+\dim \ker (F_{11}+ \un)+|D|-2|V|.
\end{align}
\end{rem}
The proof of Theorem \ref{findim} is given in Appendix A.

\medskip

As an illustration of this result, consider for $\theta\in (-\pi, \pi]$ and $\alpha \in (-\pi, \pi]$
\be
U=\begin{pmatrix} -\cos(\theta) & e^{i\alpha}\sin(\theta) & 0 \cr \sin(\theta) & e^{i\alpha}\cos(\theta) & 0 \cr 0 & 0 &e^{i\alpha}
\end{pmatrix} \ \text{and} \ F=\begin{pmatrix} -\cos(\theta) & \sin(\theta) & 0 \cr \sin(\theta) & \cos(\theta) & 0 \cr 0 & 0 &1
\end{pmatrix}
\ee
which satisfy $F=F^*=F^{-1}$ and 
$
U=F(\Pi+e^{i\alpha}(\un-\Pi)),
$
with
\be
\Pi=\begin{pmatrix} 1 &0 & 0 \cr 0 &0 & 0 \cr 0 & 0 &0
\end{pmatrix}=\Pi^2=\Pi^*.
\ee
We compute for $\theta\neq 0$ and $\alpha\neq \pi$,
\be
\sigma(F_{11})=\{-\cos(\theta)\}, \ \sigma(F_{22})=\{\cos(\theta), 1\}, \ \sigma(U)=\{e^{i\alpha}, \lambda_+, \lambda_-\}
\ee
where $\{ \lambda_+, \lambda_-\}=\phi_\alpha^{-1}(\{-\cos(\theta)\})$. 
For $\theta\not \in \{ 0, \pi\}$ and $\alpha= \pi$, we have
\be
\sigma(F_{11})=\{-\cos(\theta)\}, \ \sigma(F_{22})=\{\cos(\theta), 1\}, \ \sigma(U)=\{-1, -e^{i\theta}, -e^{-i\theta}\},
\ee
while for  $\theta= \pi$ and $\alpha= \pi$, $+1$ is a double eigenvalue of $U$:
\be
\sigma(F_{11})=\{1\}, \ \sigma(F_{22})=\{-1, 1\}, \ \sigma(U)=\{-1, 1\}.
\ee

When $\theta =0$, $+1$ is a double eigenvalue of $F_{22}$, $e^{i\alpha}$ is a double eigenvalue of $U$ and for $\alpha\neq \pi$,
\be
\sigma(F_{11})=\{-1\}, \ \sigma(F_{22})=\{1\}, \ \sigma(U)=\{e^{i\alpha}, -1\},
\ee
while for $\theta=0$ and $\alpha=\pi$ we have $U=-\un$ and
\be
\sigma(F_{11})=\{-1\}, \ \sigma(F_{22})=\{1\}, \ \sigma(U)=\{ -1\}.
\ee

\medskip

We now turn to a description of the spectral mapping in a more abstract infinite dimensional framework by means of the Feshbach-Schur map. Keeping the same notation and conventions as above, the starting point is again
\be\label{defuabs}
U=F(\Pi+\e^{i\alpha}(\un-\Pi)), \ \text{with} \ F=F^*=F^{-1}, \ \e^{i\alpha}\neq 1,
\ee
defined on $\cH=\Pi \cH\oplus (\un -\Pi)\cH=\cH_1\oplus \cH_2$, where $\Pi=\Pi^*=\Pi^2\neq0$ and $\dim \cH =\infty$ is allowed. In block form we have for any $z\in \C$
\be\label{defblocU}
U-z=\begin{pmatrix} F_{11} -z& \e^{i\alpha}F_{12} \\ F_{21} & \e^{i\alpha}F_{22}-z\end{pmatrix}.
\ee
We consider the two Schur complements or Feshbach maps defined as linear maps on $\cH_1$, respectively $\cH_2$, by
\begin{align}\label{deftildes}
\tilde S_{1}(z)&=(F_{11} -z)- \e^{i\alpha}F_{12}(\e^{i\alpha}F_{22}-z)^{-1}F_{21}, \ \text{for } \ z\in \rho(\e^{i\alpha}F_{22})\nonumber \\
\tilde S_{2}(z)&=(\e^{i\alpha}F_{22} -z)- \e^{i\alpha}F_{21}(F_{11}-z)^{-1}F_{12},  \ \text{for } \ z\in \rho(F_{11}),
\end{align}
{where $\rho(A)=\C\setminus\sigma(A)$ denotes the resolvent set of an operator $A\in {\cal B}({\cal H})$.}
The isospectrality of the Feshbach-Schur maps states that
\begin{align}
&\text{for } z\in \rho(\e^{i\alpha}F_{22}), \ z\in \rho(U)  \Leftrightarrow 0\in \rho(\tilde S_1(z)), \\
&\text{for } z\in \rho(F_{11}), \ z\in \rho(U) \Leftrightarrow 0\in \rho(\tilde S_2(z)),
\end{align}
see {\it e.g.} \cite{BFS} Section IV, or the proof of Proposition 3.15.11 in \cite{S}. 

In our setup, the Feshbach-Schur method yields the following
\begin{thm}\label{smuf} Let the unitary operator $U$ be defined by (\ref{defuabs}), and the maps $F_{jj}$ on $\cH_j$, $j=1,2$ be defined by the block decomposition \eqref{defblocU} for $z=0$. With $\ffi_\alpha$ defined by (\ref{fial}), we have \\
for $\lambda\in{\mathbb S}^1\setminus\{\pm\e^{i\alpha} \}$,
\be
\lambda\in \sigma(U)\ \Leftrightarrow \ \ffi_\alpha(\lambda)\in \sigma(F_{11}),
\ee
for $\lambda\in{\mathbb S}^1\setminus\{\pm1\}$,
\be
\lambda\in \sigma(U) \ \Leftrightarrow \ \ffi_\alpha(\lambda)\in -\sigma(F_{22}).  
\ee
\end{thm}
\begin{rem}
i) For the values $\{\pm1, \pm\e^{i\alpha}\}$, we have, as in the finite dimensional case,
\be
\ker (U\pm \un) = \ker (F_{11}\pm \un_1),  \
\ker (U\pm e^{i\alpha}\un)=\ker (F_{22}\pm\un_2), 
\ee
for $\alpha\neq \pi$ while,  for $\alpha=\pi$,  
\be
\ker (U\pm \un) = \ker (F_{11}\pm \un_1)+\ker (F_{22}\mp\un_2).
\ee

ii) The spectra of $F_{11}$ and $F_{22}$ are related by $ \sigma(F_{11})\cup\{\pm1\}= \sigma(-F_{22})\cup\{\pm1\}$, as the proof below shows. The values $\pm1$ may or may not belong to one or both spectra.
\end{rem}
\begin{proof}:
We first make use of the properties of $F$ which imply
\be\label{invfz0}
(F-z)^{-1}=\frac{F+z}{1-z^2}, \ z\in \C\setminus\{\pm1\},
\ee
to reduce the expressions of the Feshbach-Schur maps $\tilde S_1(z)$, respectively $\tilde S_2(z)$ to a simple expression in terms of the resolvent of $F_{11}$, respectively $F_{22}$.

The Schur complements of the diagonal blocs of the operator $F-z$ read
\begin{align}\label{scF}
S_1(z)&=(F_{11}-z)-F_{12}(F_{22}-z)^{-1}F_{21}\ \text{for }\ z\in \rho(F_{22})\\
S_2(z)&=(F_{22}-z)-F_{21}(F_{11}-z)^{-1}F_{12}\ \text{for }\ z\in \rho(F_{11}),
\end{align}
defined as operators of $\cH_1$, respectively $\cH_2$.
Consequently, considering the block expression of the resolvent $(F-z)^{-1}$ 
\begin{align}
(F-z)^{-1}&=\begin{pmatrix} S_{1}^{-1}(z) & * \\ * & *\end{pmatrix},\ \text{for} \ z\in \rho(F_{22})\setminus\{\pm1\} \nonumber \\
(F-z)^{-1}&=\begin{pmatrix}* & * \\ * &  S_{2}^{-1}(z)\end{pmatrix},\ \text{for} \ z\in \rho(F_{11})\setminus\{\pm1\},
\end{align}
where the isospectrality property ensures the Schur complements $S_j(z)$, $j=1,2$, have bounded inverses for the values of $z$ considered. Identity (\ref{invfz0}) yields
\be
S_1(z)^{-1}=\frac{(F_{11}+z)}{1-z^2}, \ \text{for }\ z\in \rho(F_{22})\setminus\{\pm1\},
\ee
which implies that $z\in \rho(-F_{11})\setminus\{\pm1\}$. By similar considerations on the second Schur complement, we get 
\be\label{oppospec}
\rho(F_{11})\setminus\{\pm1\}=\rho(-F_{22})\setminus\{\pm1\} \, \Leftrightarrow \, \sigma(F_{11})\cup\{\pm1\}= \sigma(-F_{22})\cup\{\pm1\}.
\ee
Therefore we deduce that 
\begin{align}
S_1(z)=(1-z^2)(F_{11}+z)^{-1}, \ \text{respectively} \ S_2(z)=(1-z^2)(F_{22}+z)^{-1},
\end{align}
for  $z\in \rho(-F_{11})\setminus\{\pm1\}$, respectively  $z\in \rho(F_{11})\setminus\{\pm1\}$.
Hence, together with (\ref{scF}), we can write
\begin{align}
-F_{12}(F_{22}-z)^{-1}F_{21}&=(1-z^2)(F_{11}+z)^{-1}-(F_{11}-z), \  z\in \rho(-F_{11})\setminus\{\pm1\}\nonumber\\
-F_{21}(F_{11}-z)^{-1}F_{12}&=(1-z^2)(F_{22}+z)^{-1}-(F_{22}-z), \ z\in \rho(F_{11})\setminus\{\pm1\},
\end{align}
which we substitute in (\ref{deftildes}) to get with (\ref{oppospec})
\begin{align}\label{simschur}
\tilde S_{1}(z)&=(1-z^2\e^{-2i\alpha})(F_{11} +z\e^{-i\alpha})^{-1}+z(\e^{-i\alpha}-1), \ \e^{-i\alpha}z\in \rho(-F_{11})\setminus\{\pm1\}\nonumber \\
\tilde S_{2}(z)&=\e^{i\alpha}\big\{(1-z^2)(F_{22} +z)^{-1}+z(1-\e^{-i\alpha})\big\},  \ z\in \rho(F_{11})\setminus\{\pm1\}.
\end{align}
At this point we invoke the spectral mapping theorem for the self-adjoint operators $F_{11}$ and $F_{22}$ to determine the values of $z$ such that $0\in\rho(\tilde S_{j}(z))$, $j=1,2$.

We have that $\tilde S_{2}(z)$ is boundedly invertible for $ z\in \rho(F_{11})\setminus\{\pm1\}$ if and only if
 \be
-\frac{z(1-\e^{-i\alpha})}{1-z^2}\in \rho((F_{22} +z)^{-1}),
 \ee
 where 
 \be
 \sigma((F_{22} +z)^{-1})=\big\{(\nu+z)^{-1}, \, \nu\in \sigma (F_{22})\subset [-1,1]\big\}.
 \ee
 Hence,  for $ z\in \rho(F_{11})\setminus\{\pm1\}$, $0\in \rho(\tilde S_{2}(z))$ iff
 \be
 -\frac{z(1-\e^{-i\alpha})}{1-z^2}\neq \frac{1}{\nu+z}, \ \forall\nu\in \sigma (F_{22}).
 \ee 
 In other words, $\lambda\in {\mathbb S}^1\setminus \{\pm1\}$ belongs to $\sigma (U)$ iff $\ffi_\alpha(\lambda)=-\nu\in \sigma(F_{22})$, recall (\ref{fial}).

 An analogous computation based on the first equation (\ref{simschur}) yields $\lambda \in {\mathbb S}^1~\setminus\{\pm\e^{i\alpha}\}$ belongs to $\sigma(U)$ iff $\ffi_\alpha(\lambda)=\mu \in \sigma(F_{11})$. 
\end{proof}\qed

Coming back to $\cH=l^2(D)$,  we now relate the compression $F_{11}$ defined on $\cH_1=\Pi \, l^2(D)$ to an operator $T$ defined on $l^2(V)$, following \cite{HKSS2, HSS}. Recall that $\dim \Pi=|V|=\dim l^2(V)$.
\medskip

Let $\{ |x\ket \}_{x\in V}$ denote a fixed orthonormal basis of $l^2(V)$, where $|x\ket$ is a vector attached to the vertex $x\in V$, and recall the convention
$\omega(x)=\sum_{y\sim x}\omega_y(x)|xy\ket \in \cH_x^{\rm I}\subset l^2(D)$, where $\|\omega(x)\|=1$, for all $x\in V$.
We introduce
\be\label{defrrs}
R=\sum_{x\in V}|x\ket\bra \omega(x)| :  l^2(D) \ra  l^2(V), \ \text{and} \ 
R^*=\sum_{x\in V}| \omega(x)\ket\bra x| :  l^2(V) \ra  l^2(D),
\ee
which are readily shown to satisfy
\begin{align}
R^*R&=\sum_{x\in V}| \omega(x)\ket\bra \omega(x)| = \Pi :  l^2(D) \ra  l^2(D) \label{idpi}\\
RR^*&=\sum_{x\in V}|x\ket\bra x| =\un:  l^2(V) \ra  l^2(V) \label{idv} \\
&R\Pi=R\,  , \ R^*=\Pi R^*,  \ \text{and} \\
&\|R\|=\|R^*\|=1,
\end{align}
where $\|\cdot \|$ is the operator norm. The operator $R$ is  the {\it boundary operator} of \cite{HSS}.
We define
\be\label{defT}
T=RFR^*=RU_{\alpha}R^*: l^2(V)\ra l^2(V).
\ee
Thanks to the properties of $R$, $T=T^*$ is a contraction and we have
\be\label{tdfdt}
T=RF_{11}R^*,  \ R^*TR=F_{11}.
\ee
The matrix elements of $T$ read for $x, y\in V$,
\be
T_{xy}=\bra x |T y\ket=\overline{\omega_y(x)}\omega_x(y),
\ee
so that in case $\omega_y(x)=\frac{1}{\sqrt{d_x}}$, $\forall y\sim x$, $T$ is the renormalized adjacency matrix of the graph $G$.
\medskip

The spectral properties of $T$ and $F$ are related by the 
\begin{prop}
With $\un_1$ denoting the identity on $\cH_1=\Pi \, l^2(D)$, $\un_V$ the identity on $l^2(V)$, we have for all $z\in \rho(T)\cap \rho(F_{11})$,
\begin{align}
&R^*(T-z\un_V)^{-1}R=(F_{11}-z\un_{1})^{-1}\label{rtrf}\\
&(T-z\un_V)^{-1}=R(F_{11}-z\un_{1})^{-1}R^*.\label{trfr}
\end{align}
Consequently $\sigma(T)=\sigma(F_{11})$ and for $\lambda\in [-1,1]$ a common eigenvalue of $T$ and $F_{11}$, the respective spectral projectors $P^\lambda_T$ and $P^\lambda_1$, satisfy
\be\label{pdproj}
P^\lambda_1=R^*P^\lambda_T R, \ P^\lambda_T= RP^\lambda_1 R^*,
\ee
and 
\be
\dim P^\lambda_1=\dim P^\lambda_T.
\ee
\end{prop}
\begin{rem} i) The results above are also true for $G$ infinite, {\it i.e.} $|V|=\infty$, $|D|=\infty$. The series defining $R$, $R^*$, $R^*R$ and $R R^*$ are meant in the strong sense in this case. \\
ii) Together with Theorems \ref{findim} and  \ref{smuf}, we get a spectral mapping result between $U$ and $T$. In particular,
\be
F_{11}\psi_1=\mu\psi_1 \ \Leftrightarrow \ T\ffi=\mu\ffi, \ \text{with} \ \ffi=R\psi_1 \ \Leftrightarrow \ \psi_1=R^*\ffi.
\ee
\end{rem}
\begin{proof}:
We first note that thanks to (\ref{idpi}) and (\ref{tdfdt})
\be
T^2=RF_{11}R^*RF_{11}R^*=RF_{11}\Pi F_{11}R^*=RF_{11}^2R^*,
\ee
so that by induction $T^k=RF_{11}^kR^*$, $\forall k\in \N$.
Then, for any $z\in \C$ such that $|z|>1$, we have the convergent Neumann series
\be
(T-z\un_V)^{-1}=-\frac{1}{z}\sum_{k=0}^\infty \frac{T^k}{z^k}=-\frac{1}{z}\sum_{k=0}^\infty R\frac{F_{11}^k}{z^k}R^*,
\ee
where $F_{11}^0=\Pi|_{\ch_1}=\un_{1}$.
Therefore (\ref{trfr}) holds for all $z\in \rho(T)\cap \rho(F_{11})$. Now (\ref{idv}) and (\ref{tdfdt}) imply
\be
F_{11}^2=R^*TRR^*TR=R^*T\un_V TR=R^*T^2R,
\ee
so that by induction $F_{11}^k=R^*T^kR$, $\forall k\in \N$. As above, we consider the Neumann series for $(F_{11}-z\un_1)^{-1}$ to deduce that (\ref{rtrf})  holds for all $z\in \rho(T)\cap \rho(F_{11})$.

Now, (\ref{trfr}) and  (\ref{rtrf}) show that the resolvents are singular on the same set, so that $\sigma(T)=\sigma(F_{11})$. The relations (\ref{pdproj}) between spectral projectors associated with $\lambda$ follow from their expressions via the Riesz formula for isolated eigenvalues, or via the strong limit of $-i\epsilon(A-\lambda -i\epsilon)^{-1}$ as $\epsilon\ra 0$, for $A=T$ and $A=F_{11}$, for embedded eigenvalues, see {\it  e.g.} Theorem 6.10 in \cite{HiSi}. The equality of their dimensions is a consequence of the cyclicity of the trace and (\ref{idv}).
\end{proof}\qed

\section{Scattering Open Quantum Walks on a Graph}\label{SOQWG}
\subsection{Definition and Properties}

In this section, we define an open QW on a graph $G$, denoted $\Phi_{\cS}$, parameterized by the set of scattering matrices $\cS=\{S(x)\}_{x\in V}$, in the same spirit as the unitary QW $U_\cS$. An open QW is a completely positive trace preserving (CPTP) linear map $\Phi_{\cS}$ on $\cT(l^2(D))$, the set of trace class operators acting on $l^2(D)$, also known as a quantum channels. 

\medskip

Let us briefly recall these notions in our framework, referring the reader to the books \cite{A, AJP, AL, Kr, Sc, Wa}, for more information.  Let $\cH$ be a separable Hilbert space over $\C$ and $\cB(\cH)$ the $C^*$ algebra of all bounded operators on $\cH$ equipped with the operator norm. A positive map $\Phi : \cB(\cH)\ra \cB(\cH)$ is characterized by $\Phi (A)\geq 0$ for all  $A\geq 0$, $A\in \cB(\cH)$, while it is called $n-$positive if 
$\Phi\otimes \un_n: \cB(\cH)\otimes M_n(\C)\ra \cB(\cH)\otimes M_n(\C)$ is positive, where $M_n(\C)$ is the set of square complex matrices on $\C^n$ and $\un_n$ denotes the identity map on $M_n(\C)$. A map $\Phi$ is completely positive (CP), if it is $n-$positive for all $n\in \N$ and it is called unital if $\Phi(\un)=\un$. 

Let  $\cT(\cH)\subset \cB(\cH)$ be the set of trace class operators on $\cH$, which is a Banach space when endowed with the trace norm $\|\cdot \|_1$. For a map $\Phi: \cT(\cH)\ra \cT(\cH)$, the notions of (complete) positivity are defined as above and such a map is called trace preserving (TP) if $\tr \,\Phi(A)=\tr A$ for all $A\in \cT(\cH)$. The set of CPTP maps is convex and the composition of CPTP maps is CPTP. Kraus's Representation Theorem states that any CPTP map $\Phi$ on $\cT(\cH)$ can be represented by means of a non unique set of {\it Kraus operators} $\{K_j\}_{j\in J}$, $K_j\in \cB(\cH)$, with index set $J$ at most countable, as follows:
\be\label{krausrep}
\Phi(A)=\sum_{j\in J} K_j A K_j^*, \ \ \text{with } \ \sum_{j\in J} K_j^* K_j=\un,
\ee
where the convergence of the second sum is understood in the strong sense, implying the first sum converges in the trace norm sense. Conversely, any map $\Phi$ defined by (\ref{krausrep}) yields a CPTP map on $\cT(\cH)$. Moreover such maps have operator norms equal to one:
\be\label{bddphi}
\sup_{A\in \cT(\cH)\atop A\neq 0}\frac{\|\Phi(A)\|_1}{\|A\|_1}=1.
\ee

If $\dim \cH <\infty$, the index set $J$ can be chosen so that $|J|\leq (\dim \cH)^2$. In this case, CPTP maps can be unital since $\un \in \cT(\cH)$;  consider the map defined by $\Phi(A)=UAU^*$,   where $U$ is a unitary operator on $\cH$, for example.

Furthermore, since $\cB(\cH)$ is the dual of $\cT(\cH)$ with duality bracket
 \be
 \cB(\cH)\times  \cT(\cH) \ni (B,A) \mapsto  \tr (BA )\in \C,
 \ee 
 any map $\Phi: \cT(\cH)\ra \cT(\cH)$ admits an adjoint $\Phi^*:\cB(\cH)\ra \cB(\cH)$ defined by 
\be\label{traceduality}
\tr (B \Phi(A))=\tr (\Phi^*(B) A), \ \ \forall \ (B,A) \in \cB(\cH)\times \cT(\cH).
\ee
By Kraus's Theorem, one also gets that the adjoint of a CPTP map given by (\ref{krausrep}) has the form 
\be\label{defadjkraus}
\Phi^*(B)=\sum_{j\in J} K_j^* B K_j, \ \ \text{with } \ \sum_{j\in J} K_j^* K_j=\un,
\ee
where the first sum converges in the strong sense, a consequence of the second one converging in the strong sense. The map $\Phi^*$ is CP and unital and, moreover, any CP unital map on $\cB(\cH)$ has the form (\ref{defadjkraus}). Also $\Phi^*$ has operator norm equal to one
\be
\sup_{B\in \cB(\cH)\atop B\neq 0}\frac{\|\Phi^*(B)\|}{\|B\|}= 1.
\ee

We finally introduce the set of states or density matrices, a convex subset of $\cT(\cH)$, by 
\be
\cD\cM(\cH)=\{\rho\in \cT(\cH), \ { \rm s.t.} \ \rho=\rho^*\geq 0 \ {\rm and } \ \tr (\rho)=1\}, 
\ee
which is invariant under CPTP maps.
\medskip

Given a graph $G$ and  a set of scattering matrices $\cS=\{S(x)\}_{x\in V}$, with the notions introduced above, we first define a set of  Kraus operators labelled by $x\in V$:
\be\label{kraus}
K(x)=\sum_{y\sim x \atop z\sim x}S_{zy}(x)|zx\ket\bra xy| \ \ \mbox{s.t.} \ \  K(x)=K(x)P_x^{\rm I} =P_x^{\rm O}K(x), \ \ \forall x\in V.
\ee
Note that $U_\cS=\sum_{x\in V}K(x)$, see (\ref{defug}),
so that,
\be\label{unitchannel}
U_\cS A U_\cS{}^*=\sum_{x, x' \in V}K(x)AK^*(x'), \ \ \forall A\in \cB(l^2(D)).
\ee
 in the strong sense, by Remark \ref{remconvu}.
\begin{lem}\label{propK}
\be
K^*(x)K(x)=P_x^{\rm I}, \ \ K(x)K^*(x)=P_x^{\rm O},
\ee
and 
\be\label{convkraus}
\sum_{x\in V}K^*(x)K(x)=\sum_{x\in V}K(x)K^*(x)=\un,
\ee 
with strong convergence in case $|V|=\infty$.
\end{lem}
\proof:
The first identities are proven along the lines of  Lemma \ref{unitUS} while the second identities hold since  
$\{P_x^{\rm \#}\}_{x\in V}$, $\#\in \{\rm I, O\}$ form resolutions of the identity.
\qed

We can now proceed with definitions.
\begin{definition}
Given the set of Kraus operators $\{K(x)\}_{x\in V}$, the open QW on the graph $G$ is defined by the map $\Phi_{\cS}$ on $\cT(l^2(D))$
\be\label{defphi}
 \Phi_{\cS}(A)=\sum_{x\in V}K(x)AK^*(x), \ \ \forall A\in \cT(l^2(D)).
\ee
We further introduce two maps on $\cT(l^2(D))$ 
\be
D^{\#}(A)=\sum_{x\in V}P_x^{\#}AP_x^\#, \ \ \#\in \{{\rm I, O}\}.
\ee
\end{definition}
\begin{prop}\label{firstphi}
The maps $\Phi_{\cS}$ and $D^{\#}$, $\#\in \{{\rm I, O}\}$ are CPTP, and they all can be extended to unital maps on $\cB(l^2(D))$. \\
Moreover, 
\be\label{dphy}
\Phi_{\cS}=\Phi_{\cS}\circ D^{\rm I}=D^{\rm O}\circ \Phi_{\cS} = D^{\rm O}\circ \Phi_{\cS} \circ D^{\rm I},
\ee
and the link with $U_\cS$ reads
\be\label{uphi}
\Phi_{\cS}(\cdot )=D^{\rm O}(U_\cS \cdot U_\cS{}^*)=U_\cS D^{\rm I}(\cdot ) U_\cS{}^*.
\ee
\end{prop}
\proof:  Using criterion (\ref{krausrep}) and Lemma \ref{propK}, 
we get that $\Phi_{\cS}$ and $D^{\#}$ are CPTP maps. By criterion (\ref{defadjkraus}) with $K^*(x)$ in place of $K(x)$ and  Lemma \ref{propK} again, we see that these three maps can 
be extended to $\cB(l^2(D))$ and are unital.

The second statement follows directly from properties (\ref{kraus}) and (\ref{projio}), whereas the last statement stems from formula
(\ref{oui}) for $U_\cS A U_\cS{}^*$. Applying $D^{\rm O}$ to this expression and replacing $A$ by $D^{\rm I}(A)$ yields (\ref{uphi}) thanks to (\ref{kraus}) and (\ref{projio}) again.
\qed
\begin{rem}
i) The operators $D^{\rm \#}$ are projectors.\\
ii) In finite dimensions, unital quantum channels are known to be affine combinations of unitary quantum channels \cite{MW}; {allowing  negative coefficients, as shown in the examples in \cite{LS}. Moreover, unital quantum channels are} entropy non-decreasing:
\be
S(\rho)\leq S(\Phi_{\cS}(\rho)), \ \text{for all } \ \rho \in  {\cal DM}(l^2(D)).
\ee
Here $S: {\cal DM}(l^2(D))\ra [0,\ln(\dim l^2(D))]$ is the Von Neumann entropy $S(\rho)=-\tr(\rho\ln \rho)$.
\end{rem}

Note that thanks to (\ref{projio}),  $ D^{\rm O}$ and $D^{\rm I}$ commute and the CPTP map on $\cT(l^2(D))$, and its extension to $\cB(l^2(D))$, defined by
\be\label{defdiag}
\text{Diag}(\cdot)=D^{\rm O}\circ D^{\rm I}(\cdot)=D^{\rm I}\circ D^{\rm O}(\cdot)=\sum_{x, y \in V \atop x\sim y} |xy\ket\bra xy| \cdot |xy\ket\bra xy|,
\ee
is a projector as well. 
To address the dynamics, we consider the restriction
\be\label{defdiagphi}
\Phi^{\rm Diag}={\rm Diag}\circ \Phi_{\cS} |_{{\rm Diag} (\cT(l^2(D))) },
\ee
where we dropped mention to $\cS$ in order to simplify the notation.

\begin{cor}\label{cordynphi}
The map  $\Phi^{\rm Diag}$ determines the dynamics in the sense that for any $n\geq 2$
\be\label{dyndiag}
\Phi_{\cS}^n= \Phi_{\cS}\circ  {(\Phi^{\rm Diag})}^{n-2}\circ  {\rm Diag}\circ \Phi_{\cS} 
\ee
\end{cor}
\proof:
Proposition \ref{firstphi}  yields
\be
\Phi_{\cS}^n=\Phi_{\cS}\circ (D^{\rm I}\circ D^{\rm O}\circ \Phi_{\cS})^{n-2}D^{\rm I}\circ D^{\rm O}\circ \Phi_{\cS},
\ee
and  \eqref{defdiag} with the fact  that ${\rm Diag}$ is a projector proves the statement.
\qed
\begin{rem}
i) When acting on $\cT(l^2(D))$, $\Phi_{\cS}^n$, $n\geq 2$, leaves the subspace ${\rm Diag}\cT(l^2(D))$ only at the last time step. The essential part of the dynamics is thus driven by $\Phi^{\rm Diag}$ and takes place inside ${\rm Diag}\cT(l^2(D))$. \\
ii) The extension of $\Phi^{\rm Diag}$ to $\cB(l^2(D))$ reads
\be\label{matelphid}
\Phi^{\rm Diag}(\cdot) =\sum_{x\in V}\sum_{y\sim x \atop z\sim x} |S_{zy}(x)|^2|zx\ket\bra xy| \cdot |xy\ket\bra zx|
\ee
with Kraus operators $K_{zy}(x)=S_{zy}(x)|zx\ket\bra xy|$ mapping $\cH_x^{\rm I}$ to $\cH_x^{\rm O}$ and satisfying
\be
\sum_{x\in V}\sum_{y\sim x \atop z\sim x} K_{zy}^*(x)K_{zy}(x)=\sum_{x\in V}P_x^{\rm I}=\un, \ \ \  \sum_{x\in V}\sum_{y\sim x \atop z\sim x} K_{zy}(x)K^*_{zy}(x)=\sum_{x\in V}P_x^{\rm O}=\un,
\ee
with convergence in the strong sense.
\end{rem}

%%%%%%%%%%%%%%

When restricted to $\cD\cM(l^2(D))$, the set of density matrices, the CPTP map $\Phi_{\cS}$  describes the one time step quantum evolution of a quantum walker in a mixed state, now characterized by a density matrix. Again, the dynamical system $(\Phi_{\cS}^n)_{n\in \N}$ on $\cT(l^2(D))$ or $\cD\cM(l^2(D))$, and the related spectral properties of the CPTP map $\Phi_{\cS}\in \cB(\cT(l^2(D)))$ are of interest. We will address the spectrum of $\Phi_{\cS}$, making use of the specifics of our construction. 

\medskip

{\bf Quantum Mechanical Interpretation:} Similarly to the case of pure states considered in Section \ref{purestate}, quantum mechanical considerations induce natural time dependent distribution probabilities in the framework of density matrices. Given an initial state $\rho_0\in \cD\cM(l^2(D))$, the state at time $n$, is $\rho_n=\Phi_{\cS}^n(\rho_0)\in \cD\cM(l^2(D))$.  The quantum mechanical probability to find the corresponding quantum walker at time $n$ on the directed edge $(xy)$ of $G$, by a measurement of its position, is 
\be
\P_n^{\rho_0}(xy)=\tr(|xy\ket\bra xy| \Phi_{\cS}^n(\rho_0))=\tr(|xy\ket\bra xy| \rho_n)=\bra xy | \rho_n xy\ket.
\ee
The probability to find the quantum walker at time $n\in \N$ in the subspace $\cH_x^{\rm I}$ or on the vertex $x\in V$, by a measurement of its position, is given by 
\be\label{prohix}
\Q_n^{\rho_0}(x)=\tr(P_x^{\rm I} \Phi_{\cS}^n(\rho_0))=\tr(P_x^{\rm I}\rho_n)=\sum_{y\sim x}\P_n^{\rho_0}(xy).
\ee 

\medskip

Proposition \ref{firstphi} yields the following operational interpretation of the quantum channel defined by $\Phi_{\cS}$. Given an initial quantum state $\rho_0\in \cD\cM(l^2(D))$, with $\rho_0>0$, one measures the position of the quantum walker on the vertices $V$. The probability to get the outcome $x\in V$ by this measurement is $\Q_0^{\rho_0}(x)=\tr ( P_x^{\rm I} \rho_0)>0$ and the wave packet reduction postulate says that after getting the outcome $x$ in the measurement process, the state immediately becomes $P_x^{\rm I}\rho_0 P_x^{\rm I}/\tr (P_x^{\rm I}\rho_0)$. The expectation value $\bar \rho$ of the state obtained after a position measurement is then 
\be
\bar \rho = \sum_{x\in V} \frac{P_x^{\rm I}\rho_0 P_x^{\rm I}}{\tr (P_x^{\rm I}\rho_0)} \Q_0^{\rho_0}(x)= \sum_{x\in V}P_x^{\rm I}\rho_0P_x^{\rm I}=D^{\rm I}(\rho_0),
\ee
inducing decoherence. 
Hence, the action of the quantum channel $\Phi_{\cS}$ on the state $\rho_0$ results in first taking its expectation value with respect to a position measurement, 
so that the state becomes $D^{\rm I}(\rho_0)$, followed by evolving the result by means of the unitary quantum channel $U_\cS\cdot U_\cS{}^*$. 

\subsection{Quantum Trajectory}

We strengthen this interpretation by considering the {\it Quantum Trajectory} associated with the measurement of the position on $V$. {See \cite{H} for quantum measurement theory, \cite{KM1, KM2} for the first analyses of quantum trajectories and }\cite{BJPP} for a recent account of this topic.  
The iterative protocol is as follows (dropping unessential indices from the notation and writing $\cH$ for $l^2(D)$).
\begin{itemize}
\item At time $0$, the state is $\rho_0\in \cD\cM(\cH)$. We measure the position on $V$ and  get an outcome $x_1\in V$, with probability $\tr (P_{x_1} \rho_0 )$ and evolve the reduced state by $U \cdot U^*$ to get at time $1$
\be
\rho_1= U P_{x_1}\rho_0 P_{x_1} U^*/\tr (P_{x_1}\rho_0).
\ee
Denoting $\P_1(x_1)=\tr (P_{x_1} \rho_0 )$ the probability to get the outcome $x_1\in V$, we get
\be
\rho_1 \P_1(x_1)=U P_{x_1}\rho_0 P_{x_1} U^*.
\ee
\item At time $1$, the state is $\rho_1\in \cD\cM(\cH)$. We measure  the position on $V$,  get an outcome $x_2\in V$ with probability 
\be
\tr (P_{x_2} \rho_1 )= \tr (P_{x_2} U P_{x_1}\rho_0 P_{x_1} U^*P_{x_2})/\tr (P_{x_1}\rho_0)
\ee 
and evolve the reduced state  by $U \cdot U^*$ to get at time $2$
\be
\rho_2= U P_{x_2}\rho_1 P_{x_2}U^*/\tr (P_{x_2} \rho_1).
\ee
Here, $\tr (P_{x_2} \rho_1 )=\P(x_2|x_1)$ is the probability to get the outcome  $x_2$, given the first outcome is $x_1$, so that the probability to get the sequence of outcomes $(x_1, x_2)\in V^2$ reads
\be
\P_2(x_1, x_2)= \tr (P_{x_2} U P_{x_1}\rho_0 P_{x_1} U^*P_{x_2}).
\ee
Hence we have
\be
\rho_2 \P_2(x_1, x_2)=U P_{x_2}\rho_1 P_{x_2}U^*.
\ee
\item At time $j$, $j\geq 1$, the state is $\rho_j\in \cD\cM(\cH)$. We measure  the position on $V$,  get an outcome $x_{j+1}\in V$, with probability $\tr (P_{x_{j+1}} \rho_{j} )$ and evolve the reduced state  by $U \cdot U^*$ to get at time $j+1$
\be\label{inducstate}
\rho_{j+1}= U P_{x_{j+1}}\rho_{j} P_{x_{j+1}}U^*/\tr (P_{x_{j+1}}\rho_j).
\ee
By induction, the probability to get the sequence of outcomes $(x_1, x_2, \dots, x_{j+1})\in V^{j+1}$ is 
\be\label{defpqj}
\P_{j+1}(x_1, x_2, \dots, x_{j+1})= \tr (P_{x_{j+1}}U\dots  P_{x_2} U P_{x_1}\rho_0 P_{x_1} U^*P_{x_2}\dots U^* P_{x_{j+1}}),
\ee
and 
\be\label{426}
\rho_{j+1}\P_{j+1}(x_1, x_2, \dots, x_{j+1})=U P_{x_{j+1}}\dots  P_{x_2} U P_{x_1}\rho_0 P_{x_1} U^*P_{x_2}\dots P_{x_{j+1}}U^*.
\ee
\end{itemize}
The sequence of outcomes and states $\{(x_j, \rho_j)\}_{1\leq j\leq n}\in (V\times  \cD\cM(\cH))^n$, $n\in \N^*$, is a realization of a quantum trajectory of length $n$ of the quantum walker under repeated measurement of  the position on $V$, with corresponding probability 
\be\label{lawqtraj}
\P_{n}(x_1, x_2, \dots, x_{n})=\tr (P_{x_{n}}U\dots  P_{x_2} U P_{x_1}\rho_0 P_{x_1} U^*P_{x_2}\dots U^* P_{x_{n}})\geq 0.
\ee
\begin{rem} i) We did not pay attention to the possibility that some denominators $\tr(P_{x_{j+1}}\rho_{j})$ could vanish. This is harmless since  the quantities of interest, $\P_{n}(x_1, x_2, \dots, x_{n})$ and $\rho_n\P_{n}(x_1, x_2, \dots, x_{n})$, are always well defined.\\
ii) Note that $\P_j(x_1, x_2, \dots, x_{j})$ is indeed a probability on $V^j$ for all $j\in \N^*$: setting $\cU(\cdot)=U\cdot U^*$ and $D(\cdot ) = \sum_{x\in V}P_x\cdot P_x$ on $\cB(\cH)$ we get for each $j$,  
\be\label{prob1}
\sum_{x_1, x_2, \dots, x_{j}\in V^j}\P_j(x_1, x_2, \dots, x_{j})=\tr (D\circ \cU\dots \circ D\circ \cU\circ D(\rho_0))=\tr (\rho_0)=1,
\ee
as $\cU$ and $D$ are CPTP maps and $\rho_0\in \cD\cM(\cH)$.
\end{rem}

We are ready to provide a probabilistic interpretation of $\Phi_{\cS}^n$, for all $n\in \N^*$. 

\begin{prop} Let $\rho_0\in \cD\cM(l^2(D))$ be an initial state and  
consider the quantum trajectory of length $n\in \N^*$ corresponding to repeated measurements of the position on $V$ according to the protocol above. Let  
\be
\bar \rho_n=\sum_{x_1, x_2, \dots, x_{n}\in V^n} \rho_n \P_n(x_1, x_2, \dots, x_{n}) 
\ee
be the expectation value of the state obtained at time $n$ with respect to the probability distribution $\P_n$ on $V^n$ defined by (\ref{lawqtraj}).  We have
\be
\bar \rho_n=\Phi_{\cS}^n(\rho_0), \ \ \forall n\in \N^*. 
\ee
\end{prop}
\proof:
By definition and using (\ref{426}) (with the simpler notation above),
\be
 \bar \rho_n=\hspace{-.2cm} \sum_{x_1, x_2, \dots, x_{n}\in V^n} UP_{x_n}\dots UP_{x_1}\rho_0P_{x_1}U^*\dots P_{x_n}U^*=\cU \circ D\circ \cU\dots \circ D\circ \cU\circ D(\rho_0),
\ee
which yields the result thanks to (\ref{uphi}) in Proposition \ref{firstphi}. For $G$ infinite, the convergence takes place in trace norm, thanks to (\ref{krausrep}) and (\ref{bddphi}): let the positive map
\be
D_m(\cdot)=\sum_{x\in V\atop |x|\leq m}P_x\cdot P_x
\ee 
such that for $0\leq A\in \cT(l^2(D))$, we have $D_m(A)\leq D(A)$, $(D-D_m)(A)\leq D(A)$
 and  $(D-D_m)(A)\geq 0$ converges in trace norm to zero when $m\ra\infty$. The difference
 $ (\cU \circ D)^n(\rho_0)- (\cU \circ D_m)^n$ is a finite sum of trace class positive terms of the form
 \be\label{indivter}
 (\cU \circ D_m)^{n_1}\circ (\cU \circ (D-D_m))^{n_2}\circ \cdots  (\cU \circ D_m)^{n_{2p-1}} \circ (\cU \circ(D-D_m))^{n_{2p}}(\rho_0),
 \ee
 where $p$ is finite,  $\sum_{i=1}^{2p}n_i=n$, $n_i\geq 0$, and $n_{2k}\geq 1$, for some $k\leq p$. Using the inequalities above on each of the positive maps that compose \eqref{indivter}, except 
 $ (\cU \circ (D-D_m))^{n_{2k}}$, we get that its trace  is bounded above by
 \be
 \tr\Big((\cU \circ D)^{n'_1} (\cU \circ( D-D_m))^{n_{2k}} (\cU \circ D)^{n'_2}(\rho_0) \Big) \leq \tr\Big(( D-D_m)\circ (\cU \circ D)^{n'_2} (\rho_0)\Big),
 \ee
 where $n'_1, n'_2 \in \N$ are such that $n'_1+n'_2+n_{2k}=n$, and where we used that $\cU \circ D$, and $\cU$ are trace preserving. The last term vanishes as $m \ra\infty$, which ends the proof. 
\qed

{\bf Comparison with the open QW of  \cite{APSS1}. }
Let us finally show that open SQWs are distinct from the open QWs devised in \cite{APSS1}. To do so, we need to consider regular graphs $G$ with $d_x=d$ for all $x\in V$, which we further assume be finite, for simplicity. We have the identification $l^2(D)=\bigoplus_{x\in V}\cH_x^{\rm I}\simeq l^2(V)\otimes \C^d$. Let $ \{| \tau\ket \}_{\tau\in I} $ with $|I|=d$, be an ONB of $ \cH_x^{\rm I}\simeq \C^d$, and consider a labelling of the vertices so that the following identification holds, for each $x\in V$ fixed, 
\be 
|xy\ket \simeq |x\ket\otimes |\tau\ket, \ \ \text{where} \ \ {y\sim x}, \  {\tau\in I}.
\ee
Hence the identification of basis vectors of $\cB(l^2(D))\simeq \cB(l^2(V)\otimes \C^d)$
\be
|xy\ket \bra x' y'| \simeq | x\ket\bra x'|\otimes |\tau\ket\bra \tau' |. 
\ee 
The open QWs introduced in \cite{APSS1} are denoted by $\cM$, and act on density matrices on $l^2(V)\otimes \C^d$. Their defining characteristics is that they leave the following subset of states invariant
\be
\cV=\Big\{\rho\in \cD\cM(l^2(V)\otimes \C^d)\ | \ \rho=\sum_{x\in V}|x\ket\bra x|\otimes \rho(x), \ \rho(x)\in \cB(\C^d) \Big\},
\ee
see Corollary 2.5  in \cite{APSS1}.
Now, $|x y\ket\bra x y |\simeq |x\ket\bra x|\otimes |\tau\ket\bra \tau|\in \cV$, while the computation
\be
\Phi_{\cS}(|xy\ket \bra xy| )=\sum_{z\sim x \atop t\sim x}S_{z y}(x)\overline{S_{t y}(x)} |zx\ket\bra tx|
\ee
shows  $\Phi_{\cS}(|xy\ket \bra xy| )$ is not equivalent to an element of $\cV$, unless $S(x)$ {is a permutation matrix times a diagonal matrix of phases}.

\medskip

We now turn to the spectral properties of $\Phi_{\cS}$.

\section{Spectral and Dynamical Properties of $\Phi_{\cS}$}\label{sec:sdp}

We will mainly consider the finite dimensional case corresponding to finite graphs $G$ in this section, and will comment along the way on the infinite graph  case.

\medskip

For $G$ finite, we will not distinguish $\cT(l^2(D))$ and $\cB(l^2(D))$. Let us equip $\cB(l^2(D))$ with the Hilbert-Schmidt scalar product to make it a Hilbert space,
\be\label{hsip}
\bra A, B \ket_{HS}=\tr (A^* B), \ \ \forall A, B\in \cB(l^2(D)), \ \ \text{s.t.} \ \ \|A\|_{HS}=\sqrt{\tr(A^*A)}.
\ee
We denote the adjoint of a map $\Phi\in \cB(l^2(D))$ with respect to this scalar product by $\Phi^\dagger$. In particular, we have
\be
 \Phi_{\cS}{}^\dagger(\cdot)=\sum_{x\in V}K^*(x)\cdot K(x), \ \ (U_\cS \cdot U_\cS{}^*)^{\dagger}=U_\cS{}^*  \cdot U_\cS, \ \ D^{\#}=D^{\#}{}^{\dagger}, \ \ \#\in \{{\rm O,I}\},
\ee
which makes $D^\#$ and $\text{Diag}$, recall (\ref{defdiag}), orthogonal projectors on $\cB(l^2(D))$ and implies 
$
 \Phi_{\cS}{}^\dagger=D^{\rm I}\circ \Phi_{\cS}{}^\dagger \circ D^{\rm O}.
$
\begin{thm}\label{multipli}
Let $G$ be finite. The open SQW $\Phi_{\cS}$ is a partial isometry on $\cB(l^2(D))$ with $\ker \Phi_{\cS}=\ran (\un -D^{\rm I})$, where $\dim \ker \Phi_{\cS}=(\sum_{x\in V}d_x)^2-\sum_{x\in V}d_x^2.$\\
Moreover, $\sigma(\Phi_{\cS})\setminus\{0\}=\sigma\big({\rm Diag} \circ \Phi_{\cS}|_{{\rm Diag}\cB(l^2(D))}\big)\setminus\{0\}$, with identical geometric multiplicities.
\be
\sum_{\lambda\in\sigma(\Phi_{\cS})\setminus\{0\}}\dim \ker(\Phi_{\cS}-\lambda)\leq \rank{\rm Diag},
\ee
and $\Phi_{\cS}$ is not diagonalizable if $\exists x\in V$ such that $d_x\geq 2$.
\end{thm}
\proof:
Thanks to Proposition (\ref{firstphi}),
\be
\Phi_{\cS}=\Phi_{\cS}\circ D^{\rm I} \ \ \text{and} \ \  \Phi_{\cS}{}^\dagger \circ \Phi_{\cS}=D^{\rm I},
\ee
so that $\Phi_{\cS}((\un -D^{\rm I})(A))=0$ and  $\ran (\un -D^{\rm I})\subset \ker \Phi_{\cS}$. \\
Moreover, for any $A\in \cB(l^2(D))$, 
\be
\|\Phi_{\cS}(A) \|_{HS}^2=\bra A| \Phi_{\cS}{}^\dagger \circ \Phi_{\cS}(A)\ket_{HS}=\bra A|D^{\rm I}(A)\ket_{HS}=\| D^{\rm I}(A)\|_{HS}^2,
\ee
so that on $\ran (\un -D^{\rm I})^\perp=\ker (\un -D^{\rm I})=\ran D^{\rm I}$, $\Phi_{\cS}$ is an isometry and $\ker \Phi_{\cS}=\ran (\un -D^{\rm I})$.\\
Then, $\ran D^{\rm I}$ is spanned by $\{|xy\ket\bra xz|, \ x\in V, \ y\sim x, z\sim x \}$, and thus has dimension $\sum_{x\in V}d_x^2$, which yields the first statement.

Let us turn to the nonzero eigenvalues of  $\Phi_{\cS}$. If $\lambda\in\sigma(\Phi_{\cS})\setminus\{0\}$, there exists $A\in  \cB(l^2(D))\setminus \{0\}$ such that
$D^{\rm I}(A)\neq 0$ and
\begin{align}
 \Phi_{\cS}(A)=D^{\rm O}\circ \Phi_{\cS}\circ D^{\rm I}(A)=\lambda (D^{\rm I}(A)+(\un-D^{\rm I})(A)),
\end{align}
by Proposition \ref{firstphi}. Projecting onto $D^{\rm I}\cB(l^2(D))$ and $(\un-D^{\rm I})\cB(l^2(D))$, we get 
\begin{align}
&D^{\rm I}\circ D^{\rm O}\circ \Phi_{\cS}\circ D^{\rm I}(A)=\lambda D^{\rm I}(A), 
&(\un-D^{\rm I})(A)=\frac{1}{\lambda}(\un-D^{\rm I})\circ \Phi_{\cS}\circ D^{\rm I}(A).
\end{align}
The first equation implies that $D^{\rm I}(A)$ is an eigenvector associated with $\lambda\neq 0$ for the restriction $D^{\rm I}\circ D^{\rm O}\circ \Phi_{\cS}|_{D^{\rm I}\cB(l^2(D))}$ which, with \eqref{defdiag}, implies $D^{\rm I}(A)=\text{Diag}(A)$. Hence $\text{Diag}(A)$ is an eigenvector of the restriction $\text{Diag} \circ \Phi_{\cS}|_{\text{Diag}\cB(l^2(D))}$ associated with $\lambda\neq 0$. 

Conversely, if $\text{Diag}(A)$ is a an eigenvector of $\text{Diag} \circ \Phi_{\cS}|_{\text{Diag}\cB(l^2(D))}$ associated with an eigenvalue $\lambda\neq 0$, then
\be\label{eigvecnz}
B={\rm Diag}(A)+\frac{1}{\lambda}(\un-D^{\rm I})\circ \Phi_{\cS}\circ {\rm Diag}(A)
\ee
satisfies thanks to Proposition \ref{firstphi}, $\ker \Phi_{\cS}=\ran (\un -D^{\rm I})$, and $\un-{\rm Diag}=\un-D^{\rm I}\circ D^{\rm O}$
\begin{align}
 \Phi_{\cS}(B)&=\Phi_{\cS}({\rm Diag}(A))+\frac{1}{\lambda}\Phi_{\cS}((\un-D^{\rm I})\circ \Phi_{\cS}\circ {\rm Diag}(A))\nonumber\\
 &={\rm Diag}\circ\Phi_{\cS}({\rm Diag}(A))+(\un-D^{\rm I})\circ D^{\rm O}\circ\Phi_{\cS}({\rm Diag}(A))\nonumber\\
 &=\lambda \Bigg( {\rm Diag}(A)) + \frac{1}{\lambda}(\un-D^{\rm I})\circ \Phi_{\cS}\circ {\rm Diag}(A) \Bigg)= \lambda B.
\end{align}

Finally, since $\rank {\rm Diag}= \sum_{x\in V}d_x$, the total geometric multiplicity of $\Phi_{\cS}$ cannot exceed 
$
(\sum_{x\in V}d_x)^2-\sum_{x\in V}d_x^2+\sum_{x\in V}d_x,
$
which agrees with $\dim(\cB(l^2(D)))=(\sum_{x\in V}d_x)^2$ iff $d_x=d^2_x, \ \forall x\in V$.
\qed
\begin{rem}
Similar statements hold for  $\Phi_{\cS}^\dagger$.
\\
Equation \eqref{eigvecnz} provides the explicit one to one correspondence between eigenvectors of ${\rm Diag} \circ \Phi_{\cS}|_{{\rm Diag}\cB(l^2(D))}$ and of $ \Phi_{\cS}$
associated with the same nonzero eigenvalues.\\
Denoting by $\vertiii{\cdot}$ the operator norm induced on $\cB(l^2(D))$ by $\|\cdot\|_{HS}$, we recover
\be
\vertiii{\Phi_{\cS}}=\vertiii{\Phi_{\cS}^\dagger}=1, 
\ee
which holds for all CPTP unital maps,  and all $p$-norms, $1<p$, see \cite{PGWPR}.\\
In particular, the spectrum of $\Phi_{\cS}$ satisfies $\sigma(\Phi_{\cS})\subset \{z\in \C, \ \text{s.t} \ |z|\leq 1\}$, where all eigenvalues of modulus one are semi-simple, and both $0$ and $1$ belong to the spectrum, since $|V|\geq 2$. \\
Also, $\sigma(\Phi_{\cS})=\overline{\sigma(\Phi_{\cS})}$, as a consequence of $\Phi_{\cS}(A)^*=\Phi_{\cS}(A^*)$, $\forall A\in \cB(l^2(D))$.
\end{rem}

We address now the computation of nonzero eigenvalues of $\Phi_{\cS}$. Recall (\ref{defdiagphi}), 
\be
\Phi^{\rm Diag}={\rm Diag}\circ \Phi_{\cS} |_{{\rm Diag} (\cB(l^2(D))) }, \ \ \text{so that} \ \ \Phi^{\rm Diag}{}^\dagger={\rm Diag}\circ \Phi_{\cS}{}^\dagger |_{{\rm Diag }(\cB(l^2(D))) }.
\ee
We view $\Phi^{\rm Diag}$ as a matrix in the ONB basis ${\rm Diag} (\cB(l^2(D)))=\spa \{|xy\ket\bra xy|\}_{x,y\in V\atop x\sim y}:
$
\begin{lem}
The restriction $\Phi^{\rm Diag}$  has a matrix representation with respect to the orthonormal basis $\{|xy\ket\bra xy|\}_{x,y \in V \atop x\sim y}$
which is entry-wise positive and bistochastic.
\end{lem}
\proof:
We have, see \eqref{matelphid}, 
\be
\Phi^{\rm Diag}(|xy\ket\bra xy|)=\sum_{z\sim x}|S_{zy}(x)|^2|zx\ket\bra zx|.
\ee
Hence, labelling the matrix elements by the directed edges $(xy)$, 
\be
\Phi^{\rm Diag}\simeq (\Phi^{\rm Diag}_{zt\, xy })_{x,y, z, t\in V } \ \ \text{where} \ \ \Phi^{\rm Diag}_{zt \, xy }=\begin{cases}
    |S_{zy}(x)|^2& \text{if } t=x, y \sim x, z\sim x \\
    0 & \text{otherwise}.
\end{cases}
\ee
Since $S(x)=(S_{zy}(x))_{z\sim x\atop y\sim x}$ is unitary for all $x\in V$, we have 
\be
\sum_{zt} \Phi^{\rm Diag}_{zt \, xy }=\sum_{z\sim x} |S_{zy}(x)|^2=1 \ \ {\rm and} \ \ 
\sum_{xy} \Phi^{\rm Diag}_{zt \, xy }=\sum_{y\sim t} |S_{zy}(t)|^2=1.
\ee
\qed
\begin{rem}
i) The operator $\Phi^{\rm Diag}{}^\dagger$ being the adjoint of $\Phi^{\rm Diag}$ has a matrix representation given by the transpose of $ (\Phi^{\rm Diag}_{zt\, xy })_{x,y, z, t\in V } $.\\
ii) The map $\Phi^{\rm Diag}$ is the transition matrix of a naturally related classical Markov process on $D$, the set of directed edges, parameterized by $\cS$.\\
\end{rem}

The main spectral properties of $\Phi^{\rm Diag}$, determining $\sigma(\Phi_{\cS})\setminus \{0\}$, read as follows according to Perron-Frobenius Theorem, see  {\it e.g.} \cite{KS, N, Wo}:
\begin{itemize}
\item $\sigma(\Phi^{\rm Diag})\subset \{z\in \C, \ {\rm s.t.} \ |z|\leq 1\}$, with semi-simple modulus one eigenvalues.
\item If $\Phi^{\rm Diag}$ is irreducible, $1$ is a simple eigenvalue with corresponding eigenvector $\sum_{x\sim y} |xy\ket\bra xy|$. If  $\Phi^{\rm Diag}$ has period $p>1$, $\sigma(\Phi^{\rm Diag})\cap S^1=\{e^{i2\pi k/p}\}_{0\leq k< p}$  are simple eigenvalues. Moreover, the whole spectrum is invariant under rotation by a the angle $2\pi/p$. 
\item If $\Phi^{\rm Diag}$ is irreducible and aperiodic, $1$ is the only eigenvalue of modulus one and is simple. 
\end{itemize}

In particular, we have 
\begin{prop}\label{irredlem} Assume $S_{zy}(x)\neq 0$, for all $x, y, z\in V$, $y\sim x, z\sim x$. Then the bistochastic map $\Phi^{\rm Diag}$ is irreducible.
Moreover, $\Phi^{\rm Diag}$ is aperiodic if the graph $G$ contains a cycle with an odd number of vertices, and it is irreducible with period 2 otherwise.
\end{prop}
\begin{rem}
Irreducibility may hold under weaker assumptions on the scattering matrices, see the example of the tree with three vertices discussed in Section \ref{treetv}.
\end{rem}
\begin{proof}:
Let $\Gamma$ be the  graph associated with the matrix $\Phi^{\rm Diag}$ as follows. It has vertices labelled by the directed edges of $G$, $(xy)$, $x\sim y$, and directed edges from  
$(xy)$ to $(zt)$  iff $\Phi^{\rm Diag}_{zt \,  xy}=\delta_{xt}|S_{zy}(x)|^2>0$. Under the assumption on the elements of $S(x)$, we get that $\Gamma$ has an edge from 
$(xy)$ to $(zx)$ for all $z\sim x$, $y\sim x$, $x\in V$, {\it i.e.} between any incoming edge of $G$ at $x$ and any outgoing edge of $G$ from $x$. The matrix $\Phi^{\rm Diag}$ is irreducible iff there exists a path of directed edges of $\Gamma$ between any two 
of its vertices $(xy)$ and $(zt)$. Since $G$ is connected, there exists a path from $x$ to $t$, made of directed edges $(x_1x)$, $(x_2x_1)$, \dots, $(tx_n)$. But $(xy)$ and $(x_1x)$ are connected in $\Gamma$, as are $(x_1x)$ and $(x_2x_1)$, up to $(x_n x_{n-1})$ and $(tx_n)$, so that $\Phi^{\rm Diag}$ is irreducible.

Because for any directed edge $(xy)$ of $G$, $(xy)$ and $(yx)$ form a directed edge of $\Gamma$, the matrix element $({\Phi^{\rm Diag}}^2)_{xy \, xy }>0$,  so that either the irreducible matrix  $\Phi^{\rm Diag}$ has period 2, or it is aperiodic. When $G$ contains a cycle with $n$ vertices, for any edge $(xy)$ of $G$ between vertices of the cycle, we have $({\Phi^{\rm Diag}}^n)_{xy \, xy }>0$, estimating the matrix element by the product of the elements along consecutive edges. Since an edge $(xy)$ satisfies $({\Phi^{\rm Diag}}^{2p+1})_{xy \, xy }>0$, $p\in N^*$, iff $G$ contains a cycle with $2p+1$ vertices including $x$ and $y$, this yields the result.  
\end{proof}\qed

Consequently, the discrete asymptotic dynamics induced by $\Phi_{\cS}$ is essentially independent of the set of scattering matrices $\cS$ that parameterize $\Phi_{\cS}$:
\begin{cor}\label{asyPhin}
Let $\Phi_{\cS}$ be defined on $\cB(l^2(D))$ with $|V|<\infty$.  If $S_{zy}(x)\neq 0$, for all $x, y, z\in V$, $y\sim x, z\sim x$, then 
\be
\frac1N\sum_{n=0}^{N-1} \Phi_{\cS}^n(\cdot)=\frac{\un}{\sum_{x\in V}d_x}\tr(\cdot)+\ode(1/N).
\ee
If, furthermore, $G$ contains a cycle with an odd number of vertices,
$\exists \gamma>0$ such that
\be\label{expdecfinvol}
\Phi_{\cS}^n(\cdot)=\frac{\un}{\sum_{x\in V}d_x}\tr(\cdot)+ \ode(\e^{-\gamma n}).
\ee
\end{cor}
\begin{proof}: 
Since $\Phi_{\cS}$ is a unital CPTP map, the identity $\un \in \cB(l^2(D)))$ is invariant under $\Phi_{\cS}$ and $\Phi_{\cS}{}^\dagger$. Moreover, the irreducibility of $\Phi^{\rm Diag}$ and Theorem \ref{multipli} imply that $\un$ spans the invariant space. Consequently, the rank one self-adjoint spectral projector of $\Phi_{\cS}$ associated with the eigenvalue one is self-adjoint with respect to the Hilbert Schmidt inner product (\ref{hsip}) and reads 
\be
P_{\{1\}}(\cdot)=\frac{\un}{\sum_{x\in V}d_x}\tr(\cdot).
\ee
If $\Phi^{\rm Diag}$ has period 2, $-1\in\sigma(\Phi_{\cS})$ is simple, and since all other eigenvalues have moduli strictly smaller than 1, the spectral decomposition of $\Phi_{\cS}$, see \cite{Ka}, and the Ces\`aro average yield the result. In case $\Phi^{\rm Diag}$ is aperiodic, $1$ is the only eigenvalue of modulus one, so that the spectral decomposition yields (\ref{expdecfinvol}), which ends the proof.
\end{proof}\qed

This corollary implies that the asymptotics in time of the (Ces\`aro mean) 
probability $\Q_n^{\rho_0}(x)$ to find the quantum walker in the subspace $\cH^{\rm I}_x$, see (\ref{prohix}), satisfies 
\be\label{asymprobvert}
\lim_{N\ra \infty}\frac1N\sum_{n=0}^{N-1}\Q_n^{\rho_0}(x) = \frac{d_x}{\sum_{y\in V}d_y}.
\ee
Of course, in case $\Phi_{\cS}^n(\cdot)$ converges exponentially fast to its invariant spectral projector, no Ces\`aro mean is necessary and the convergences are exponential as well.

\begin{rem}
i) For an alternative argument providing the asymptotic behaviour of $\Phi_{\cS}^n$ on the basis of the spectral properties of $\Phi^{\rm Diag}$, see Appendix C.\\
ii) The asymptotic CPTP map shows that if $G$ is infinite, the existence of a non trivial large times asymptotics  for ${\Phi^{Diag}}^n$ is not guaranteed. We consider such an infinite dimensional case at the end of the present Section.
\end{rem}

\subsection{Tree with three vertices}\label{treetv}

For illustration purposes, we consider the special case of $G=T_3$, the tree with three vertices, see figure  \ref{fig:T_3}. 
Note that this case coincides with the star-graph $SG$ with $N=2$, and central vertex $y$.

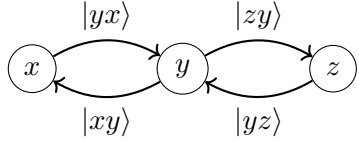
\begin{figure}[h]
\centering
\begin{tikzpicture}
  \node[circle, draw] (x) at (0,0) {$x$};
  \node[circle, draw] (y) at (2,0) {$y$};
  \node[circle, draw] (z) at (4,0) {$z$};

  \draw[->, thick] (x) to [bend left] [edge label=$|yx\rangle$](y);
\draw[->, thick] (y) to [bend left] [edge label=$|xy\rangle$](x);
  \draw[->, thick] (y) to [bend left] [edge label=$|zy\rangle$](z);
\draw[->, thick] (z) to [bend left] [edge label=$|yz\rangle$](y);
\end{tikzpicture}
\caption{\small The graph $G=T_3$ with basis vectors of $\cH_{T_3}$.}
  \label{fig:T_3}
\end{figure}

The degrees of the vertices are $d_x=d_z=1$ and $d_y=2$, so that $\dim l^2(D)=4$ and $\dim \cB(l^2(D))=16$.  The associated scattering matrices in $\cS=\{S(x), S(y), S(z)\}$ are
\begin{align}
S(x)&=e^{i\theta_x}\in U(1), \ \ S(z)=e^{i\theta_z}\in U(1)\nonumber\\
S(y)&=\begin{pmatrix}S_{xx}(y) & S_{xz}(y) \cr S_{zx}(y) & S_{zz}(y)\end{pmatrix}\in U(2)
\end{align}
according to our notation. The matrix representation of $U_\cS$ in the ordered basis 
\be
\{|yx\ket, \, |yz\ket, \, |xy\ket, \, |zy\ket \}
\ee
is given by, see (\ref{stargraphN}),
\begin{align}
U_{\cS}&=
\begin{pmatrix}
0&0&e^{i\theta_x}&0\\
0&0&0&e^{i\theta_z}\\
S_{xx}(y) & S_{xz}(y) &0&0\\
S_{zx}(y) & S_{zz}(y)&0&0
\end{pmatrix}.
\end{align}
Set $D=\diag(e^{i\theta_x},e^{i\theta_z})\in U(2)$. Then, {see \eqref{specUSG},} the spectrum of $U_\cS$ satisfies
\be
\sigma(U_\cS^2)=\sigma(DS(y))=\{e^{i\alpha_1},e^{i\alpha_2}\}, \ \ \textit{i.e.} \ \ \sigma(U_\cS)=\{\pm e^{i\alpha_1/2}, \pm e^{i\alpha_2/2}  \}.
\ee

To express $\Phi_{\cS}$, we observe that
\begin{align}
P_x^{\rm I}&=|xy\ket\bra xy|, \ \ P_z^{\rm I}=|zy\ket\bra zy| \  \ \text{and}\ \ 
P_y^{\rm I}=|yx\ket\bra yx|+ |yz\ket\bra yz|, \nonumber\\
P_x^{\rm O}&=|yx\ket\bra yx|, \ \ P_z^{\rm O}=|yz\ket\bra yz| \  \ \text{and}\ \ 
P_y^{\rm O}=|xy\ket\bra xy|+ |zy\ket\bra zy|,
\end{align}
so that 
\begin{align}
 \ran D^{\rm I}&=\spa \{|xy\ket\bra xy|, |zy\ket\bra zy|, |yx\ket\bra yx|,  |yz\ket\bra yz|, |yx\ket\bra yz|, |yz\ket\bra yx|\} \nonumber \\
\ran D^{\rm O}&=\spa \{|yx\ket\bra yx|, |yz\ket\bra yz|, |xy\ket\bra xy|,  |zy\ket\bra zy|, |xy\ket\bra zy|, |zy\ket\bra xy|\}.
\end{align}
Consequently, with the shorthand $D^{\rm \#}_\perp=\un - D^{\rm \#} $, $\#\in \{\rm O, I\}$,
\begin{align}
\ran {\rm Diag }&=\spa \{|xy\ket\bra xy|, |zy\ket\bra zy|, |yx\ket\bra yx|,  |yz\ket\bra yz|\}, \nonumber \\
\ran D^{\rm I }\circ D^{\rm O}_\perp&=\spa \{|yx\ket\bra yz|, |yz\ket\bra yx|\}\nonumber\\ 
\ran D^{\rm O }\circ D^{\rm I}_\perp&=\spa \{|xy\ket\bra zy|, |zy\ket\bra xy|\},
\end{align}
and we refrain from spelling out the eight basis vectors of $\ran D^{\rm O }_\perp\circ D^{\rm I}_\perp $.
Recall that 
$
\ker \Phi_{\cS}= \ran (\un - D^{\rm I})
$ has dimension 10.
We express $\Phi_{\cS}$ as a block matrix in the following ordered list of subspaces, with their respective ordered bases listed above
\be
\ran {\rm Diag } \ \  \ \ \ran D^{\rm I }\circ D^{\rm O}_\perp \ \  \ \ \ran D^{\rm O }\circ D^{\rm I}_\perp \ \  \ \ \ran D^{\rm O }_\perp\circ D^{\rm I}_\perp.
\ee
Below, the symbol ${\bf 0}\in \C^8$ denotes a vector of  zeros and the vertical and horizontal lines delimitate blocs with respect to the projectors $D^{\rm I}$ and $D^{\rm I}_\perp$
\begin{align}
&\Phi_{\cS}=\\
&
\begin{pmatrix}
0 & 0 & |S_{xx}(y)|^2 & |S_{xz}(y)|^2 &  S_{xx}(y)\overline{S_{xz}(y)} & S_{xz}(y)\overline{S_{xx}(y)}&\vline  & 0 & 0  & {\bf 0}^T \\
0 & 0 & |S_{zx}(y)|^2 & |S_{zz}(y)|^2 & S_{zx}(y)\overline{S_{zz}(y)} & S_{zz}(y)\overline{S_{zx}(y)}& \vline&0 & 0&{\bf 0}^T\\
1 & 0 & 0 & 0 &0 & 0&\vline &0 & 0 &{\bf 0}^T\\
0 & 1 & 0 & 0 &0 & 0&\vline &0 & 0 &{\bf 0}^T\\
0&0&0&0&0&0&\vline&0 & 0 &{\bf 0}^T\\
0&0&0&0&0&0&\vline&0 & 0 &{\bf 0}^T  \\
\hline
%  \hdashline  
   \hspace{-.13cm}\phantom{\Big|}    0 & 0 & S_{xx}(y)\overline{S_{zx}(y)} &  \overline{S_{zz}(y)}S_{xz}(y) &S_{xx}(y)\overline{S_{zz}(y)}  &  S_{xz}(y)\overline{S_{zx}(y)}&\vline&0 & 0 &{\bf 0}^T\\ 
        0 & 0 & \overline{S_{xx}(y)}S_{zx}(y) & S_{zz}(y)\overline{S_{xz}(y)} & S_{zx}(y)\overline{S_{xz}(y)}  & S_{zz}(y)\overline{S_{xx}(y)} &\vline&0 & 0 &{\bf 0}^T
\\
%\hdashline
{\bf 0}&{\bf 0}&{\bf 0}&{\bf 0}&{\bf 0}&{\bf 0}&\vline&{\bf 0} &{\bf 0} &{\bf 0}\, {\bf 0}^T
\end{pmatrix}.
\nonumber
\end{align}
We observe that the Hilbert-Schmidt norm of each nonzero column equals one, as it should. 
The four by four upper left corner of the matrix corresponds to the bistochastic matrix representation of $\Phi^{\rm Diag}$ in the chosen basis of $\ran {\rm Diag }$
\be
\label{bistoch}
\Phi^{\rm Diag}=\begin{pmatrix}
0 & 0 & |S_{xx}(y)|^2 & |S_{xz}(y)|^2 \\
0 & 0 & |S_{zx}(y)|^2 & |S_{zz}(y)|^2 \\
1 & 0 & 0 & 0 \\
0 & 1 & 0 & 0 
\end{pmatrix}.
\ee
The spectral data of $\Phi^{\rm Diag}$ as a function of the scattering matrix $S(y)$ are readily determined. Without going into details, we have 
\begin{itemize}
\item If the matrix elements of $S(y)$ are all nonzero, $\Phi^{\rm Diag}$ is irreducible with period 2. For $S(y)=\frac{1}{\sqrt 2}\begin{pmatrix}1 & 1 \cr -1 & 1\end{pmatrix}$, the Hadamard matrix, $\Phi^{\rm Diag}$ is not diagonalisable, has simple spectrum $\sigma(\Phi^{\rm Diag})=\{1, -1, 0\}$, with an eigennilpotent associated with the spectral projector on its kernel.
\item If $S(y)=\un$, $\Phi^{\rm Diag}$ is reducible with spectrum  $\sigma(\Phi^{\rm Diag})=\{1, -1\}$, each eigenvalue being of multiplicity 2. 
\item For  $S(y)=\begin{pmatrix} 0 & 1\cr 1 & 0\end{pmatrix}$, $\Phi^{\rm Diag}$ is irreducible with period 4, so that $\sigma(\Phi^{\rm Diag})=\{1, -1, i, -i\}$.
\end{itemize}

\subsection{The graph $\Z$ with Hadamard Scattering Matrices}

We consider the infinite graph $G=\Z$ with $S(x)=\frac{1}{\sqrt{2}}\begin{pmatrix} 1 & 1 \\ -1 & 1\end{pmatrix}\in U(2)$, the Hadamard matrix for each $x$, to illustrate the difference in the asymptotics given in Corollary \ref{asyPhin} for finite graphs.  Hence, the nonzero matrix elements of the infinite stochastic matrix  ${\Phi^{\rm Diag}}$, see \eqref{matelphid}, are all equal to $1/2$. 

To analyze ${\Phi^{\rm Diag}}$, we label the vertices of $G$ by $x\in \Z$, and we find it convenient to denote the vectors corresponding to incoming edges at $x$ from the left as $|2x-1\ket$, and those associated with outgoing edges from $x$ to the left as $|2x\ket$, see Figure \ref{figZ2H}.

\begin{figure}[h]
\centering
\begin{tikzpicture}
    % Define the vertices with smaller circles, white fill, and black text
    \foreach \x in {-3, -2, -1, 0, 1, 2, 3} {
        \node[circle, fill=white, draw, minimum size=0.8cm, inner sep=0pt] (\x) at (1.5*\x, 0)
        {\textcolor{black}{\x}};
    }
    
    % Draw curved arrows with labels showing 2*\x and 2*\x+1
    \foreach \x in {-3, -2, -1, 0, 1, 2} {
        % Compute the values 2*x and 2*x+1
        \pgfmathtruncatemacro{\doublex}{2*\x+1}
        \pgfmathtruncatemacro{\doublexplusone}{2*(\x+1)}
        
        % Arrow from x to x+1 (upper arrow with |2*x⟩ label)
        \draw[->, line width=0.3mm] (\x) .. controls +(0.5, 0.5) and +(-0.5, 0.5) .. (\the\numexpr\x+1\relax) 
        node[midway, above] {$|\doublex\rangle$};
        
        % Arrow from x+1 to x (lower arrow with |2*x+1⟩ label)
        \draw[<-, line width=0.3mm] (\x) .. controls +(0.5, -0.5) and +(-0.5, -0.5) .. (\the\numexpr\x+1\relax) 
        node[midway, below] {$|\doublexplusone\rangle$};
    }
       
    % Add three dots at the left and right ends
    \node at (-5.5, 0) {\(\cdots\)};
    \node at (5.5, 0) {\(\cdots\)};
\end{tikzpicture}
\caption{\small The graph $G=\Z$ with incoming and outgoing edges at the vertices. }
\label{figZ2H}
\end{figure}
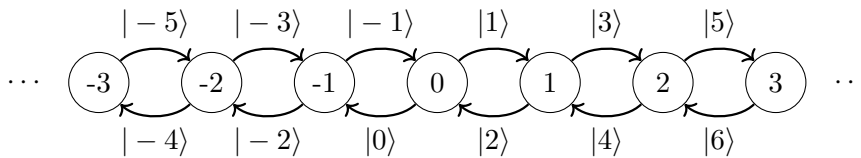

We further write the basis vectors of ${\rm Diag} \cT(l^2(D))$ as $|x\ket\bra x |=e_x$, $x\in \Z$, so that  
${\rm Diag} \cT(l^2(D))\simeq l^1(\Z)$, with ordered basis 
\be\label{obl1}
\{\dots, e_{-2}, e_{-1}, e_{0}, e_1, e_2, \dots \}.
\ee
Accordingly, the matrix elements of $\Phi^{\rm Diag}$ in the ordered basis \eqref{obl1} are denoted by $\Phi^{\rm Diag}_{x y}$. More precisely, we have
\begin{align}\label{phidmat}
\Phi^{\rm Diag}(e_{2x})&=\frac12\big( e_{2(x-1)}+e_{2(x-1)+1}  \big),\nonumber\\
\Phi^{\rm Diag}(e_{2x+1})&=\frac12\big( e_{2(x+1)}+e_{2(x+1)+1}  \big).
\end{align}
This shows that $\Phi^{\rm Diag}$ is an irreducible stochastic matrix of period $2$, looking at the graph $\Gamma$ associated with $\Phi^{\rm Diag}$  in Figure \ref{figgamma}; see also the proof of Proposition \ref{irredlem}.

\begin{figure}[h]
\centering
\begin{tikzpicture}
    % Define the vertices with smaller circles, white fill, and black text
    \foreach \x in {-3, -2, -1, 0, 1, 2, 3} {
        \node[circle, fill=white, draw, minimum size=0.8cm, inner sep=0pt] (\x) at (1.5*\x, 0)
        {\textcolor{black}{$e_{\x}$}};
    }

    % Draw arrows from e_{2x} to e_{2x-1} and e_{2x-2}, and from e_{2x+1} to e_{2x+2} and e_{2x+3}
    \foreach \x in {-1, 0, 1} {
        % Compute even and odd indices
        \pgfmathtruncatemacro{\evenx}{2*\x}
        \pgfmathtruncatemacro{\oddx}{2*\x+1}

        % Arrows from e_{2x} to e_{2x-1} and e_{2x-2} (left-pointing)
        \ifnum \evenx>-3
            % Shorter arrow from e_{2x} to e_{2x-1} (lower curve)
            \draw[->, line width=0.3mm] (\evenx) .. controls +(-0.5, -0.3) and +(0.5, -0.3) .. (\the\numexpr\evenx-1\relax);
        \fi
        \ifnum \evenx>-2
            % Longer arrow from e_{2x} to e_{2x-2} (lower curve)
            \draw[->, line width=0.3mm] (\evenx) .. controls +(-0.5, -0.7) and +(0.5, -0.7) .. (\the\numexpr\evenx-2\relax);
        \fi

        % Arrows from e_{2x+1} to e_{2x+2} and e_{2x+3} (right-pointing)
        \ifnum \oddx<3
            % Shorter arrow from e_{2x+1} to e_{2x+2} (upper curve)
            \draw[->, line width=0.3mm] (\oddx) .. controls +(0.5, 0.3) and +(-0.5, 0.3) .. (\the\numexpr\oddx+1\relax);
        \fi
        \ifnum \oddx<2
            % Longer arrow from e_{2x+1} to e_{2x+3} (upper curve, more pronounced)
            \draw[->, line width=0.3mm] (\oddx) .. controls +(0.5, 0.7) and +(-0.5, 0.7) .. (\the\numexpr\oddx+2\relax);
        \fi
    }

    \draw[->, line width=0.3mm] (-3) .. controls +(0.5, 0.3) and +(-0.5, 0.3) .. (-2); % Shorter arrow to e_{-2}
    \draw[->, line width=0.3mm] (-3) .. controls +(0.5, 0.7) and +(-0.5, 0.7) .. (-1); % Longer arrow to e_{-1}

% Add the new left-pointing long arrow from e_{-2} to the dummy coordinate (-4, 0)
    \draw[->, line width=0.3mm] (-2) .. controls +(-0.5, -0.7) and +(0.5, -0.7) .. (-5.9, -0.3); % Long arrow to dummy coordinate (-4, 0)

    % Add three dots at the left and right ends
    \node at (-5.5, 0) {\(\cdots\)};
    \node at (5.5, 0) {\(\cdots\)};
\end{tikzpicture}
\caption{\small The graph $\Gamma$ with the specified incoming and outgoing edges.}
\label{figgamma}
\end{figure}
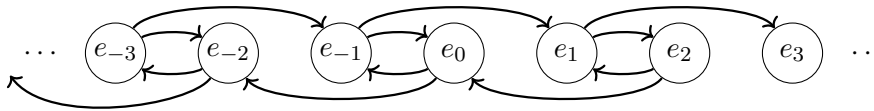

As a simple computation reveals, ${\Phi^{\rm Diag}}$ admits no invariant vector in $l^1(\Z)$. 
We prove the following property of its matrix elements in Appendix B:
\begin{lem}\label{lemwoes}
There exists $c<\infty$ such that for any $x, y\in \Z$, { there is an $N_0(x-y)\in \N$ depending on $x-y$, so that }
\be
0\leq {({\Phi^{\rm Diag})}^{n}}_{x y}\leq c/\sqrt{n} , \ \ \text{for all } \ \ n\geq N_0(x-y).
\ee 
\end{lem}
As a consequence, the map provided in Corollary \ref{cordynphi}
\be
\Phi_{\cS}^n= \Phi_{\cS}\circ  {(\Phi^{\rm Diag})}^{n-2}\circ  {\rm Diag}\circ \Phi_{\cS}, 
\ee
converges to zero in the weak sense. More precisely, $\forall (xy), (x'y'), (zt), (z't')\in D$, 
\be
\lim_{n\ra\infty}\tr (|x'y'\ket\bra z't'|\, \Phi_{\cS}^n(|xy\ket\bra zt|))=0,
\ee
since the matrix elements of $\Phi_{\cS}$ connect $P_x^{\rm I}$ and $P_x^{\rm O}$, so that the argument of the trace is a sum of finitely many matrix elements of $ {(\Phi^{\rm Diag})}^{n-2}$. Hence, for infinite graphs, the dynamics can send the state at infinity, preventing the existence of an asymptotic state.

\section{Induced Open Quantum Walks on $G$}\label{sec:IOQW}

The open QW $\Phi_{\cS}\in \cT(l^2(D))$ defined above lives on the directed edges of $G$. We provide here an open QW $ \Psi_{\cS}\in \cT(l^2(V))$ defined via  $\Phi_{\cS}$ which lives on the vertices of $G$, by reducing each incoming subspaces $\cH_x^{\rm I}$, $x\in V$ to the vectors $|x\ket\in l^2(V)$ associated with the vertex $x\in V$.  To do so, we make use of the natural boundary operator $R$ and its adjoint $R^*$  between $l^2(D)$ to $l^2(V)$, see (\ref{defrrs}), generalizing the approach \cite{HKSS2} of the unitary Grover Walk to open QW.

\medskip

We recall the expressions of $R$ and $R^*$ given by 
\begin{align}\label{boundaryop}
R=\sum_{x\in V}|x\ket\bra \omega(x)|:l^2(D)\ra l^2(V), \ \ R^*=\sum_{x\in V}|\omega(x)\ket\bra x|:l^2(V)\ra l^2(D), 
\end{align}
where $\{\omega(x)\}_{x\in V}$ is a family of normalized vectors of $l^2(D)$ indexed by $x\in V$ such that
\be\label{defomegax}
\omega(x)=\sum_{y\sim x}\omega_y(x)|xy\ket\in P_x^{\rm I}l^2(D), \ \text{with} \ \bra \omega(x) |\omega({x'})\ket =\delta_{x x'}, \ \forall x, x' \in V.
\ee
Then we consider the map
\be
\cR(\cdot)=R^*\cdot R : \cT(l^2(V))\ra \cT(l^2(D))
\ee
which is CPTP thanks to \eqref{idv}. Note that for $G$ infinite, the series \eqref{idv} for $RR^*$ converges in the strong sense. Also, the series \eqref{idpi} for $R^*R$ converges strongly so that the following map is well defined 
\be
\cR^\dagger(\cdot)=R\cdot R^* : \cT(l^2(D))\ra \cT(l^2(V))
\ee
and can be extended as a map on $\cB(l^2(D))$, which coincides with  $\cR^*$, the adjoint  of $\cR$, and is CP and unital. 
We now set
\be
\tilde \Phi_{\cS}=\cR^\dagger\circ  \Phi_{\cS} \circ \cR :  \cT(l^2(V))\ra \cT(l^2(V)),
\ee
which in terms of the expression (\ref{defphi}) for $\Phi_{\cS}$ reads
\be\label{tildefphi}
\tilde \Phi_{\cS}(B)=\sum_{x\in V}RK(x)R^* B R K^*(x) R^*, \ \ \forall B\in \cT(l^2(V)).
\ee
To assess the properties of $\tilde \Phi_{\cS}$, we compute 
\begin{align}
RK(x)R^*=\sum_{x',x''\in V}|x'\ket\bra \omega({x'})|K(x) \omega({x''})\ket\bra x''|,
\end{align}
where
\begin{align}
\bra \omega({x'})|K(x)|\omega({x''})\ket&=\sum_{y\sim x \atop z\sim x}S_{zy}(x)\bra \omega({x'})|zx\ket\bra xy|\omega({x''})\ket \nonumber \\
&=\sum_{y\sim x \atop z\sim x}{S_{zy}(x)}  \overline{\omega_x(z)}  {\omega_y(x)} \delta_{x' z}\delta_{x x''},
\end{align}
so that
\be
RK(x)R^*=\sum_{y\sim x \atop z\sim x}{S_{zy}(x)} \overline{\omega_x(z)}  {\omega_y(x)}|z\ket\bra x|.
\ee
This leads us to introduce
\be\label{newvec}
v(x)=\sum_{z\sim x}v_z(x)|z\ket \in l^2(V), \ \text{with} \ v_z(x)=(S(x)\omega(x))_z=\sum_{y\sim x} S_{zy}(x)\omega_y(x),
\ee
where  we view $\omega(x)$ and $S(x)\omega(x)$ as vectors in $\C^{d_x}$ above, and
\be\label{vecthet}
\theta(x)=\sum_{z\sim x}\overline{\omega_x(z)}v_z(x)|z\ket\in  l^2(V).
\ee
Note that $v(x)$ and $\theta(x)$  have finitely many nonzero components and since $S(x)$ is unitary and $\|\omega(x)\|=1$, we have
\be\label{normthet}
\|v(x)\|^2=1, \ \ \|\theta(x)\|^2=\sum_{z\sim x}|\omega_x(z)|^2|v_z(x)|^2\leq \sum_{z\sim x}|v_z(x)|^2=1.
\ee
Thus, the Kraus operators associated with the CP map $\tilde \Phi_{\cS}$ have rank one 
\be\label{raone}
RK(x)R^*=|\theta(x)\ket\bra x|, \ \ RK^*(x)R^*=|x\ket\bra \theta(x)|.
\ee
From the computations above we deduce that
\be\label{quin}
\tilde \Phi_\cS^*(\un_V)=\sum_{x\in V}RK^*(x)R^*RK(x)R^*=\sum_{x\in V} |x\ket\bra x| \ \|\theta(x)\|^2\leq \un_{V}.
\ee
{This makes $\tilde \Phi_{\cS}$ a quantum operation - {\it i.e.} a completely positive trace-nonincreasing map satisfying (\ref{quin})-} rather than a CPTP map. 
Incidentally, we get from (\ref{normthet}) that $\tilde \Phi_{\cS}$ is a CPTP map iff $|\omega_x(z)|^2=1$ for all $x\in V$ and all $z\sim x$, which is possible only if $|V|=2$.

\medskip

To cure this defect, we modify the vector $\theta(x)$, observing $\bra \theta(x)|x\ket=0$: set
\be\label{defchi}
\chi(x)=\theta(x)+e^{i\beta_x}\sqrt{1-\|\theta(x)\|^2}|x\ket, \ \text{where} \ \beta_x\in \R,
\ee
and, recall (\ref{raone}), we define rank one operators on $l^2(V)$ 
\be\label{rankonekraus}
\cG(x)=|\chi(x)\ket\bra x|, \ \ \cG^*(x)=|x\ket\bra \chi(x)|.
\ee
\begin{definition}
For each $x\in V$,  let $\theta(x)$,  $\chi(x)$ and $\cG(x)$ be given by  (\ref{vecthet}), (\ref{defchi}) and (\ref{rankonekraus}). We define the map 
$\Psi_{\cS}(\cdot): \cT(l^2(V))\ra \cT(l^2(V))$ by 
\be\label{definpsi}
\Psi_{\cS}(\cdot)=\sum_{x\in V}  \cG(x) \cdot \cG^*(x).
\ee
\end{definition}

\medskip

{The normalization performed by the replacement of $\theta(x)$ by $\chi(x)$ is a choice. See also Remark \ref{63} iv) below.}
To assess the properties of $\Psi_{\cS}$, it is useful to introduce a  stochastic matrix  $P\in M_{|V|}(\C)$ by its elements, where $|V|=\infty$  if $G$ is infinite.
\be\label{stomatP}
P_{x y}= |\bra y|\chi(x)\ket|^2\geq 0, \ \text{for} \ x, y\in V.
\ee
Note that $P_{x y}=0$ if $x\not \sim y$ and $x\neq y$,  which makes $P$ sparse. The matrix
$P$ is the transition matrix of a discrete time Markov chain $M$ on $V$, $M:\N \ra V$ such that at time $t\in \N$ and for $x, y\in V$,
\be
\P(M(t)=y|M(t-1)=x)= |\bra y|\chi(x)\ket|^2=P_{xy}.
\ee 
A probability vector on $V$ is a row vector of the form $p=(p_x)_{x\in V}$, {\it i.e.}, given an ordering of the vertices,
\be
p=(p_{x_1}, p_{x_2}, \dots, p_{x_{|V|}}),
\ee with $p_x\geq 0$ and $\sum_{x\in V}p_x=1$. The support of $p$ is the set $\{x\in V \ \text{s.t.} \ p_x>0\}$, and $p$ is strictly positive if its support is $V$. For a {set} of vertices 
$S\subset V$, we denote  the probability of this set by $p(S)=\sum_{x\in S}p_x$.\\
Given an initial probability  vector  $p$ we have for all $t\in \N$, 
\be\label{procm}
\P(M(0)=x)=p_x \ , \forall x\in V, \ \Rightarrow \ \P_p(M(t)=x)=(pP^t)_x,
\ee 
where the subscript $p$ indicates the initial distribution of the Markov chain $M$, see  \cite{N, KS, Wo} for example.

\begin{thm}\label{dynpsiV}
The map $\Psi_{\cS}$ defined in (\ref{definpsi}) is a CPTP map on $ \cT(l^2(V))$.\\
With $P$ given by (\ref{stomatP}), we have for all $n\geq 1$ 
\be\label{evolPP}
\Psi_{\cS}^n(\cdot)=\sum_{x, y\in V} |\chi(y)\ket\bra x| \, \cdot \,  |x\ket\bra \chi(y)| \, {(P^{n-1})}_{x y}.
\ee
Moreover, for any initial state $\rho_0\in \cD\cM(l^2(V))$, the probability to find the quantum walker on the vertex  $x\in V$ at time $n\geq 1$, $\Q_n^{\rho_0}(x)$, reads
\be
\Q_n^{\rho_0}(x)=(r_0 P^{n})_x,
\ee  
where $r_0=(\bra x_1| \rho_0 x_1\ket, \bra x_2| \rho_0 x_2\ket, \dots , \bra x_{|V|}| \rho_0 x_{|V|}\ket)$.
\end{thm}
\begin{rem}\label{63} i) For $A\in \cT(l^2(V))$, the sum over $x, y\in V$ defining $\Psi_{\cS}^n(A)$ converges in trace norm.\\
ii) For $\rho_0\in \cD\cM(l^2(V))$, $n\geq 1$, we can write
\begin{align}
\Psi_{\cS}^n(\rho_0)
&=\sum_{y\in V} |\chi(y)\ket\bra \chi(y)| \, \P_{r_0}(M(n-1)=y),
\end{align}
the expectation of the matrix valued random variable $ |\chi(\cdot)\ket\bra \chi(\cdot)|$ on $V$ with respect to the law of the Markov chain $M$ at time $n-1$.\\
iii) Similarly, the probability of presence of the quantum walker at $x\in V$ at time $n$ is given by the law of the Markov chain $M$:
\be
\Q_n^{\rho_0}(x)=\P_{r_0}(M(n)=x).
\ee 
iv) There are other possibilities to modify the Kraus operators in order to promote $\tilde \Phi_\cS$ to a $CPTP$ map. 
A well-known modification consists in adding the extra Kraus operator $(\un_V-\Phi_\cS^\dagger(\un_V))^{1/2}\geq 0$ to $\{RK(x)R^*\}_{x\in V}$. However, since
\be
(\un_V-\Phi_\cS^\dagger(\un_V))^{1/2}=\sum_{x\in V }(1-\|\theta(x)\|^2)^{1/2}|x\ket\bra x|
\ee 
has a priori a large rank, this supplementary Kraus operator  gives the corresponding CPTP map less structure than $\Psi_\cS$ has.
\end{rem}
\proof:
By construction,
$\|\chi(x)\|^2=1$
for all $x\in V$, so that the Kraus operators defining $\Psi_{\cS}$ satisfy
\be
\sum_{x\in V}\cG^*(x)\cG(x)=\sum_{x\in V}|x\ket\bra x|=\un_V,
\ee
with convergence in the strong sense if $G$ is infinite, which ensure $\Psi_{\cS}$ is CPTP.
Then, for any $n\in \N$ we compute
\begin{align}\label{psin}
\Psi_{\cS}^n(\cdot)=\sum_{x_1, x_2, \dots, x_n\in V} |\chi(x_n)\ket\bra x_1| \, \cdot \, & |x_1\ket\bra \chi(x_n)| \times \\
& |\bra x_2|\chi(x_1)\ket|^2 |\bra x_3|\chi(x_2)\ket|^2\dots  |\bra x_n|\chi(x_{n-1})\ket|^2, \nonumber
\end{align}
which, with   (\ref{stomatP}), implies \eqref{evolPP} immediately.
Finally, 
\begin{align}\label{expqnrho}
\Q_n^{\rho_0}(x)&=\bra x| \Psi_{\cS}^n(\rho_0)x\ket\nonumber \\
&= \sum_{x', y'\in V} \bra x|\chi(y')\ket\bra x'| \rho_0  x'\ket\bra \chi(y')|x\ket  \, {(P^{n-1})}_{x' y'}
\nonumber \\
&= \sum_{x', y'\in V} \bra x'| \rho_0  x'\ket {(P^{n-1})}_{x' y'}P_{y' x}=(r_0 P^{n})_x,
\end{align}
where
\be\label{rorho}
r_0=(\bra x_1| \rho_0 x_1\ket, \bra x_2| \rho_0 x_2\ket, \dots , \bra x_{|V|}| \rho_0 x_{|V|}\ket)
\ee 
which, with \eqref{procm}, ends the proof. 
\qed\\

\medskip 

We mainly consider $G$ finite in the rest of this section, and will comment along the way on the infinite graph case.

\medskip 

In this case, we can deduce the large $n$ behaviour of the  map $\Psi_{\cS}^n(\cdot)$ from the properties of the stochastic matrix $P$.
\begin{cor}\label{corpsin}
Let $G$ be finite. If $P$ is irreducible, there exists a strictly positive probability vector $\pi=(\pi_{x_1}, \pi_{x_2}, \cdots, \pi_{x_{|V|}})$ such that with ${\bf 1}^T=(1,1\cdots, 1)$ and for $N$ large enough
$\frac1N\sum_{n=0}^{N-1}P^n={\bf 1} \pi+O(N^{-1})$
and 
\be
\frac1N\sum_{n=0}^{N-1}\Psi_{\cS}^n(\cdot)=\sum_{y\in V} |\chi(y)\ket\bra \chi(y)| \,\pi_{y} \, \tr(\cdot)+O(N^{-1}).
\ee
If $P$ is irreducible and moreover aperiodic, there exists 
$\gamma>0$ such that for $n$ large, $P^n={\bf 1} \pi+O(e^{-\gamma n})$ and 
\be
\Psi_{\cS}^n(\cdot)=\sum_{y\in V} |\chi(y)\ket\bra \chi(y)| \,\pi_{y} \, \tr(\cdot) + O(e^{-\gamma n}).
\ee
Moreover, for any initial state $\rho_0\in \cD\cM(l^2(V))$, $\Q_n^{\rho_0}(x)$, the probability to find the quantum walker on the vertex  $x\in V$ at time $n$,  satisfies for $P$ irreducible
\be\label{CesaroQrn}
\frac1N\sum_{n=0}^{N-1}Q_n^{\rho_0}(x)=\pi_x+ O(N^{-1}).
\ee
If $P$ is irreducible and aperiodic, 
\be\label{expdecayQrn}
Q_n^{\rho_0}(x)=\pi_x+ O(e^{-\gamma n}).
\ee
\end{cor}
\proof:
If $P$ is irreducible, Perron-Frobenius Theorem ensures the existence of a strictly positive probability row vector $\pi=(\pi_{x_1}, \pi_{x_2}, \cdots, \pi_{x_{|V|}})$ such that with ${\bf 1}^T=(1,1\cdots, 1)$ and for $N$ large enough
\be\label{Cesaropn}
\frac1N\sum_{n=0}^{N-1}P^n={\bf 1} \pi+O(N^{-1}),
\ee
see {\it e.g.} \cite{N, KS, Wo}. Hence, the matrix elements of the projector ${\bf 1} \pi$ read $({\bf 1} \pi)_{xy}=\pi_y>0$. Inserting this in \eqref{evolPP} yields 
\be
\frac1N\sum_{n=0}^{N-1}\Psi_{\cS}^n(\cdot)=\sum_{x, y\in V} |\chi(y)\ket\bra x| \, \cdot \,  |x\ket\bra \chi(y)| \,\pi_{y} + O(N^{-1}),
\ee
which, with  $\sum_{x\in V}\bra x| \, \cdot \,  |x\ket=\tr (\cdot)$, proves the result in the irreducible case. The irreducible and aperiodic case is proven the same way.

Consider now the probability $Q_n^{\rho_0}$ given by  formula \eqref{expqnrho}. Inserting  in this expression the large time asymptotics (\ref{Cesaropn}) for the Ces\`aro mean of $P^n$ in the irreducible case proves \eqref{CesaroQrn}. Formula \eqref{expdecayQrn} is proven similarly in the irreducible and aperiodic case.
\qed

\begin{rem}\label{specproof}
i) Up to the strict positivity of $\pi$, the Corollary can be given a purely spectral proof that we sketch in Appendix C for completeness.\\
ii) Since the asymptotic map typically differs from $\un_V\tr (\cdot)/|V|$, this shows that $\Psi_\cS$ is not entropy non-decreasing.\\
iii)  If $P$ fails to be irreducible, the limiting behaviours of $\Psi_{\cS}(\rho_0)$ and $\Q_n^{\rho_0}$ depend on the initial state $\rho_0\in \cD\cM(l^2(V))$, as shown on the example in Section \ref{B}. \\
iv) The formula of the asymptotic state is likely to hold true for infinite graphs. We show that this is the case on the example of the next section.
\end{rem}

\medskip

The results presented in Theorem \ref{dynpsiV} and Corollary \ref{corpsin} only depend on the form  \eqref{rankonekraus} of the rank one Kraus operators of $\Psi_{\cS}$. In particular, the asymptotic state of $\Psi_{\cS}^n$ and corresponding probabilities to find the walker on each vertex are monitored by the invariant probability vector $\pi$ of the Markov chain associated with $P$. This is to be contrasted with the corresponding asymptotic quantities stated in Corollary \ref{asyPhin} and (\ref{asymprobvert}) for the open QW $\Phi_{\cS}^n$ which only depend on the degrees of the vertices of $G$. 

The non trivial dependence of  the invariant probability vector $\pi$  on the parameters of the construction, {\it i.e.} the graph $G$, the scattering matrices $\cS=\{S(x)\}_{x\in V}$ and the vectors $\{\omega(x)\}_{x\in V}$ appearing in $R$, is displayed by the examples of the next section.

\subsection{Examples}
We illustrate  the variety of behaviours induced open SQWs on graphs can display as a function of the scattering matrices by working out the cases where they correspond to the $\alpha-$Grover Walk, and to the Discrete Fourier Transform.
\medskip

Let $G$ be an arbitrary finite graph and consider the boundary operator $R$ and its adjoint (\ref{boundaryop}) parameterized by the family of unit vectors
\be\label{choiomeg}
\omega(x)=\frac{\bf 1}{\sqrt{d_x}} \in \C^{d_x},\ \ \forall x\in V.
\ee
This choice is motivated by the fact that it yields an equal weight to each of the vectors spanning $\cH_x^{\rm I}$, and thus is likely to make $P$ irreducible. 
The components of the vector $v(x)$ (\ref{newvec}) thus satisfy
\be
v_z(x)=\frac{1}{\sqrt{d_x}} \sum_{y\sim x}S_{zy}(x)=:  \frac{1}{\sqrt{d_x}} L_z(x),\ \ \forall x\in V, z\sim x.
\ee
Consequently, see (\ref{vecthet}), (\ref{defchi}),
\begin{align}
\theta(x)&= \frac{1}{\sqrt{d_x}}\sum_{z\sim x}\frac{L_z(x)}{\sqrt{d_z}} |z\ket, \ \ \text{with} \ \|\theta(x)\|^2= \frac{1}{d_x}\sum_{z\sim x}\frac{|L_z(x)|^2}{d_z},\nonumber\\
\chi(x)&=\theta(x)+e^{i\beta_x}\sqrt{1-\|\theta(x)\| ^2}|x\ket,
\end{align}
which leads to the stochastic matrix $P$ (\ref{stomatP}):
\be\label{stochomunif}
P_{xy}=\begin{cases} 
\frac{|L_y(x)|^2}{d_x d_y} & \text{if } x \sim y  \\
1- \|\theta(x)\|^2 & \text{if } x = y \\
0 & \text{otherwise} 
\end{cases},\ \ \forall x,y\in V. 
\ee

For this choice of boundary operator, we now consider two families of scattering matrices $\cS=\{S(x)\}_{x\in V}.$

\subsubsection{The $\alpha-$Grover Scattering Matrices}\label{A}
Following (\ref{Salpha}), we  set $\forall x\in V$,
\be\label{SalphaOW}
S(x)=\frac1{d_x}|{\bf 1}\ket\bra {\bf 1}|+e^{i\alpha}\big(\un_{d_x} - \frac1{d_x}|{\bf 1}\ket\bra {\bf 1}|\big),
\ee 
the matrix elements of which read, with $\nu(\alpha)=1-e^{i\alpha}\neq 0$,
\be
S_{zy}(x)=\begin{cases}  e^{i\alpha}+\frac{\nu(\alpha)}{d_x} & \text{if } y=z\sim x  \\
\frac{\nu(\alpha)}{d_x} &  \text{if }  y\neq z, \ y\sim x, z\sim x \\
0 & \text{otherwise.} 
\end{cases}
\ee
Thus, for $z\sim x$
\be
L_z(x)=(d_x-1)\frac{\nu(\alpha)}{d_x}+e^{i\alpha}+\frac{\nu(\alpha)}{d_x} = 1,
\ee
and, $\forall x,y\in V$,
\be\label{stochA}
P_{xy}=\begin{cases} 
\frac{1}{d_x d_y} & \text{if } x \sim y  \\
1-\frac1{d_x}\sum_{z\sim x}\frac1{d_z}& \text{if } x = y \\
0 & \text{otherwise.} 
\end{cases}
\ee
We observe that $P=P^T$ is independent of $\alpha$ and is actually bistochastic. Since $G$ is connected, $P$ is irreducible. We have shown
\begin{lem}
The invariant probability vector of the stochastic matrix $P$ associated with the open QW  on $G$ induced by the $\alpha-$Grover walk is uniform on the vertices of $G$:
\be
\pi_y=1/|V|, \ \ \forall y\in V.
\ee
\end{lem}
Thus the asymptotic probability of presence of the walker on the sites of $G$ is uniform as well. Also, the asymptotic state in Ces\`aro mean, is $\frac{1}{|V|}\sum_{y\in V} |\chi(y)\ket\bra \chi(y)|$, see Corollary \ref{corpsin}, where
\be
 |\chi(y)\ket=\frac{1}{\sqrt{d_y}}\sum_{z\sim y}\frac1{\sqrt{d_z}}|z\ket+e^{i\beta_y}\sqrt{1-\frac{1}{d_y}\sum_{z'\sim x}\frac{1}{d_{z'}}}|y\ket.
\ee

\begin{rem} i) {Actually, the arguments that lead to these results also hold if $S(x)=\un_x$, for all $x\in V$, which corresponds formally to the value $\alpha=0$ that was excluded.}\\
ii) Any unitary matrix with constant diagonal elements and constant off-diagonal elements coincides, up to a phase, with  (\ref{SalphaOW}), so that its associated stochastic matrix $P$ is identical to (\ref{stochA}). The same holds if the unitary matrix is multiplied from the left by any unitary diagonal matrix, thanks to  (\ref{stochomunif}).
\end{rem}

\subsubsection{The Discrete Fourier Transform Scattering Matrices}\label{B}
For any $x\in V$, let $\Omega_x=e^{-2\pi i/d_x}$ and consider 
\be\label{FFT}
S(x)=\frac{1}{\sqrt{d_x}}\begin{pmatrix}
1 & 1 & 1 & \cdots & 1 \\
1 & \Omega_x & \Omega_x^2 & \cdots & \Omega_x^{d_x-1} \\
1 & \Omega_x^2 & \Omega_x^4 & \cdots & \Omega_x^{2(d_x-1)} \\
\vdots & \vdots & \vdots & \ddots & \vdots \\
1 & \Omega_x^{d_x-1} & \Omega_x^{2(d_x-1)} & \cdots & \Omega_x^{(d_x-1)(d_x-1)}
\end{pmatrix}
\ee
written in the canonical basis of $\C^{d_x}$. Here, the $j k$ matrix element of $S(x)$ corresponds to $S_{x_j x_k}(x)$, for a labelling of the $d_x$ adjacent vertices $x_1, x_2, \dots, x_{d_x}$  of $x$.
Since the rows of a unitary matrix are orthogonal, we get for $x_j\sim x$,
\be\label{Llabel}
L_{x_j}(x)=\begin{cases}  \sqrt{d_x} & \text{if} \ j=1, \\
0 & \text{if}  \ j\geq 2,\end{cases}
\ee
so that $\theta(x)=\frac{1}{\sqrt{d_{x_1}}}|x_1\ket$, with $x_1\sim x$.
Consequently, $\forall x,y\in V$,
\be\label{stochB}
P_{xy}=\begin{cases} 
\frac{1}{d_{x_1}} & \text{if } y=x_1\sim x  \\
1-\frac{1}{d_{x_1}} & \text{if } y = x \\
0 & \text{otherwise.} 
\end{cases}
\ee
{Let us stress that $P_{xy}$ and the induced dynamics depend on the choice of the vertices labelling.}
This stochastic matrix is in general not irreducible, but its invariant probability vectors can be determined. Once the labelling of the adjacent vertices to $x$ is fixed for all $x\in V$, $P$ defines the map $N$ on $V$ such that 
\begin{align}\label{mapn}
N:&\ V\ra V \nonumber\\
&\ x\mapsto N(x)=y, \ \ \text{where } \ y\sim x \text{ is uniquely determined by } P_{xy}>0.
\end{align}
We can then construct a directed subgraph $\Sigma\subset G$ (without loops) associated with $P$ or equivalently with $N$: $\Sigma$ has same vertex set $V$ as $G$ and $(yx)$ form an edge of $\Sigma$ iff $y=N(x)$. $\Sigma$ is also known as a {\it functional graph} or {\it successor graph}, and may or may not by connected. 

Consider  the set of the connected components of $\Sigma$,
\be
{\rm CC}(\Sigma)=\cup_{i=1}^r C_i \ \ \text{where} \ \
C_i \subseteq V \ \text{is connected in } \Sigma, \ i=1,\dots, r.
\ee
Each component $C_i\in {\rm CC}(\Sigma)$, $i=1,\dots, r$, consists of one cycle, the vertices of which are possibly the roots of a finite tree. See Figure \ref{fig:exN} for an example corresponding to the map $N$ on $\{s,t,u,v,w,x,y,z\}$ with values 
\be\label{exaN}
\{ N(s)=y, N(t)=x, N(u)=x, N(v)=t, N(w)=t, N(x)=y, N(y)=z, N(z)=x \}.
\ee
\begin{figure}[h]
\centering
\begin{tikzpicture}[->, >=stealth, node distance=2cm, main node/.style={circle, draw, minimum size=0.8cm}]

    % Nodes
    \node[main node] (1) at (0, 0) {x};
    \node[main node] (2) at (1, 1.732 ) {y}; % (2, sqrt(3)*2)
    \node[main node] (3) at (2, 0) {z};
    \node[main node] (4) at (-2, 0) {t};
    \node[main node] (5) at (-1, 1.732) {u};
    \node[main node] (6) at (-4, 0) {v};
    \node[main node] (7) at (-3, 1.732) {w};
    \node[main node] (8) at (3, 1.732 ) {s};

    % Edges
    \path[->, thick] (1) edge (2);
    \path[->, thick] (2) edge (3);
    \path[->, thick] (3) edge (1);
    \path[->, thick] (4) edge (1);
    \path[->, thick] (5) edge (1);
    \path[->, thick] (6) edge (4);
    \path[->, thick] (7) edge (4);
    \path[->, thick] (8) edge (2);

\end{tikzpicture}
\caption{\small The functional graph $\Sigma$ associated with the map $N$ (\ref{exaN}) }
  \label{fig:exN}

\end{figure}
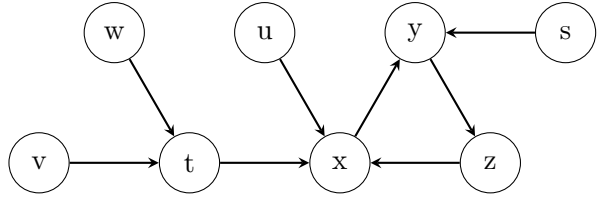

Given the stochastic matrix $P$ (\ref{stochB}), we first determine the dimension of $\ker (P-\un)$, which equals the dimension of the subspace of  left invariant vectors of $P$.
\begin{lem}
The stochastic matrix associated with DFT scattering matrices satisfies
\be
\dim \ker (P-\un)=|{\rm CC}(\Sigma)|.
\ee
\end{lem}

\proof:  We set $E=P-\un$, with 
matrix elements $E_{xy}$, ${x, y\in V}$  
\be\label{matE}
E_{xy}=\begin{cases} 
\eps_{N(x)} & \text{if } y=N(x)\sim x  \\
-\eps_{N(x)} & \text{if } y = x \\
0 & \text{otherwise} 
\end{cases},\ \ \forall x,y\in V, \ \ \text{with} \ \eps_x=1/d_x\in (0,1].
\ee
Labelling the vertices by $V=\{1,2,\dots, n\}$, with $n=|V|$, we get that 
\be
v=(v_1, v_2, \dots, v_n)^T\in \ker (P-\un)=\ker E
\ee  
if and only if
\be
\eps_{N(j)} v_j=\eps_{N(j)} v_{N(j)}, \ \forall j\in V.
\ee 
Hence, writing $V=C_1 \cup C_2 \cup \dots \cup C_r$, where $r<n$, we get for each $i\in \{1,\dots, r\}$ and all $k\in C_i\in{\rm CC(\Sigma)}$, $v_k=v(C_i)\in \C$ arbitrary, since $\eps_j>0$ for all $j\in V$.
Thus $\ker (P-\un)$ has dimension $r$, the number of connected components of $\Sigma$.\qed 

Thus, if $\Sigma$ is connected, $P$ admits a unique invariant row probability vector $\pi$. 

\begin{lem}\label{piconnected}
Assume $\Sigma$ is connected and denote by ${\rm Cyc}\subseteq V$ the set of its vertices belonging to the cycle of  $\Sigma$. Then the invariant row probability vector $\pi$ of $P$ is supported on $\rm Cyc$ and  has components  
\be\label{pidft}
\pi_x=\begin{cases}
{d_{N(x)}}/{\sum_{y\in {\rm Cyc}}d_y}& \text{if } x\in {\rm Cyc}\\
0 & \text{\em otherwise}.
\end{cases}
\ee 
\end{lem} 
\begin{rem}
The lack of strict positivity of $\pi$ shows $P$ is not irreducible.
\end{rem}
\proof:
As above, we use the labelling $V=\{1,2,\dots, n\}$, with $n=|V|$ and look for $\pi=(\pi_1, \dots, \pi_n)$ such that $\pi E=0$. Assuming the vertices $\{1,\dots, s\}={\rm Cyc} \subseteq V$ satisfy 
\be\label{periodN}
N(1)=2, N(2)=3, \dots, N(s-1)=s, N(s)=1,
\ee 
we get the following block structure for $E$
\be
E=\begin{pmatrix} A & {\mathbb O} \\
B & C\end{pmatrix}
\ee 
where $A\in M_s(\C)$ reads
\be
A=\begin{pmatrix}
-\eps_2 & \eps_2 &              &                                   &                      \\
             &-\eps_3 & \eps_3 &                                       &                       \\
             &           &       \ddots      &                                &                      \\
            &           &            &                   -\eps_{s} &            \eps_{s}        \\
    \eps_1         &           &                                   &                      &    -\eps_1      
\end{pmatrix},
\ee
$B\in M_{n-s\ s}(\C)$ and $C\in M_{n-s \ n-s}(\C)$ and ${\mathbb O}\in M_{s \ n-s}(\C)$ is the matrix with zero elements.
Writing $\pi=(\pi^{(1)}, \pi^{(2)})$ where $\pi^{(1)}$ is a length $s$ row vector and $\pi^{(2)}$ a length $n-s$ row vector, the equation $\pi E=0$ reads
\be
(\pi^{(1)}A+\pi^{(2)}B, \pi^{(2)}C)=0.
\ee
Knowing the subspace of solutions is of dimension 1, we set $\pi^{(2)}=0$ and look for non trivial solutions to $\pi^{(1)}A=0$:
\begin{align}
&\pi^{(1)}_1\eps_2=\pi^{(1)}_s\eps_1\nonumber\\
&\pi^{(1)}_{j-1}\eps_j=\pi^{(1)}_{j}\eps_{j+1}, \ \ 2\leq j\leq s-1 \nonumber\\
&\pi^{(1)}_{s-1}\eps_s=\pi^{(1)}_s\eps_1.
\end{align}
Such a solution is given  by
\be
\pi^{(1)}_j=c/\eps_{j+1},\ \ \forall 1\leq j\leq s-1, \ \ \text{and}\ \ \pi^{(1)}_s=c/\eps_{1},
\ee
where $c\in \C^*$ is arbitrary,  which after normalization and with $\eps_j=1/d_j$ and (\ref{periodN}) yields  (\ref{pidft}).
\qed

As a consequence, with the DFT scattering matrix and under the assumption that $\Sigma$ is connected, we obtain that the asymptotic state in Ces\`aro mean reads
\be\label{asstacyc}
\sum_{y\in {\rm Cyc}} |\chi(y)\ket\bra \chi(y)| \,\frac{d_{N(y)}}{\sum_{z\in {\rm Cyc}}d_z}
\ee
with 
\be\label{chiN}
 |\chi(y)\ket=\frac{1}{\sqrt{d_{N(y)}}}|N(y)\ket+e^{i\beta_y}\sqrt{1-\frac{1}{d_{N(y)}}}|y\ket.
\ee
The asymptotic probability of presence  at $x$ in Ces\`aro mean is $\pi_x=\frac{d_{N(x)}}{\sum_{y\in {\rm Cyc}}d_y}$ if $x\in {\rm Cyc}$ and zero otherwise.
\medskip

Finally, if $\#{\rm CC}(\Sigma)\geq 2$, we have $V=C_1\cup \dots \cup C_r$ and we consider $\C^{|V|}=\oplus_{i=1}^r\C^{|C_i|}$, where $\C^{|C_i|}$ is associated with the vertices in $C_i$. 
Therefore we can write  
\be\label{decomP} P=\bigoplus_{i=1}^r P_{i}, 
\ee where ${P_{i}}$ is the restriction of $P$ to the invariant subspaces $\C^{|C_i|}$ for $i\in \{1, 2, \dots, r\}$. 
Moreover, the results derived above in case $\Sigma$ is connected apply to each of the $P_i$'s. In particular, for each $i\in \{1, 2, \dots, r\}$, there exists a unique invariant probability row vector $\pi_i\in \C^{|C_i|}$ for $P_i$, supported on ${\rm Cyc}_i$, the cycle of $C_i$, according to Lemma \ref{piconnected}. In the following, we view $\pi_i\in \C^{|C_i|}\subset \C^{|V|}$ as a probability vector on $\C^{|V|}$ and, accordingly, $P_i$ as a matrix on $\C^{|V|}$. We also introduce the corresponding column vectors  ${\bf 1}_i\in \C^{|V|}$ supported on ${\rm Cyc}_i$ with nonzero components equal to one.

\medskip

From the decomposition (\ref{decomP}) and Lemma \ref{piconnected}, we deduce 
\be\label{CesaropnMany}
\frac1N\sum_{n=0}^{N-1}P^n=\sum_{1\leq i\leq r} {\bf 1}_i \pi_i+O(N^{-1}),
\ee
where the sum of the right hand side is the spectral projector of $P$ on its eigenvalue 1, which is of dimension $r=|{\rm CC}(\Sigma)|$.
Gathering these observations, we obtain in the general case for DFT scattering matrices:
\begin{lem}
For $\rho_0\in \cD\cM(l^2(V))$, the asymptotic probability of presence  at $x$, $Q^{\rho_0}_n(x)$ given by  (\ref{expqnrho}), satisfies
\be
\lim_{N\ra \infty}\frac1N\sum_{n=0}^{N-1}Q^{\rho_0}_n(x)=\begin{cases}
r_0({\rm Cyc_i})\pi_i(x)& \text{\em if } x\in {\rm Cyc}_i , \\
0 & \text{\em otherwise,}\end{cases} 
\ee 
where $\pi_i$ is given in Lemma \ref{piconnected} the index $i\in\{1,2,\dots, r\}$ is uniquely determined by $x$ so that $\pi_i(x)>0$, $r_0$ is defined by $\rho_0$ in (\ref{rorho}) and $r_0({\rm Cyc_i})=\sum_{y\in {\rm Cyc_i}}{r_0(y)}$.

Moreover, the corresponding asymptotic state  (\ref{evolPP}) applied to $\rho_0$ satisfies
\be
\lim_{N\ra \infty}\frac1N\sum_{n=0}^{N-1}\Psi_{\cS}^n(\rho_0)=\sum_{i=1}^r {r_0({\rm Cyc_i})}\sum_{y\in {\rm Cyc_i}} |\chi(y)\ket\bra \chi(y)| \pi_i(y).
\ee
\end{lem}
\begin{rem} i) The difference between the finite $N$ and asymptotic Ces\`aro mean probabilities and states  is again $O(N^{-1})$.\\ 
ii) Because $P$ is stochastic, (\ref{CesaropnMany}) implies in the limit $N\ra \infty$
\be
1=r_0 \sum_{i=1}^r  {\bf 1}_i \pi_i \,{\bf 1}=  \sum_{i=1}^r  r_0({\rm Cyc_i}),
\ee
so that $\{r_0({\rm Cyc_i})\}_{1\leq i\leq r}$ defines a probability on the cycles of $\Sigma$. \\
iii) The asymptotic state is thus the expectation with respect to the probability on the cycles above of the asymptotic states (\ref{asstacyc}) for each cycle.  
\end{rem}

\subsection{An infinite Graph with DFT scattering matrices}

To complete the picture, we consider a simple example of infinite graph, in the framework of the discrete Fourier transform. \\

{Let $(\Z, E)$ be the graph with edge set $E=\cup_{x\in \Z}\{[x,x+1]\}$ between each consecutive integers. We consider $G=(\Z, \tilde E)$ as the previous graph to which we add an extra edge between the vertices $0$ and $+1$: $\tilde E=E\cup \{[0,1]_1\}$, see Figure \ref{figGZplus}.} 
\begin{figure}[h]
	\centering
	\begin{tikzpicture}
		% Define the vertices with smaller circles, white fill, and black text
		\foreach \x in { -2, -1, 0, 1, 2, 3} {
			\node[circle, fill=white, draw, minimum size=0.8cm, inner sep=0pt] (\x) at (1.5*\x, 0)
			{\textcolor{black}{\x}};
		}
		
		% Straight edges (no arrows)
		\foreach \x in { -2, -1} {
			\draw[line width=0.3mm] (\x) -- (\the\numexpr\x+1\relax);
		}
		\draw[line width=0.3mm] (1) -- (2);
		\draw[line width=0.3mm] (2) -- (3);
		
		% Curved edges (no arrows)
		\draw[line width=0.3mm] (1) to[bend left=30] (0);
		\draw[line width=0.3mm] (0) to[bend right=-30] (1);
		
		% Add three dots at the left and right ends
		\node at (-4, 0) {\(\cdots\)};
		\node at (5.5, 0) {\(\cdots\)};
	\end{tikzpicture}
	\caption{\small The infinite graph $G=(\Z, \tilde E)$. }
	\label{figGZplus}
\end{figure}
Take the scattering matrices given by \eqref{FFT}  at each vertex, and consider a labelling of the adjacent vertices of $x\in V=\Z$, see \eqref{Llabel}, such that
\be
\begin{cases}
x_1= x+1 & \text{if } \ x\leq 0,\\
x_1=x-1 &  \text{if } \ x\geq 1.
\end{cases}
\ee

Accordingly, the only off diagonal nonzero entries of $P$ are, see (\ref{stochB}),
\be
P_{-1 \, 0}= 1/3, \ P_{0 \, 1}= 1/3,\  P_{2 \, 1}=1/3, \  P_{1 \, 0}=1/3 
\ee
and 
\be
\begin{cases}
P_{x-1 \, x}=1/2 &  \text{if } \ x\leq -1,\\
P_{x+1 \, x}=1/2 &  \text{if } \ x\geq 2,
\end{cases}.
\ee
These relations define the map $N$ (\ref{mapn}),
and the diagonal elements are determined by the stochasticity condition. This yields the infinite matrix, where the dots mark the main diagonal, with underlined $00$ element $2/3$:
\be\label{Pinf}
P = 
\begin{pmatrix}
\ddots &  1/2 & & & & & & & \\
 & 1/2 &  1/2 &  &  &  &  & &  \\
& & 1/2  &  1/2 &  &  &  &   &\\
&  & & 2/3 &  1/3 & 0 &  &  & \\
&  & & 0 & \underline{2/3} & 1/3 &  &  & \\
&  &  & 0 & 1/3 & 2/3 &  &  &\\
&  &  &  &  & 1/3 & 2/3 &  & \\
&  &  &  &  &  & 1/2 &  1/2 & \\
& & & & & & &1/2 & \ddots  \\
\end{pmatrix}
\ee

To this reducible matrix corresponds the situation depicted as Figure \ref{figdiminf}
\begin{figure}[h]
\centering
\begin{tikzpicture}
    % Define the vertices with smaller circles, white fill, and black text
    \foreach \x in { -2, -1, 0, 1, 2, 3} {
        \node[circle, fill=white, draw, minimum size=0.8cm, inner sep=0pt] (\x) at (1.5*\x, 0)
        {\textcolor{black}{\x}};
    }
    
        \foreach \x in { -2, -1} {
        \draw[->, line width=0.3mm] (\x) -- (\the\numexpr\x+1\relax);
    }
    \draw[<-, line width=0.3mm] (1) -- (2);
    \draw[<-, line width=0.3mm] (2) -- (3);
    \draw[->, line width=0.3mm] (1) to[bend left=30] (0);
     \draw[->, line width=0.3mm] (0) to[bend right=-30] (1);
    
    % Add three dots at the left and right ends
    \node at (-4, 0) {\(\cdots\)};
    \node at (5.5, 0) {\(\cdots\)};
\end{tikzpicture}
\caption{\small The infinite graph $\Sigma$ associated with the stochastic matrix  $P$ (\ref{Pinf}). }
  \label{figdiminf}

\end{figure}

It is readily verified that the only probability vector $\pi\in l^1(\Z)$ such that $\pi P=\pi$ has nonzero components
\be\label{defpi3}
\pi_{0}=1/2, \ \ \pi_1=1/2.
\ee
Incidentally, this is the invariant probability vector of the central $2\times 2$ stochastic block 
\be\label{defF}
G=\begin{pmatrix} 2/3 & 1/3 \\ 1/3 & 2/3 
\end{pmatrix}.
\ee
\begin{lem}\label{tech}
There exists $\gamma >0$ and $c>0$ such that for any $x, y\in \Z$, there is an $N(x,y)\in \N$ so that for $n\geq N(x,y)$ it holds
\be
|P_{x y }^n-\pi_y|\leq ce^{-\gamma n},
\ee
with the understanding that $\pi_y=0$ for $y\leq -1$ or $y\geq 2$.
\end{lem}
The proof of this technical lemma is provided in Appendix D.

As a consequence, thanks to Theorem \ref{dynpsiV}, we have in this infinite dimensional example
\be
\Psi_{\cS}^n(\cdot)=\sum_{y=0}^1 \frac12|\chi(y)\ket\bra \chi(y)| \, \tr(\cdot) + O(e^{-\gamma n}),
\ee
in where $\chi(y)$ are defined by \eqref{chiN}, in the weak sense. This result is in keeping with the finite dimensional result Lemma \ref{piconnected}.

\medskip

\noindent {\bf Acknowledgment: } This work is partially supported by the French National Research Agency in the framework of the
"France 2030" program (ANR-11-LABX-0025-01) for the LabEx PERSYVAL, and by the grant Dynacqus ANR-24-CE40-5714-02.

\section{Appendices}

{\bf Appendix A:} This appendix is devoted to the proof of Theorem \ref{findim}.
We denote $U_{\alpha}$ by $U$ for short below, and  we occasionally drop the identity symbols in expressions like $F_{jj}-z\un_j$, $z\in \C$, writing $F_{jj}-z$ instead.

\begin{proof}: 
Let us assume that $U \psi=\lambda\psi$, $\lambda\in {\mathbb S}^1$, $\psi\neq0$.  With $\psi= \begin{pmatrix}\psi_1 \\ \psi_2\end{pmatrix}$, where $\psi_1=\Pi \psi$, $\psi_2=(\un -\Pi)\psi$,  and 
(\ref{decomp}), (\ref{Calpha}), we have, equivalently
\be\label{eve}
\left\{\begin{matrix} 
F_{11}\psi_1+\e^{i\alpha}F_{12}\psi_2=\lambda \psi_1 \phantom{i}\\
F_{21}\psi_1+\e^{i\alpha}F_{22}\psi_2=\lambda \psi_2 .  
\end{matrix}\right.
\ee
Composing the eigenvalue equation with $F$, we obtain similarly
\be\label{feve}
\left\{\begin{matrix} 
\lambda F_{11}\psi_1+\lambda F_{12}\psi_2= \psi_1 \phantom{xxi}\\
\lambda F_{21}\psi_1+\lambda F_{22}\psi_2=\e^{i\alpha} \psi_2 .  
\end{matrix}\right.
\ee
Expressing $F_{12}\psi_2$ from the first equation \eqref{feve} and inserting in the first equation \eqref{eve}, we get
\begin{align}\label{sp11}
F_{11}\psi_1=\frac{\lambda^2-\e^{i\alpha}}{\lambda(1-\e^{i\alpha})} \psi_1=\mu \psi_1,
\end{align}
while using the last equations \eqref{feve},  \eqref{eve} to eliminate  $F_{21}\psi_1$ yields similarly
\begin{align}\label{sp22}
F_{22}\psi_2=-\frac{\lambda^2-\e^{i\alpha}}{\lambda(1-\e^{i\alpha})} \psi_2=-\mu \psi_2.
\end{align}
We also note that $\psi_1=0$ implies  $\psi_2\neq 0$ and, with $F_{22}=F_{22}^*$,
\be
F_{22}\psi_2=\e^{-i\alpha}\lambda \psi_2=\overline{\lambda} \e^{i\alpha}\psi_2 \  \Rightarrow \ \lambda = \pm \e^{i\alpha} \ \& \ \mu = \mp 1.
\ee 
Similarly, $\psi_2=0$ implies  $\psi_1\neq 0$ and 
\be
F_{11}\psi_1=\lambda \psi_1=\overline{\lambda} \psi_1 \  \Rightarrow \ \lambda = \mu = \pm 1.
\ee 
Assume $\lambda\not\in\{ \pm \e^{i\alpha}\}$. Since $\psi_1\neq 0$ by our choice of $\lambda$, (\ref{sp11}) implies $\mu\in \sigma(F_{11})$. Also, if $\lambda\not\in\{\pm1\}$, then $\psi_2\neq 0$ so that by (\ref{sp22}), $-\mu\in \sigma(F_{22})$.
\medskip

Consider here $\alpha\in (-\pi, \pi)\setminus\{0\}$ and suppose $\lambda=\pm1$. Then $\mu=\pm1$, so that $\| F_{11}\psi_1\|=\|\psi_1\|\neq 0$. Since $F$ is unitary and $\Pi, \un -\Pi$ are orthogonal projectors, we deduce that
\be\label{propunitF}
\|\psi_1\|^2=\|F\psi_1\|^2=\|F_{11}\psi_1\|^2+\|F_{21}\psi_1\|^2=\|\psi_1\|^2+\|F_{21}\psi_1\|^2 \ \Rightarrow \ F_{21}\psi_1=0.
\ee
Hence, the second equation (\ref{eve}) implies $\psi_{2}=0$, since $F_{22}\mp e^{-i\alpha}\un_2 $ is invertible, as $F_{22}=F_{22}^*$. Altogether, $U\psi=\pm\psi$ implies $\psi=\Pi \psi=\psi_1\neq 0$ and $F_{11}\psi_1=\pm \psi_1$, {\it i.e.} 
\be\label{1incl}
\ker (U\mp\un)\subset \ker (F_{11}\mp\un_1).
\ee
If  $\lambda=\pm \e^{i\alpha}$, then $\mu=\mp 1$ and $\psi_2\neq 0$. Hence $\pm1\in \sigma(F_{22})$, and by an argument similar to the one above, $F_{12}\psi_2=0$, so that $\psi_1=0$ by the first equation (\ref{eve}), since $F_{11}\mp e^{i\alpha}\un_1$ is invertible. Hence, $U\psi=\pm \e^{i\alpha}\psi$ implies $\psi=(\un-\Pi)\psi=\psi_2$ and  $F_{22}\psi_2=\pm\psi_2$, {\it i.e. } 
\be\label{2incl}
\ker (U\mp\e^{i\alpha}\un)\subset \ker (F_{22}\mp \un_2).
\ee

Conversely, assume
\be
F_{11}\psi_1=\mu \psi_1,\ee 
with, for now, $\mu  \in (-1,1)$, $\psi_1\neq 0$. Then $\phi_{\alpha}^{-1}(\{\mu\})=\{\lambda_+, \lambda_-\}$, where $\lambda_\epsilon\neq \pm \e^{i\alpha}$, $\epsilon\in\{+, -\}$. We construct two eigenvectors of $U$ associated with $\lambda_+$ and $\lambda_-$ respectively, as follows. With $\lambda$ denoting $\lambda_+$ or $\lambda_-$, we set 
\be\label{defpsi2}
\psi_2=-\e^{-i\alpha}(F_{22}-\lambda\e^{-i\alpha})^{-1}F_{21}\psi_1,
\ee
which is well defined and consider
\be\label{conev}
\psi=\psi_1+\psi_2=\big(\un -\e^{-i\alpha}(F_{22}-\lambda\e^{-i\alpha})^{-1}F_{21}\big)\psi_1.
\ee 
Therefore, using $F_{22}\psi_2=-\e^{-i\alpha}F_{21}\psi_1+\lambda \e^{-i\alpha}\psi_2$,
\begin{align}\label{prepev}
U\psi & = F_{11}\psi_1+\e^{i\alpha}F_{12}\psi_2+ F_{21}\psi_1+\e^{i\alpha}F_{22}\psi_2\nonumber \\
&= \mu\psi_1 +\lambda \psi_2 -F_{12}(F_{22}-\lambda\e^{-i\alpha})^{-1}F_{21}\psi_1.
\end{align}
It thus remains to show that $-F_{12}(F_{22}-\lambda\e^{-i\alpha})^{-1}F_{21}\psi_1=(\lambda-\mu)\psi_1$ to get eventually $U\psi=\lambda\psi$.
We proceed by considering the Schur complement of the $22$ block of the invertible operator
\be
F-\lambda\e^{-i\alpha}=\begin{pmatrix} F_{11} -\lambda\e^{-i\alpha}& F_{12} \\ F_{21} & F_{22}-\lambda\e^{-i\alpha}\end{pmatrix}, 
\ee
given by 
\be\label{schur}
S_{1}=F_{11} -\lambda\e^{-i\alpha} - F_{12}(F_{22} -\lambda\e^{-i\alpha})^{-1}F_{21} : \cH_1 \ra \cH_1.
\ee
Recall that since $F-\lambda\e^{-i\alpha}$ and $F_{11} -\lambda\e^{-i\alpha}$ are both invertible, $S_{1}$ is invertible as well and such that
\begin{align}
(F-\lambda\e^{-i\alpha})^{-1}=\begin{pmatrix} S_{1}^{-1} & * \\ * & *
\end{pmatrix},
\end{align}
where  a $*$ denotes a block we do not need to know here. Then, since the resolvent of $F=F^*=F^{-1}$ reads  for $z\neq \pm1$
\be\label{invfz}
(F-z)^{-1}=\frac{F+z}{1-z^2},
\ee
we get an alternative expression for $S_{1}^{-1}$ by considering 
\begin{align}
\frac{F+\lambda\e^{-i\alpha}}{1-\lambda^2\e^{-2i\alpha}}=\begin{pmatrix} S_{1}^{-1} & * \\ * & *
\end{pmatrix},
\end{align}
which yields
\be\label{altschur}
S_{1}=(1-\lambda^2\e^{-2i\alpha})(F_{11}+\lambda \e^{-i\alpha})^{-1}.
\ee
Altogether, we deduce from \eqref{schur} and \eqref{altschur}
\be
 - F_{12}(F_{22} -\lambda\e^{-i\alpha})^{-1}F_{21} =(1-\lambda^2\e^{-2i\alpha})(F_{11}+\lambda \e^{-i\alpha})^{-1}-(F_{11} -\lambda\e^{-i\alpha}),
\ee
which applied to $\psi_1$ gives
\begin{align}
 - F_{12}(F_{22} -\lambda\e^{-i\alpha})^{-1}F_{21} \psi_1&=\frac{(1-\lambda^2\e^{-2i\alpha})}{\mu+\lambda \e^{-i\alpha}}\psi_1-(\mu-\lambda \e^{-i\alpha})\psi_1\nonumber\\
 &=-\mu \psi_1+\frac{1+\lambda \e^{-i\alpha}\mu}{\mu+\lambda \e^{-i\alpha}}\psi_1.
 \end{align}
 Using $\mu = \ffi_\alpha(\lambda)$ under the form $\e^{-i\alpha}\lambda^2 + \mu \lambda =1+\e^{-i\alpha} \mu \lambda$, the last fraction reduces to $\lambda$ which, together with \eqref{prepev}, shows $U\psi=\lambda\psi$. 
 
 Also, if $\dim \ker (F_{11}-\mu \un_1)=m_1^\mu>1$, and $\{\psi_1^{(1)}, \dots, \psi_1^{(m_1^\mu)}\}$ form a basis of this eigenspace, the previous construction yields for each $\psi_1^{(j)}$ two eigenvectors of $U$ associated with $\lambda_-$ and $\lambda_+$ denoted by $\psi_\pm^{(j)}$ via (\ref{conev}). Then $\{\psi_\pm^{(1)}, \dots, \psi_\pm^{(m_1^\mu)}\}$ are linearly independent since their projections $\{\Pi \psi_\pm^{(j)}\}$ are, so that
 \be\label{dimker1}
 \dim \ker (U-\lambda_\pm\un)\geq \dim \ker (F_{11}-\mu \un_1).
 \ee
 
With $\mu\in (-1,1)$ again, so that $\lambda_\pm\not\in \{\pm 1\}$, we can start from 
\be
F_{22}\psi_2=-\mu \psi_2
\ee 
with $\psi_2\neq 0$ to construct eigenvectors for $U$ associated with $\lambda\in\{\lambda_-, \lambda_+\}=\ffi_\alpha^{-1}(\{\mu\})$ in a similar way as above. Indeed, setting
\be
\psi_1=-\e^{i\alpha}(F_{11}-\lambda)^{-1}F_{12}\psi_2,
\ee
which is well defined, we check using the Schur complement of the $11$ block of the invertible operator $F-\lambda$, (\ref{invfz}), and the relation $\ffi_\alpha(\mu)=\lambda$, that  
\be\label{conev2}
\psi=\psi_1+\psi_2=\big(\un -\e^{-i\alpha}(F_{11}-\lambda)^{-1}F_{12}\big)\psi_2
\ee 
satisfies $U\psi=\lambda\psi$, for $\lambda=\lambda_-$ and $\lambda=\lambda_+$. And by an argument similar to that just given,
\be\label{dimker2}
 \dim \ker (U-\lambda_\pm\un)\geq \dim \ker (F_{22}+\mu \un_2).
\ee
  
We consider now the cases $\mu=\pm1$ excluded above, assuming $F_{11}\psi_1=\pm \psi_1\neq 0$. Equation (\ref{propunitF}) implies $F_{21}\psi_1=0$, and we get from 
 (\ref{prepev}) $U\psi_1=F_{11}\psi_1=\pm\psi_1$, {\it i.e.}  $\lambda_\pm=\pm 1 \in \sigma(U)$, where $\pm 1\in \ffi_\alpha^{-1}(\{\pm 1\})$. This yields
 \be\label{inc1} 
 \ker (F_{11}\mp\un_1) \subset \ker (U\mp\un).
 \ee
 Note that the vector $\psi_2$ given in (\ref{defpsi2}) is well defined  for $\lambda_\pm=\pm 1$ and equals zero, as it should. 
 
Next, we consider $\psi_2\in \ker (F_{22}\mp \un_2)$. The equation corresponding to (\ref{propunitF}) with indices exchanged yields $F_{12}\psi_2=0$, so that
(\ref{prepev}) implies $U\psi_2=\pm\e^{i\alpha}\psi_2$ and
\be\label{inc2} 
 \ker (F_{22}\mp \un_2)\subset  \ker (U\mp \e^{i\alpha}\un).
\ee
Together with (\ref{1incl}) and (\ref{2incl}), the last two inclusions prove (\ref{eqker}).

At this point, we note with (\ref{dimprojpi}) that
\begin{align}
&\sum_{\mu\in\sigma(F_{11})}\dim\ker (F_{11}-\mu)=\dim \Pi=|V|\nonumber\\
&\sum_{\mu\in\sigma(F_{22})}\dim\ker (F_{22}-\mu)=\dim(\un -\Pi)=|D|-|V|,
\end{align}
and that (\ref{dimker1}), (\ref{dimker2}) yield for $\mu\in(-1,1)$
\be\label{ineqker2}
 \dim \ker (U-\lambda_-\un)+ \dim \ker (U-\lambda_+\un)\geq \dim \ker (F_{11}-\mu \un_2)+\dim \ker (F_{22}+\mu \un_2).
\ee
With (\ref{eqker}), we deduce
\begin{align}\label{argue}
|D|&=\sum_{\lambda\in\sigma(U)}\dim\ker (U-\lambda) \\
&= \sum_{\lambda\in\sigma(U)\setminus\{\pm1, \pm\e^{i\alpha}\}}\dim\ker (U-\lambda)+\sum_{\lambda\in\{\pm1, \pm\e^{i\alpha}\}}\dim\ker (U-\lambda)  \nonumber\\
&\geq \sum_{\mu\in\sigma(F_{11})}\dim\ker (F_{11}-\mu)+\sum_{\mu\in\sigma(F_{22})}\dim\ker (F_{22}-\mu)=|D|. \nonumber
\end{align}
Therefore, we have equality in (\ref{ineqker2}) for all $\mu\in (-1,1)$. In turn, with (\ref{dimker1}), (\ref{dimker2}) this implies 
\be
\dim \ker (U-\lambda_\pm\un)=\dim (\ker F_{11}-\mu \un_2)=\dim (\ker F_{22}+\mu \un_2).
\ee
This ends the proof in the case $\alpha\in (-\pi,\pi)\setminus\{0\}$. 
\medskip

When $\alpha =\pi$, the potential eigenvalues $\pm\e^{i\alpha}$ and $\mp1$ for $U$ are degenerate, which requires revisiting the argument. 
Assuming $U\psi=\lambda\psi$ with $\psi\neq 0$ and $\lambda=\pm1$, we do not have necessarily $\psi_1\neq 0$. However we derive from (\ref{eve}) and (\ref{feve}) that 
\be
F_{11}\psi_1\mp\psi_1=0 \ \ \text{and}\ \
F_{22}\psi_2\pm\psi_2=0,
\ee
which implies
\be
\Pi \ker(U\mp\un)\subset \ker (F_{11}\mp\un_1), \ \ \text{and} \ \ (\un-\Pi) \ker(U\mp\un)\subset \ker (F_{22}\mp\un_2).
\ee
Conversely, if $0\neq \psi_1\in  \ker (F_{11}\mp\un_1)$, then $F_{21}\psi_1=0$, see \eqref{propunitF}, so that $U\psi_1=\pm\psi_1.$ Similarly, $0\neq \psi_2\in  \ker (F_{22}\mp\un_2)$ implies $U\psi_2=\mp\psi_2$. In other words,
\be
\ker (F_{11}\mp\un_1)\subset  \ker(U\mp\un), \ \ \text{and} \ \ \ker (F_{22}\pm\un_2)\subset \ker(U\mp\un),
\ee
where the subspaces $\ker (F_{11}\mp\un_1)$ and $\ker (F_{22}\pm\un_2)$ are orthogonal. Hence the validity of (\ref{eqkerpi}).
Since \eqref{ineqker2} also holds for $\alpha=\pi$ and $-1<\mu<1$, we can argue as \eqref{argue} to deduce the validity of (\ref{30}) for $\alpha=\pi$.
\end{proof}\qed

\medskip 

{\bf Appendix B:} Let us prove Lemma \ref{lemwoes}. 
We first observe that (\ref{phidmat}) implies
\be\label{evenonly}
\Phi^{\rm Diag}(e_{2x+1})=\Phi^{\rm Diag}(e_{2(x+2)}),
\ee
which allows us to focus on ${\Phi^{\rm Diag}}^n(e_{2x})$, for $n\in \N$. Consequently,
\be
{\Phi^{\rm Diag}}^2(e_{2x})=\frac12\big({\Phi^{\rm Diag}}(e_{2(x-1)})+{\Phi^{\rm Diag}}(e_{2(x+1)})\big).
\ee
By induction one gets for any $n\in N$
\be
{\Phi^{\rm Diag}}^n(e_{2x})=\frac1{2^{n-1}}\sum_{k=0}^{n-1}\binom{n-1}{k} {\Phi^{\rm Diag}}(e_{2(x-(n-1)+2k)}),
\ee
where 
\be
{\Phi^{\rm Diag}}(e_{2(x-(n-1)+2k)})=\frac12\big(e_{2(x-n+2k)} + e_{2(x-n+2k)+1} \big).
\ee
Therefore, the matrix elements of ${\Phi^{\rm Diag}}$ read, using \eqref{evenonly},
\begin{align}
&{\Phi^{\rm Diag}}^n_{2y \,2x}={\Phi^{\rm Diag}}^n_{2y+1\, 2x}=\frac1{2^n}\binom{n-1}{\frac{n+y-x}{2}} \nonumber\\
&{\Phi^{\rm Diag}}^n_{2y\,  2x+1}={\Phi^{\rm Diag}}^n_{2y\,  2x+1}=\frac1{2^n}\binom{n-1}{\frac{n+y-x}{2}-1}.
\end{align}
with the understanding that $\binom{r}{s}=0$ if $r<s$, or $(r,s) \not \in \N\times \N$.   
It remains to establish the asymptotics for $n$ large of the right hand sides. Consider the (nonzero) binomial coefficient
\be
\binom{N}{\frac{N+\Delta}{2}}=\frac{N!}{\frac{N+\Delta}{2}!\frac{N-\Delta}{2}!},
\ee
for $\Delta$ fixed, as $N\ra\infty$, with $N$ and $\Delta$ of identical parity. Using Stirling's formula $N!=\sqrt{2\pi N}(N/e)^N(1+o(1))$, we infer
\be
\frac{1}{2^{N+1}}\binom{N}{\frac{N+\Delta}{2}}=\frac{1}{\sqrt{2\pi N}}+o(1)\leq c/\sqrt{N}, \ \ \forall N\geq N_0(\Delta),
\ee
which proves the statement. \qed

\medskip 

{\bf Appendix C:}
This appendix sketches an alternative spectral argument to show Corollary \ref{corpsin}.  
Since the stochastic matrix $P$ satisfies  $\|P^n\|_{\infty}\leq 1$ for all $n\in \N$, where $\| M \|_\infty=\max_{1\leq i\leq d}\sum_{j=1}^n |M_{ij}|$, for any $M\in M_{d}(\C)$, and admits  ${\bf 1}$ as an invariant vector, it has spectral radius $1$. Moreover, all its eigenvalues of modulus 1 are semi-simple. In case the eigenvalue 1 is simple, the rank one matrix ${\bf 1} \pi$ is the corresponding spectral projector such that $\pi P=\pi$. Since the eigenvalues of $P$ have either modulus strictly smaller than 1, or have modulus one and are semi-simple, the Ces\`aro mean (\ref{Cesaropn}) holds thanks to functional calculus, see \cite{Ka}. In case no other eigenvalue than 1 has modulus 1, one gets $P^n={\bf 1} \pi+O(e^{-\gamma n})$, $\gamma>0$, by functional calculus again. Finally, noting that $P^n$ and its Ces\`aro mean are stochastic matrices, the $xy$ matrix element of (\ref{Cesaropn}) satisfies $0\leq \frac1N\sum_{n=0}^{N-1}P^n_{xy}=\pi_y+O(N^{-1})$, which yields $\pi_y\geq 0$, in the limit $N\ra \infty.$

When $\dim \ker (P-1)>1$,  the same result holds, with the spectral projector $P_{\{1\}}$ onto $ \ker (P-1)$  in place of ${\bf 1} \pi$.
\medskip

{\bf Appendix D:} We prove Lemma \ref{tech} here. 
Let us decompose the operator \eqref{Pinf} as 
\be\label{defC}
P=B+C, \ \text{ where } \  C=\frac13(|-1\ket\bra 0|+|2\ket\bra 1|),
\ee
and $B$  has the block decomposition 
\be
B=\begin{pmatrix} \sigma_- & 0 & 0 \\ 0 & G & 0\\ 0 & 0 & \sigma_+
\end{pmatrix} = \sigma + G.
\ee
Here $G$ is defined in (\ref{defF}) and $\sigma_\pm$ correspond, up to a factor $1/2$ and except for one diagonal element, to the identity plus a one sided shift. 
The immediate properties
\be
\sigma G = G\sigma = 0, \ \ C^2=0, \  \ C\sigma ^n=0, \ \ G^n C=0,
\ee
yield
\be\label{decsum}
P^n=(B+C)^n 
=B^n+\sum_{k=0}^{n-1}\sigma^kCG^{n-(k+1)}.
\ee
Moreover, this implies that the general term of the sum satisfies
\be
\bra x | \sigma^kCG^{n-(k+1)} y\ket =0 \ \ \text{if} \ \ 0\leq x \leq 1 \ \text{or} \ y\in \{0, 1\}^{C},
\ee
which with  the block structure of $B^n$ shows that $P^n_{x y}=0$ for $x\geq 0$ and $y\leq -1$, or for $x\leq 1$ and $y\geq 2$.

We check by induction that for $y\geq 3$
\be\label{decal}
P^n |y\ket = B^n |y\ket=\sigma^n |y\ket= \frac{1}{2^n}\sum_{k=0}^n\begin{pmatrix}n \\ k\end{pmatrix} |y+k\ket,
\ee 
while for $y\leq -2$, the same formula holds with $|y-k\ket$ in place of $|y+k\ket$. Hence 
$P^n_{x y}=0$ for $x<y$ and $y\geq 3$, and for $x>y$ and $y\leq -2$.

For $y=2$, the expression is a little different due to the value $2/3$ of the entry $22$  of $\sigma_+$. By induction
\be\label{morbin}
P^n|2\ket=\Big(\frac23\Big)^n\Big\{|2\ket+\frac12\sum_{k=0}^{n-1} \Big(\frac23\Big)^{-k-1} P^k|3\ket\Big\},
\ee
while the same expression holds for $P^n|-1\ket$, with $|-2\ket$ in place of $|3\ket$.

The spectral decomposition of the $2\times 2$ block $G$ (\ref{defF}) reads
\be\label{specdecG}
G=Q_{\{1\}}+\frac13 Q_{\{1/3\}},  
\ee
where $\sigma(G)=\{1, 1/3\}$ and  the corresponding eigenprojectors are $Q_{\{1\}}= \frac12|{\bf 1}\ket\bra {\bf 1} |= {\bf 1} \pi $, with $\pi$ given in  (\ref{defpi3}), and 
$Q_{\{1/3\}}=\un -Q_{\{1\}}$.

\medskip

This yields for $0\leq x\leq 1, 0\leq y\leq 1$
\be
P^n_{x y}=B^n_{x y}= G^n_{x y} = {Q_{\{1\}}}_{x y}+O(e^{-\gamma n})=\pi_y+O(e^{-\gamma n}),
\ee
for $0< \gamma \leq  \ln 3$. For $x, y \geq 3$, with $x\geq y$,
\be
P^n_{x y}=\sigma^n_{x y}=\frac{1}{2^n}\begin{pmatrix}n \\ x-y\end{pmatrix} \leq \frac{1}{2^n} \frac{n^{x-y}}{(x-y)!}\leq  e^{-\gamma n},
\ee
for $0<\gamma <\ln 2$, and $n\geq N'(x-y)$, for some $N'(x-y)>0$. A similar bound holds in case $x, y \leq -2$. Equation (\ref{morbin}) implies 
$P^n_{2 \, 2}= P^n_{-1 \, -1 }=(2/3)^n\leq e^{-n\ln(3/2)}$, and together with (\ref{decal})
 for $x\geq 3$, 
\begin{align}\label{xg3}
P^n_{x 2}&= \Big(\frac23\Big)^n\frac12\sum_{k=0}^{n-1} \Big(\frac23\Big)^{-k-1} P^k_{x 3}=  \Big(\frac23\Big)^{n-1}\frac12\sum_{k=x-3}^{n-1} \Big(\frac34\Big)^{k}\binom{k}{x-3}.
\end{align}
Thanks to the relation for any $z\in \C$, s.t. $|z|<1$, $j\geq 0$,
\be\label{limbinoz}
\left| \sum_{k=j}^{M}z^k\begin{pmatrix}k \\ j\end{pmatrix} - \frac{z^j}{(1-z)^{j+1}} \right| \leq ce^{-\gamma_z M},
\ee
for $0<\gamma_z < |\ln |z||$, $c=c(\gamma_z)$ and $M\geq N_0(j)$, we eventually obtain for $n\geq N''(x)$,
\be
P^n_{x 2}= \Big(\frac23\Big)^{n}(3^{x-2}+O(e^{-\gamma_z n}))\leq (3^{x-2}+1) e^{-n\ln(3/2)}\leq C_0 e^{-\gamma n},
\ee
where $0<\gamma<\ln(3/2)$, and $C_0$ is independent of $x$. Note that the error term can be made uniform in $x$ 
by lowering the value of $\gamma$, and choosing $n$ large enough, in an $x$ dependent way. 
A similar estimate holds for $P^n_{x\, -1}$, $x\leq -2$.

We are left to consider the sum in (\ref{decsum}) for $x\in \{0,1\}^{C}$ and $y\in \{0, 1\}$.  Assume $x\geq 2$, the case $x \leq -1$ being similar.\\
For  $x\geq 3$, we have thanks to (\ref{defC}) and (\ref{xg3})
\begin{align}\label{leftsum}
\sum_{k=0}^{n-1}\bra x | \sigma^kCG^{n-(k+1)} y\ket&=\sum_{k=x-2}^{n-1}\frac13 \Big(\frac{2}{3}\Big)^{k-1}\frac12\sum_{l=x-3}^{k-1} \Big(\frac{3}{4}\Big)^{l}\begin{pmatrix}l \\ x-3\end{pmatrix}\bra 1 |G^{n-(k+1)} y\ket.
\end{align}
Inserting the spectral decomposition of $G$ into the scalar product, we get $\bra 1 | Q_{\{1\}} y\ket =\pi_y$ for the contribution stemming from the eigenprojector $Q_{\{1\}}$.
Exchanging the summation indices, the remaining double sum reads
\be
\frac14\sum_{l=x-3}^{n-2}  \Big(\frac{3}{4}\Big)^{l}\begin{pmatrix}l \\ x-3\end{pmatrix}\sum_{k=l+1}^{n-1}\Big(\frac{2}{3}\Big)^{k}=
\frac34\sum_{l=x-3}^{n-2}  \Big(\frac{3}{4}\Big)^{l}\begin{pmatrix}l \\ x-3\end{pmatrix}\Big\{\Big(\frac{2}{3}\Big)^{l+1}-\Big(\frac{2}{3}\Big)^{n}\Big\}.
\ee
Using (\ref{limbinoz}), the first contribution equals $1+O(e^{-\gamma n})$ while, the second is $O(e^{-\gamma n})$, for some $0<\gamma<\ln(3/2)$. Let us show that the contribution from the second spectral projector to (\ref{leftsum}) is exponentially decreasing. The corresponding scalar product in the right hand side of the sum reads
\be
\frac{3^{(k+1)}}{3^{n}} \bra 1| Q_{\{1/3\}} y\ket
\ee 
so that performing the same steps as above, the double sum is now of order
\begin{align}
\frac1{3^{n}} \sum_{l=x-3}^{n-2}  \Big(\frac{3}{4}\Big)^{l}\begin{pmatrix}l \\ x-3\end{pmatrix}\sum_{k=l+1}^{n-1}2^k&=\frac1{3^{n}} \sum_{l=x-3}^{n-2}  \Big(\frac{3}{4}\Big)^{l}\begin{pmatrix}l \\ x-3\end{pmatrix}(2^n-2^{l+1})\nonumber\\
&\leq \frac1{3^{n}} \sum_{l=x-3}^{n-2}  \Big(\frac{3}{4}\Big)^{l}\begin{pmatrix}l \\ x-3\end{pmatrix}2^n.
\end{align}
The last series is of order $(2/3)^n$ by  (\ref{limbinoz}), so that  we eventually get for $x\geq 3$, $y\in \{0,1\}$
\be
P^n_{x y}=\pi_y+ O(e^{-\gamma n}),
\ee
where $0<\gamma < \ln 2$, for $n\geq N'''(x)$. 

Eventually, we address the entry $P^n_{2y}$, $y\in \{0, 1\}$, which reads thanks to (\ref{decal}), (\ref{morbin}) and (\ref{specdecG})
\begin{align}
P^n_{2y}&=\sum_{k=0}^{n-1}\frac13\bra 2 | \sigma^k 2\ket \Big(\bra 1 | Q_{\{1\}}y \ket + \frac1{3^{n-(k+1)}}\bra 1 |  Q_{\{1/3\}}y\ket\Big)\nonumber\\
&=\sum_{k=0}^{n-1}\frac13 \Big(\frac23\Big)^k\Big(\pi_y + \frac{3^{k+1}}{3^{n}}\bra 1 |  Q_{\{1/3\}}y\ket\Big).
\end{align}
Computing the geometric series, one gets $P^n_{2y}=\pi_y+O(e^{-\gamma n})$, for $\gamma >0$, which finishes the proof of the lemma.
\qed


\begin{thebibliography}{99}

\bibitem[AAKV]{AAKV} A. Ambainis, D. Aharonov, J. Kempe, U. Vazirani, Quantum Walks on Graphs, {\it Proc. 33rd ACM STOC}, 50-59 (2001)

\bibitem[AAM+]{AAM+} A. Ahlbrecht, A. Alberti, D. Meschede, V. B. Scholz, A. H. Werner, and R. F. Werner, Molecular binding in interacting quantum walks, {\it New Journal of Physics}, {\bf 14}, 073050, (2012).


\bibitem[AJR]{AJR} S. Andr\'eys, A. Joye,  R. Raqu\'epas, Fermionic walkers driven out of equilibrium, 
{\it J. Stat. Phys.}, {\bf 184}, 14, (2021)


\bibitem[ABJ1]{ABJ1} Asch, J., Bourget, O., Joye, A.,  Localization Properties of the Chalker-Coddington Model. {\it Ann. H. Poincar\'e}, {\bf 11}, 1341--1373, (2010).

\bibitem[ABJ2]{ABJ2} Asch, J., Bourget, O., Joye, A., Dynamical Localization of the Chalker-Coddington Model far from Transition, {\it J. Stat. Phys.}, {\bf 147}, 194-205 (2012).

\bibitem[ABJ3]{ABJ3} Asch, J., Bourget, O., Joye, A., Spectral Stability of Unitary Network Models, {\it Rev. Math. Phys.}, {\bf 27}, 1530004, (2015).

\bibitem[ABJ4]{ABJ4} Asch, J., Bourget, O., Joye, A., Chirality induced Interface Currents in the Chalker Coddington Model ", 
{\it J. Spectral Theory}, {\bf 9}, 405-1429, (2019).

\bibitem[ABJ5]{ABJ5}Asch, J., Bourget, O., Joye, A., On stable quantum currents, {\it J. Math. Phys.}, {\bf 61}, 092104, (2020)

\bibitem[A]{A} Attal, S., Lectures in Quantum Noise Theory, Chapter 6: Quantum Channels, { \texttt http://math.univ-lyon1.fr/~attal/chapters.html}

\bibitem[AG-PS]{AG-PS} Attal, S., Guillotin-Plantard, N., and Sabot, C., Central limit theorems for open quantum random walks and quantum measurement records, {\it Ann. H. Poincar\'e} {\bf 16} 15-43, (2015).

\bibitem[AJP]{AJP} "Open Quantum Systems", Edt. by Attal, S., Joye, A., Pillet, C.-A., 
Springer {\it Lecture Notes in Mathematics}, {\bf 1880, 1881 \& 1882}, (2006). 


\bibitem[AL]{AL} R. Alicki, K. Lendi, Quantum Dynamical Semigroups and Applications, {\it Lect. Notes Phys.} {\bf 717} Springer (2007)

\bibitem[APSS1]{APSS1}
S Attal, F Petruccione, C. Sabot, I. Sinayskiy, Open quantum random walks,
{\it J. Stat. Phys.} {\bf 147} (4), 832-852, (2012).


\bibitem[ASW]{ASW} A. Ahlbrecht, V. B. Scholz, A. H. Werner, Disordered quantum walks in one lattice dimension, {\it J. Math. Phys.} {\bf 52} (2011).

\bibitem[BFS]{BFS} Bach, V., Fr\"ohlich, J., Sigal, I. M., Quantum Electrodynamics of Confined Nonrelativistic Particles, {\it Adv. Math.} {\bf 137}, 299–395 (1998).
  

\bibitem[BHJ]{BHJ} Bourget, O., Howland, J.S., Joye, A.: Spectral Analysis of Unitary Band Matrices, {\it Commun. Math. Phys.} {\bf 234}, 191-227 (2003).

\bibitem[BJPP]{BJPP} T. Benoist, V. Jaksic, Y. Pautrat, C.-A. Pillet,  On entropy production of repeated quantum measurements I. General theory. {\it Commun. Math. Phys.},  {\bf 357}, 77-123. (2018)


\bibitem[BM]{BM} H. Boumaza, L. Marin, Absence of absolutely continuous spectrum for random scattering zippers. {\it J. Math. Phys.}, {\bf 56}, 2015.

\bibitem[CC]{CC} Chalker, J. T., Coddington, P. D. Percolation, quantum tunnelling and the integer Hall effect. {\it J. Phys. C: Solid State Physics}, {\bf 21}, 2665, (1988).

\bibitem[C]{C} A. M. Childs, Universal Computation by Quantum Walk, {\it Phys. Rev. Lett.} {\bf 102}, 180501 (2009).

\bibitem[CFGW]{CFGW} C. Cedzich, J. Fillman, T. Geib, A.H. Werner, Singular continuous Cantor spectrum for magnetic quantum walks
{\it  Lett. Math. Phys.}, {\bf 110} 1141-1158, (2020) .

\bibitem[CFO]{CFO} C. Cedzich, J. Fillman, D. C. Ong, Almost Everything About the Unitary Almost Mathieu Operator
{\it  Commun.Math.Phys.} {\bf 403} , 745-794, (2023)

\bibitem[CFL+]{CFL+} C. Cedzich, J. Fillman, L. Li, D. C. Ong, Q. Zho, Exact Mobility Edges for Almost-Periodic CMV Matrices via Gauge Symmetries,
{\it Int. Math. Res. Not.} {\bf 8},  6906-6941, (2024).

\bibitem[CJWW]{CJWW}  C. Cedzich, A. Joye, A. H. Werner, R. F. Werner, Exponential Tail Estimates for Quantum Lattice Dynamics ", 
{\it arXiv}, arXiv:2408.02108v,  (2024)

\bibitem[DFT]{DFT} Delplace, P., Fruchart, M., Tauber, C. Phase rotation symmetry and the topology of oriented scattering networks. {\it Phys. Rev. B}, {\bf 95}, 205413, (2017).

\bibitem[CGG+]{CGG+} C Cedzich, T. Geib, F. A. Grünbaum, C. Stahl, L Velázquez, A. H. Werner, The topological classification of one-dimensional symmetric quantum walks, {\it Ann. H. Poincar\'e}, {\bf 19}, 325-383, (2018)

\bibitem[D]{D} P. Delplace ,Topological chiral modes in random networks,
\textit{SciPost Phys.} 8, 081 (2020) 

\bibitem[FH1]{FH1} E. Feldman, M. Hillery, Quantum walks on graphs and quantum scattering theory, in Coding Theory and Quantum Computing, Edt: D. Evans, J. Holt, C. Jones, K. Klintworth, B. Parshall, O. Pfister,
and H. Ward, \textit{ Contemp. Math.}, {\bf 381}, 71 (2005).

\bibitem[FH2]{FH2} E. Feldman, M. Hillery, Modifying quantum walks: a scattering theory approach, {\it J. Phys. A}, {\bf 40}, 11343 (2007)

\bibitem[GZ]{GZ} C. Godsil, H. Zhan, Discrete quantum walks on graphs and digraphs, London Math. Soc., LNS 484, 2023

\bibitem[G]{G} S. Gudder. Quantum Markov chains. {\it J. Math. Phys.}, {\bf 49} 072105, (2008).

\bibitem[HJ1]{HJ} E.Hamza, A.Joye : Spectral Transition for Random Quantum Walks on Trees,
{\it Commun. Math. Phys.},  {\bf 326}, 415-439, (2014).

\bibitem[HJ2]{HJ2} E.Hamza, A.Joye, Thermalization of Fermionic Quantum Walkers, {\it J.
Stat. Phys.}, {\bf 166}, 1365–1392, (2017).

\bibitem[HJS1]{HJS1} Hamza, E., Joye, A., Stolz, G.: Localization for Random Unitary Operators, {\it Lett. Math. Phys.}, {\bf 75}, 255-272, (2006).


\bibitem[HJS2]{HJS2} Hamza, E., Joye, A., Stolz, G.:  Dynamical Localization for Unitary Anderson Models" , 
{\it Mathematical Physics, Analysis and Geometry}, {\bf 12}, (2009), 381-444.

 
 
  \bibitem[HKSS2]{HKSS2} Y. Higuchi, N. Konno, I. Sato, E. Egawa, Spectral and asymptotic properties of Grover walks on crystal lattices, \textit{J. Fun. Ana.}, 267 (2014), 4197-4235.
 
 \bibitem[HS]{HS} Y. Higuchi, E. Segawa, Quantum walks induced by Dirichlet random walks on infinite trees, \textit{J. Phys. A: Math. Theor.} {\bf 51} (2018), 075303.
 
 \bibitem[HiSi]{HiSi} P. D. Hislop, I. M. Sigal, Introduction to Pectral Theory with Applications to Schr\"odinger Operators, Applied mathematical Sciences {\b 113}, Springer 1996.

\bibitem[HSS]{HSS} Y. Higuchi, E. Segawa, A. Suzuki, Spectral mapping theorem of an abstract quantum walk, \textit{Quantum Inf Process} {\bf 18}, (2019), 333.

\bibitem[H]{H}  A. S. Holevo, Statistical Structure of Quantum Theory, Lecture Notes in Physics Monographs, 67, 2001

\bibitem[J1]{J1} Joye, A.: Density of States and Thouless Formula
for Random Unitary Band Matrices, {\it Ann. Henri Poincar\'e} {\bf 5},
347--379, (2004).

\bibitem[J2]{J2} A. Joye, Fractional Moment Estimates for Random Unitary Band Matrices, {\it Lett. Math. Phys.} {\bf 72}, 51–64 (2005)


\bibitem[J3]{J3} A.Joye: Dynamical Localization for $d$-Dimensional Random Quantum Walks,
 {\it Quantum Inf. Proc.},  Special Issue: Quantum Walks, {\bf 11}, 1251-1269,  (2012).
 
 \bibitem[J4]{J4} A. Joye, Dynamical Localization of Random Quantum Walks on the Lattice. In
{\it XVIIth International Congress on Mathematical Physics, Aalborg, Denmark, 6-11
August 2012}, A. Jensen, Edt., World Scientific (2013) 486-494.



\bibitem[JMa]{JMa} A.Joye, L.Marin : "Spectral Properties of Quantum Walks on Rooted Binary Trees", {\it J. Stat. Phys.},  {\bf 155},  1249-1270, (2014).

\bibitem[JMe]{JM} A.Joye and M.Merkli: "Dynamical Localization of Quantum Walks in Random Environments",
{\it J. Stat. Phys.}, {\bf 140}, 1025-1053, (2010).

\bibitem[KFC+]{KFC+} M. Karski, L. F\"orster, J.M. Chioi, A. Streffen, W. Alt, D. Meschede, A. Widera, Quantum Walk in Position Space with Single Optically Trapped Atoms, {\it Science}, {\bf 325}, 174-177, (2009).
 
 \bibitem[Ka]{Ka} Kato, T., Perturbation Theory for Linear Operators (Springer-Verlag Berlin Heidelberg New York 1980).
 
 \bibitem[Ke]{Ke} J. Kempe, Quantum random walks: an introductory overview, {\it Contemp. Phys.}, {\bf 44}, 307–327, (2003).
 
 \bibitem[KM1]{KM1} K\"ummerer B., Maassen H., An ergodic theorem for quantum counting processes. {\it J. Phys. A} {\bf 36},
 2155 (2003).
 
 \bibitem[KM2]{KM2} K\"ummerer B., Maassen H., A pathwise ergodic theorem for quantum trajectories. {\it J. Phys. A} {\bf 37},
 11889–11896 (2004). 
 
 \bibitem[Kr]{Kr} K. Kraus, "States, Effects and Operations: Fundamental Notions of Quantum Theory", Edt. by A. B\"ohm, J. D. Dollard, W. H. Wootters, {\it Lecture Notes in Physics} {\bf 190}, Springer,1983.
 
 \bibitem[KS]{KS} Koralov, L. B., Sinai Y., G. Theory of Probability and Random Processes, Springer, 2007.
 
  \bibitem[Ko1]{Ko1} N. Konno, {\it J. Math. Soc. Japan},  {\bf 57}, 1179-1195, (2005).

 \bibitem[Ko2]{Kon} N. Konno, Quantum walks. In {\it Quantum potential theory}, 309–452. Springer, 2008.
 
 
 \bibitem[Kos]{Ko} Koshovets, I.A.: Unitary Analog of the Anderson
Model. Purely Point Spectrum, {\it Theoret. and Math. Phys.} {\bf
89}, 1249--1270 (1992).

\bibitem[KMOR]{KMOR} Krovi, H., Magniez, F., Ozols, M., Roland, J., Quantum walks can find a marked element on any graph, {\it Algorithmica}, {\bf 74}, 851–907 (2015).

\bibitem[LS]{LS} L.J. Landau, R. F. Streater,  On Birkhoff’s Theorem for Doubly Stochastic Completely Positive Maps of Matrix Algebras, {\it Linear Algebra Appl.},  {\bf 193}, 107-127 (1993).

 \bibitem[MS-B]{MS-B} L. Marin, H. Schulz-Baldes, Scattering zippers and their spectral theory. {\it J.
Spectral Theory}, {\bf 3}, 47–82, 2013.
 
 \bibitem[MW]{MW} C.B. Mendl, M.M. Wolf, Unital Quantum Channels - Convex Structure and Revival of Birkhoff's Theorem, {\it Comm. Math. Phys.}, {\bf 289}, 1057-1086, (2009)
 
 \bibitem[N]{N} Norris, J. R.: Markov chains. Cambridge University Press, 1997.
 
 \bibitem[P]{P} Portugal, R. Quantum Walks and Search Algorithms, Springer 2013. 


\bibitem[P-GWPR]{PGWPR} D. Perez-Garcia, M. Wolf, D. Petz, M.-B. Ruskai, Contractivity of positive and trace preserving maps under $Lp$ norms, \textit{J. Math. Phys.} 47, 083506 (2006)

\bibitem[QMS]{QMS} X. Qiang, S. Ma, H. Song, Review on Quantum Walk Computing: Theory, Implementation, and Application, {\it ArXiv} arXiv:2404.04178, (2024)

\bibitem[R]{R} R. Raqu\'epas, On fermionic walkers interacting with a correlated structured environment, {\it Lett. Math. Phys.}, {\bf 110}, 121–145, (2020).

\bibitem[RST1]{RST1} S. Richard, A. Suzuki, R. Tiedra de Aldecoa, Quantum walks with an anisotropic coin I: spectral theory,
 {\it Lett. Math. Phys.} {\bf 108}, 331-357, (2018)
 
\bibitem[RST2]{RST2} S. Richard, A. Suzuki, R. Tiedra de Aldecoa, Quantum walks with an anisotropic coin II: scattering theory,
{\it Lett. Math. Phys.} {\bf 109}, 61-88, (2019) 


\bibitem[RT]{RT} S. Richard, R. Tiedra de Aldecoa, Decay Estimates for Unitary Representations with Applications to Continuous and Discrete Time Models, {\it Ann. H. Poincar\'e}, {\bf 24}, 1-36, (2022) 

\bibitem[SS-B]{SS-B} Sadel, C., Schulz-Baldes, H. Topological boundary invariants for Floquet systems and quantum walks. {\it Math. Phys. Anal. Geom.} (2017) 20–22.

\bibitem[Sa]{Sa} M. Santha, Quantum walk based search algorithms. In {\it International Conference on Theory and Applications of Models of Computation, TAMC 2008}, 31–46. Springer, 2008.

\bibitem[Sc]{Sc} R. Schrader. Perron-Frobenius theory for positive maps on trace ideals. {\it Fields Inst. Commun.}, {\bf 30}, 107–182, 2006.

\bibitem[ST]{ST} J. Shapiro and C. Tauber, Strongly Disordered Floquet Topological Systems,  {\it Ann. H. Poincar\'e},  {\bf 20}, 1837-1875, (2019)

\bibitem[Si]{S} B. Simon, Operator Theory, A Comprehensive Course in Analysis, Part 4, (American Mathematical Society 2015).

\bibitem[Sz]{Sz} Szegedy, M., Quantum speed-up of Markov chain based algorithms. In: Proc. 45th An-
nual IEEE Symp. Found. Comput. Sci., FOCS ’04, pp. 32–41. IEEE Computer Society,
Washington, DC, USA (2004)

\bibitem[T]{T} R. Tiedra de Aldecoa, Spectral and scattering properties of quantum walks on homogenous trees of odd degree,
{\it Ann. Henri Poincar\'e},  {\bf 22} , 2563-2593, (2021)

\bibitem[TMT]{TMT} M. Tamura, T. Mukaiyama, K. Toyoda, Quantum walks of a phonon in trapped ions, {\it Phys. Rev. Lett.}, {\bf 124}, 200501, (2020).

\bibitem[V-A]{V-A} Venegas-Andraca, S. E., ‘‘Quantum walks: A comprehensive review,’’ {\it Quantum Information Processing}, {\bf 11}, 1015–1106 (2012).

\bibitem[WM]{WM} J. Wang, K. Manouchehri. Physical implementation of quantum walks. Springer, 2013.

\bibitem[Wa]{Wa} J. Watrous, The Theory of Quantum Information. Cambridge University Press, 2018. 

\bibitem[Wo]{Wo} W. Woess, Denumerable Markov Chains, European Mathematical Society, 2009.

\bibitem[ZKG+]{ZKG+} F. Z\"ahringer, G. Kirchmair, R. Gerritsma, E. Solano, R. Blatt, and C. Roos, Realization of a quantum walk with one and two trapped ions, {\it Phys. Rev. Lett.}, {\bf 104}, 100503, (2010).

\end{thebibliography}
\end{document}